\documentclass[12pt]{article}

\usepackage{amsmath, amssymb, amsthm, amsfonts}
\usepackage{titlesec}
\usepackage{hyperref}
\hypersetup{hidelinks}
\usepackage{chngcntr}
\usepackage{graphicx}
\usepackage{dsfont}
\usepackage{bm}
\usepackage{bbm}
\usepackage[mathscr]{euscript}
\usepackage{enumitem}
\usepackage{colortbl}
\usepackage{tabularx}
\usepackage{booktabs}
\usepackage{soul}
\usepackage{tikz}
\usetikzlibrary{positioning, arrows, bending}
\usetikzlibrary{arrows.meta}
\usepackage{pgfplots}
\pgfplotsset{compat=1.18}
\usepackage{verbatim}
\usepackage{graphbox}
\usepackage[title]{appendix}
\usepackage[english]{babel}
\usepackage[
backend=biber,
style=apa,
natbib
]{biblatex}
\usepackage{csquotes}
\usepackage{algorithm2e}[ruled]
\usepackage{refcount}

\usepackage[a4paper, margin=1in]{geometry}
\usepackage[nodisplayskipstretch]{setspace}
\usepackage{caption}
\usepackage{subcaption}
\usepackage{float}
\usepackage{array}
\newcolumntype{L}[1]{>{\raggedright\arraybackslash\small}p{#1}}

\DeclareMathOperator{\E}{\mathbb{E}}

\newcommand{\calC}{\mathcal{C}}

\newcommand{\calG}{\mathcal{G}}

\newcommand{\calP}{\mathcal{P}}
\newcommand{\calQ}{\mathcal{Q}}

\newcommand{\bA}{\bm A}

\newcommand{\bB}{\bm B}

\newcommand{\bC}{\bm C}

\newcommand{\bD}{\bm D}

\newcommand{\be}{\bm e}

\newcommand{\bI}{\bm I}

\newcommand{\bL}{\bm L}

\newcommand{\bM}{\bm M}

\newcommand{\bO}{\bm O}

\newcommand{\bQ}{\bm Q}

\newcommand{\bR}{\bm R}

\newcommand{\bS}{\bm S}

\newcommand{\bT}{\bm T}

\newcommand{\bu}{\bm u}

\newcommand{\bX}{\bm X}
\newcommand{\bx}{\bm x}

\newcommand{\by}{\bm y}

\newcommand{\bgamma}{\bm \gamma}

\newcommand{\bepsilon}{\bm \epsilon}
\newcommand{\bvarepsilon}{\bm \varepsilon}

\newcommand{\bPhi}{\bm \varPhi}
\newcommand{\bPsi}{\bm \varPsi}
\newcommand{\bOmega}{\bm \varOmega}

\newcommand{\R}{\mathbb{R}}

\theoremstyle{definition}
\newtheorem{assumption}{Assumption}
\newtheorem{theorem}{Theorem}

\newtheorem{lemma}{Lemma}

\newtheorem{definition}{Definition}

\newtheorem{example}{Example}
\newtheorem*{excont}{Example \continuation}
\newcommand{\continuation}{??}
\newenvironment{continueexample}[1]
 {\renewcommand{\continuation}{\ref{#1}}\excont[continued]}
 {\endexcont}

\definecolor{UMdblue}{RGB}{0,28,61}
\definecolor{UMlblue}{RGB}{0,162,219}
\definecolor{UMorangered}{RGB}{232,78,16}
\definecolor{UMorange}{RGB}{243,148,37} %
\definecolor{UMred}{RGB}{174,11,18} %
\definecolor{UMorangesyllabus}{RGB}{215,92,45}

\definecolor{red1}{RGB}{255,0,0}
\definecolor{red2}{RGB}{250,40,15}
\definecolor{red3}{RGB}{245,0,35}
\definecolor{red4}{RGB}{235,30,30}

\definecolor{blue1}{RGB}{0,0,255}
\definecolor{blue2}{RGB}{40,15,250}
\definecolor{blue3}{RGB}{0,35,245}
\definecolor{blue4}{RGB}{30,30,235}

\definecolor{green1}{RGB}{0,150,0}
\definecolor{green2}{RGB}{15,145,40}
\definecolor{green3}{RGB}{35,145,0}
\definecolor{green4}{RGB}{30,135,30}

\definecolor{yellow1}{RGB}{180,180,20}
\definecolor{yellow2}{RGB}{160,140,30}

\definecolor{ada}{RGB}{247,252,13}
\definecolor{nonada}{RGB}{54,38,134}

\definecolor{mycolor1}{RGB}{0, 114, 178} %
\definecolor{mycolor2}{RGB}{230, 159, 0} %
\definecolor{mycolor3}{RGB}{0, 158, 115} %
\definecolor{mycolor4}{RGB}{204, 121, 167} %
\definecolor{mycolor5}{RGB}{86, 180, 233} %
\definecolor{mycolor6}{RGB}{213, 94, 0} %
\definecolor{mycolor7}{RGB}{240, 228, 66} %

\newcommand{\papertitle}{Transmission Channel Analysis in Dynamic Models}
\newcommand{\starfootnote}[1]{\begingroup\renewcommand{\thefootnote}{\fnsymbol{footnote}}\footnote[1]{#1}\endgroup}

\begin{document}

\begin{center}
{\Large\textbf{\papertitle}}

\vspace{12pt}
\today

\vspace{18pt}
\textbf{Enrico Wegner\starfootnote{Corresponding author: \href{mailto:e.wegner@maastrichtuniversity.nl}{e.wegner@maastrichtuniversity.nl}}}\\
School of Business and Economics, Maastricht University, Maastricht, The Netherlands

\vspace{8pt}
\textbf{Lenard Lieb}\\
School of Business and Economics, Maastricht University, Maastricht, The Netherlands

\vspace{8pt}
\textbf{Stephan Smeekes}\\
School of Business and Economics, Maastricht University, Maastricht, The Netherlands

\vspace{8pt}
\textbf{Ines Wilms}\\
School of Business and Economics, Maastricht University, Maastricht, The Netherlands
\end{center}

\begingroup
\renewcommand{\abstractname}{\vspace{-\baselineskip}}
\begin{abstract}
\footnotesize
We propose a framework for analysing transmission channels in a large class of dynamic models. Our method, labelled Transmission Channel Analysis (TCA), decomposes the total effects captured by impulse response functions into the effects flowing through transmission channels, thereby providing a quantitative assessment of the strength of various well-defined channels. This decomposition requires no identification assumptions beyond those needed to identify the structural shock whose effects the researcher wants to decompose. We provide algorithms to compute the effect through any transmission channel and prove that, for structural vector autoregressive models, impulse response functions are sufficient statistics for doing so. Applications to policy shocks in a variety of popular macroeconomic models demonstrate the empirical relevance of TCA for policy evaluation.
\end{abstract}
\endgroup

\vspace{12pt}


\clearpage

\section{Introduction}

Impulse response functions (IRFs) measure the total dynamic causal effect of macroeconomic shocks, but they do not reveal the channels through which a shock reaches macroeconomic outcomes. We propose a formal framework for quantitatively analysing such transmission channels in a large class of dynamic models. The framework decomposes the total effect of a structural shock into disjoint dynamic partial effects, with the decomposition chosen to match the research question. Quantifying these effects requires identifying only a single shock, so the framework applies wherever standard structural impulse response analysis is feasible.

The literature on the dynamic effects of shocks is large \citep[see e.g.][]{RameyMacroeconomicShocks2016,kilianStructuralVectorAutoregressive2017,nakamuraIdentification2018}, but the formal analysis of transmission channels remains sparse, and existing approaches fall short in several aspects. Textbook-style partial-equilibrium reasoning \citep[see e.g.][]{galiMonetaryPolicyInflation2015, woodfordInterestPricesFoundations2004} argues about transmission qualitatively, presuming a sequential decision process in which partial-equilibrium equations encode how one agent reacts to another. In general equilibrium, however, agents use information about the full model and decide simultaneously, so these equations lose their meaning and qualitative reasoning can mislead \citep{dograParadoxesProblemsCausal2024a}. 
HANK-style direct-indirect decompositions \citep{auclertMonetaryPolicyRedistribution2019, kaplanMonetaryPolicyAccording2018, bilbiieNewKeynesianCross2020} give channels a 
mathematical meaning and quantify their contributions, but they too rest on partial-equilibrium reasoning. In addition, they rely on behavioural equations of a fully specified theoretical model, and can therefore not be applied to empirical models. 
Counterfactual analysis \citep{bernankeChapter21Financial1999, christianoNominalRigiditiesDynamic2005, iacovielloHousePricesBorrowing2005} defines channels by setting a structural parameter to zero or infinity and comparing equilibria, which captures all general equilibrium dynamics but contrasts two economies with different structures and so is uninformative about the transmission of a shock within a given equilibrium. Moreover, in empirical models, counterfactual analysis is restricted to specific policy counterfactuals \citep{mckayWhatCanTime2023}; beyond these, it too requires a fully specified theoretical model. None of these approaches provides a 
definition of a transmission channel that applies to both theoretical and empirical models, accounts for all general-equilibrium dynamics, and describes transmission within a given equilibrium.

Our first contribution is a general framework to formally define and quantify transmission channels. Unlike the methods discussed above, our framework does not require a fully specified theoretical model, although it applies to such models as well. Our channels also remain interpretable in general equilibrium. Moreover, because our framework holds the equilibrium fixed, it quantifies transmission within a given equilibrium, unlike counterfactual analysis. This lets us combine our framework with counterfactual experiments to investigate how transmission changes as the model's characteristics change.

Our second contribution formulates transmission channel analysis (TCA henceforth) as a decomposition of a shock's total effect, obtained through impulse response analysis, into the contributions of well-defined channels. We show that disjoint channels decompose this effect exactly, at any horizon, and that quantifying a channel requires identifying only the shock of interest, not the full system; reduced-form information otherwise suffices. When channels overlap or omit some transmission, the residual retains a clear interpretation, unlike the 'general equilibrium' residual left by partial-equilibrium reasoning and counterfactual analysis. 

Our third contribution is to develop TCA for both theoretical and empirical models. Our framework applies to all invertible structural vector autoregressive moving average (SVARMA) models, which include a large class of linearised dynamic stochastic general equilibrium (DSGE) models, so that, unlike the approaches above, it covers empirical and theoretical models alike. For structural vector autoregressive (SVAR) models, the impulse responses alone suffice to compute transmission effects. Since local projections \citep{jordaEstimationAndInference2005} estimate the same population impulse responses as SVARs \citep{plagborg-mollerLocalProjectionsVARs2021}, our method can also be applied using local projections.

Our fourth contribution is technical. To establish the results above, we connect the impulse-response-space representation of the model's dynamics to a directed acyclic graph (DAG), which lets us define a transmission channel as a collection of paths along the graph. We complement this graphical representation with one in terms of potential outcomes and show that the two are equivalent: a channel can be defined either as a collection of paths through a graph or as a specific potential outcome. Beyond their role in our analysis, both representations are of independent interest for the study of dynamic causal effects. 
For example, they can inform the identification of structural shocks, including schemes that restrict transmission channels.

Our fifth and final contribution is computational. To facilitate applications, we provide algorithms that compute the effects through transmission channels, implemented in MATLAB and Julia and documented extensively on a companion website.\footnote{The software and its documentation are available at \url{https://tca.enricowegner.com}.} Both implementations simplify the definition of channels and automatically select the appropriate algorithm for computing their effects. The researcher therefore only needs to provide a model, whether an SVARMA, a set of local projections, or a linearised DSGE model, and define the transmission channel. The software handles the rest.

We apply TCA to three workhorse settings that span the main approaches to identification and estimation. In an SVAR, we compare monetary transmission under the narrative identification of \citet{romerNewMeasureMonetary2004} and the high-frequency identification of \citet{gertlerMonetaryPolicySurprises2015}. The Romer–Romer shock transmits predominantly through the contemporaneous policy rate and the Gertler–Karadi shock through other channels, so the two cannot be regarded as measuring the same underlying shock. Using local projections, we decompose a government military-spending news shock \citep{rameyGovernmentSpendingMultipliers2018} into an implementation channel and a residual capturing anticipation, and find the implementation channel matters little on impact but grows at longer horizons, quantifying an effect that \citet{Ramey2011} could only infer from the shape of impulse responses. In a DSGE model with investment adjustment costs, we find that the consumption rather than the investment channel drives most of the inflation response to monetary policy, and, within a counterfactual experiment, that the investment channel becomes more disinflationary at early horizons as the depreciation rate rises, contrary to partial-equilibrium reasoning.

Next to the methods described above, the framework relates most closely to mediation analysis \citep[see, for example,][]{imaiIdentificationInferenceSensitivity2010, danielCausalMediationAnalysis2015, pearlCausalMediationFormula2012}, which likewise decomposes a total causal effect into effects operating through intermediate variables. The two differ in what a channel is. Mediation analysis requires a fully causal system, in which the links between intermediate variables are themselves causal, so that a channel is a chain of causal effects. Such links cannot be established in most general equilibrium (GE) macroeconomic models, where all variables adjust simultaneously to the structural shock rather than sequentially to one another \citep{dograParadoxesProblemsCausal2024a}. We show that total effects can nonetheless be decomposed into channels interpretable as chains of conditional predictions triggered by the shock, thus providing the mediation-analysis equivalent for general equilibrium models. Because a channel in a dynamic model may comprise many paths, whose number grows exponentially with the horizon, neither standard path analysis \citep{wrightCorrelationCausation1921, wrightMethodPathCoefficients1934} nor regression-based methods apply directly; we therefore provide algorithms that compute channel effects efficiently and consistently with the underlying equilibrium model.

Our framework is also related to the graphical causal literature \citep[see, for example,][]{pearlCausalityModelsReasoning2009} and to the recent macroeconomic literature on potential outcomes \citep{asheshrambachanWhenCommonTime2021, kolesarDynamicCausalEffects2024, cloyneStateDependentLocalProjections2023, angristSemiparametricEstimatesMonetary2018}. The graphical causal literature assumes that every edge of the DAG is causal, but no such fully causal DAG exists in general equilibrium models, for the reason given above. We show instead that the equilibrium system always admits a DAG in which edges between endogenous variables represent conditional predictions and edges from shocks represent conditional causal effects. A path through the graph is therefore a chain of conditional predictions triggered by the shock, and our transmission channels inherit this interpretation. 
Because we prove the graphical and potential-outcomes representations equivalent, our potential outcomes carry the same interpretation.

Finally, TCA complements existing decomposition methods such as forecast error variance decompositions (FEVD) and historical decompositions, but answers a different question. FEVD attributes forecast uncertainty to structural shocks at each horizon; historical decompositions attribute past episodes to specific shocks. Both ask which shocks matter. TCA asks how a given shock transmits through the equilibrium system to an outcome; it thus opens a new dimension of analysis. 

The remainder of the paper is structured as follows. Section~\ref{sec:sec2-illustrative-example} introduces TCA for a simple macroeconomic model and highlights the differences between TCA and counterfactual analysis. Section~\ref{sec:sec3-general-framework} provides the formal definition of TCA for invertible SVARMA models, and Section~\ref{sec:sec4} presents the main theoretical results. Section~\ref{sec:sec5-applications} applies TCA to an SVAR, local projections, and a DSGE model, and Section~\ref{sec:conclusion} concludes. Appendix~\ref{appendix:proofs} contains the proofs of the main results, and Appendices~\ref{appendix:proof-algebra} to \ref{appendix:counterexample-irf-sufficiency}, available in the online appendix, provide supplementary results.

\section{TCA in a Simple Macroeconomic Model}
\label{sec:sec2-illustrative-example}

We illustrate the idea behind TCA by tracing out the transmission channels of a contractionary monetary policy shock within a small-scale New Keynesian DSGE model. The model builds on the canonical three-equation model of \citet{galiMonetaryPolicyInflation2015}, extended to include capital. Households accumulate capital subject to adjustment costs, and firms rent it from households to use in production alongside labour. Appendix~\ref{appendix:dsge-model} provides details on the DSGE model. Throughout, we focus on the dynamics of three endogenous variables: consumption ($C_t$), investment ($I_t$), and inflation ($\pi_t$). These are driven by a monetary policy shock ($\varepsilon^M_t$), a preference shock ($\varepsilon^Z_t$) acting as a demand shock, and a capital quality shock ($\varepsilon^C_t$) acting as a supply shock. Conditional on these 
shocks, and once the system is reduced to consumption, investment, and inflation, the log-linearised model admits an SVARMA(3, 3) representation, which forms the basis for TCA.

The model's partial-equilibrium relations trace out two clear demand channels through which a contractionary monetary policy shock affects inflation: a consumption channel and an investment channel. In the consumption channel, the shock raises the real rate, which reduces households' consumption through intertemporal substitution, in turn lowering aggregate demand and therefore inflation. In the investment channel, the higher real rate lowers the discounted future income from capital goods, so that installed capital is worth less today. With investment adjustment costs, households, who own the capital, respond by cutting investment. This again reduces aggregate demand and therefore prices. However, general equilibrium effects often overturn such partial-equilibrium reasoning. This section therefore asks whether, and by how much, the general equilibrium consumption and investment channels explain the fall in inflation after a contractionary monetary policy shock. We first address this question from the viewpoint of TCA. We then contrast this view to counterfactual analysis common in macroeconomics.

\paragraph{Transmission Channel Analysis (TCA)}\hspace{-1.0em} traces how a structural shock propagates through the connections between endogenous variables, accounting for all general-\hspace{0pt}equilibrium effects. It decomposes the total effect into the contribution of each connection, that is, of each transmission channel. TCA thus works within a given equilibrium, yielding a decomposition consistent with all model equations at once, unlike partial-equilibrium reasoning, which takes one equation at a time. As we will show, a channel is a chain of conditional predictions triggered by the shock. Because TCA operates on the links between endogenous variables, we call it \emph{endogenous-variable-driven}. 
In contrast, the counterfactual approach, 
which we call \emph{parameter-driven},
focuses on parameters inside the model rather than on the connections between endogenous variables.

Throughout this section we restrict attention to the impact horizon, since most of the conceptual challenges already arise there. We therefore decompose the contemporaneous response of inflation $\pi_0$ to a monetary policy shock $\varepsilon^M_0$ into a contribution associated with the response of consumption $C_0$ and a contribution associated with the response of investment $I_0$. To avoid clutter, we henceforth omit the time subscript. Section~\ref{sec:sec3-general-framework} extends TCA formally to the dynamic setting, and Section~\ref{sec:sec5-dsge} returns to the DSGE model of this section, decomposes the full dynamic response of inflation to the monetary policy shock, and investigates how the decomposition changes as the model's characteristics vary.

The contemporaneous response of the model economy and its SVARMA representation can be depicted as a graph, shown in Figure~\ref{fig:section2-graph-cyclic}, built from three ingredients: nodes, edges, and path coefficients, where, for expositional clarity, we suppress the path coefficients. The nodes represent either the endogenous variables (solid black circles) or the shocks (shaded blue squares). Edges are directed: one edge runs from each shock to exactly one endogenous variable, and further edges run between endogenous variables. Each edge carries a path coefficient encoding the contemporaneous relation between the two nodes it connects. The graph is hence just an alternative representation of the algebraic formulation of the dynamic general equilibrium, but one that makes the definition of transmission channels more intuitive. We therefore work with the graph throughout this section; Section~\ref{sec:sec3-general-framework} provides a formal treatment.

\begin{figure}
    \centering
    \begin{subfigure}[b]{0.3\textwidth}
        \centering
        \begin{tikzpicture}[scale=0.8, transform shape]
            \node[circle, draw, minimum size=1cm, line width=0.8pt] (C) at (0, 4) {$C$};
            \node[circle, draw, line width=0.8pt, minimum size=1cm, black] (I) at (0, 2) {$I$};
            \node[circle, draw, line width=0.8pt, minimum size=1cm, black] (pi) at (0, 0) {$\pi$};
    
            \node[rectangle, draw=blue!50, fill=blue!20, line width=0.5pt, minimum size=0.3cm] (emp) at (-2, 4) {\textcolor{black}{$\varepsilon^{M}$}};
            \node[rectangle, draw=blue!50, fill=blue!20, line width=0.5pt, minimum size=0.3cm] (ez) at (-2, 2) {\textcolor{black}{$\varepsilon^{Z}$}};
            \node[rectangle, draw=blue!50, fill=blue!20, line width=0.5pt, minimum size=0.3cm] (ecq) at (-2, 0) {\textcolor{black}{$\varepsilon^{C}$}};
    
            \draw[-latex, line width=0.8pt, black] (C) edge[in=45, out=-45] (I);
            \draw[-latex, line width=0.8pt, black] (C) edge[in=45, out=-45] (pi);
            \draw[-latex, line width=0.8pt, black] (I) edge[in=-135, out=135] (C);
            \draw[-latex, line width=0.8pt, black] (I) edge[in=45, out=-45] (pi);
            \draw[-latex, line width=0.8pt, black] (pi) edge[in=-135, out=135] (I);
            \draw[-latex, line width=0.8pt, black] (pi) edge[in=-135, out=135] (C);
    
            \draw[-latex, line width=0.8pt, black] (ez) -- (I); 
            \draw[-latex, line width=0.8pt, black] (emp) -- (C); 
            \draw[-latex, line width=0.8pt, black] (ecq) -- (pi); 
            
        \end{tikzpicture}
        \caption{}
        \label{fig:section2-graph-cyclic}
    \end{subfigure}%
    \begin{subfigure}[b]{0.3\textwidth}
        \centering
        \begin{tikzpicture}[scale=0.8, transform shape]
            \node[circle, draw, minimum size=1cm, line width=0.8pt] (C) at (0, 4) {$C$};
            \node[circle, draw, line width=0.8pt, minimum size=1cm, black] (I) at (0, 2) {$I$};
            \node[circle, draw, line width=0.8pt, minimum size=1cm, black] (pi) at (0, 0) {$\pi$};
    
            \node[rectangle, draw=blue!50, fill=blue!20, line width=0.5pt, minimum size=0.3cm] (emp) at (-2, 4) {\textcolor{black}{$\varepsilon^{M}$}};
            \node[rectangle, draw=blue!50, fill=blue!20, line width=0.5pt, minimum size=0.3cm] (ez) at (-2, 2) {\textcolor{black}{$\varepsilon^{Z}$}};
            \node[rectangle, draw=blue!50, fill=blue!20, line width=0.5pt, minimum size=0.3cm] (ecq) at (-2, 0) {\textcolor{black}{$\varepsilon^{C}$}};
            
            \draw[-latex, line width=0.8pt, black] (C) -- (I); 
            \draw[-latex, line width=0.8pt, black] (C) edge[in=45, out=-45] (pi); 
            \draw[-latex, line width=0.8pt, black] (I) -- (pi); 
    
            \draw[-latex, line width=0.8pt, black] (ez) -- (I); 
            \draw[-latex, line width=0.8pt, black] (ez) -- (C); 
            \draw[-latex, line width=0.8pt, black] (ez) -- (pi); 
            \draw[-latex, line width=0.8pt, black] (ecq) -- (pi); 
            \draw[-latex, line width=0.8pt, black] (ecq) -- (C); 
            \draw[-latex, line width=0.8pt, black] (ecq) -- (I); 
            \draw[-latex, line width=0.8pt, black] (emp) edge (C); 
            \draw[-latex, line width=0.8pt, black] (emp) edge (I); 
            \draw[-latex, line width=0.8pt, black] (emp) edge (pi);
            
        \end{tikzpicture}
        \caption{}
        \label{fig:section2-graph-DAG-nogray}
    \end{subfigure}
    \caption{Contemporaneous graph of the small-scale New Keynesian model. Panel~(a) is the native cyclic representation; panel~(b) is the equivalent DAG.  Solid circles are endogenous variables, shaded squares are shocks; $C$ is consumption, $I$ investment, $\pi$ inflation, $\varepsilon^M$ the monetary policy shock, $\varepsilon^Z$ the preference shock, and $\varepsilon^C$ the capital-quality shock.}
    \label{fig:section2-graphs}
\end{figure}

The induced graph is cyclic, however. Thus, any attempt to trace the effect of the monetary policy shock on inflation through the graph runs into infinitely many routes, distinguished only by the number of times the cycle is traversed. For example, both $\varepsilon^{M} \to C \to \pi$ and $\varepsilon^{M} \to C \to \pi \to C \to \pi$ represent routes from the monetary policy shock to inflation through consumption. Selecting a finite subset is arbitrary; admitting all infinitely many routes is the only consistent alternative, but doing so conceptually forces one variable to be treated as the first-mover, therefore collapsing the cycle.

To collapse the cycle, we exploit that equilibrium systems admit multiple algebraically equivalent representations. We choose the representation that orders variables sequentially---first-mover, second-mover, and so on---each conditional on those before it; the algebra is in Section~\ref{sec:sec3-general-framework}. The ordering collapses the cycle: the induced graph, shown in Figure \ref{fig:section2-graph-DAG-nogray}, is a directed acyclic graph (DAG) with nodes, edges, and path coefficients, in which each shock now connects to more than one endogenous variable. 

Each edge in the DAG still carries a path coefficient, but its interpretation differs by edge type. Edges from a shock to an endogenous variable are the \emph{causal} effects of the shock on the component of a variable not predicted by the variables preceding it; we refer to these as the conditional \textit{causal} effects of the shock on the variable.
For example, $\varepsilon^M \to C$ is the total causal effect of the monetary policy shock on consumption, since no other variable is ordered before consumption; $\varepsilon^M\to I$ is the causal effect of the monetary policy shock on investment, conditional on consumption, i.e. holding consumption fixed at its steady-state value; $\varepsilon^M\to \pi$ is the causal effect of the monetary policy shock on inflation, conditional on consumption and investment.

Edges between two endogenous variables, by contrast, are not causal but \textit{associative}: each is a conditional prediction, in the standard regression sense, of a unit change in the origin node. In our example, $C \to I$ is the expected change in investment due to a unit change in consumption; $C \to \pi$ is the expected change in inflation due to a unit change in consumption, conditional on investment; $I \to \pi$ is the expected change in inflation due to a unit change in investment, conditional on consumption. The predictive interpretation is consistent with the more general point that causality between endogenous variables cannot be established in general equilibrium models \citep[see, for example,][]{dograParadoxesProblemsCausal2024a}. With the cycles removed and each edge now having a clear, distinguishable interpretation, transmission channels can be defined coherently on this DAG.

\begin{figure}
    \centering
    \begin{subfigure}[b]{0.2\textwidth}
        \centering
        \begin{tikzpicture}[scale=0.8, transform shape]
            \node[circle, draw, minimum size=1cm, line width=0.8pt] (C) at (0, 4) {$C$};
            \node[circle, draw, line width=0.8pt, minimum size=1cm, black] (I) at (0, 2) {$I$};
            \node[circle, draw, line width=0.8pt, minimum size=1cm, black] (pi) at (0, 0) {$\pi$};
    
            \node[rectangle, draw=blue!50, fill=blue!20, line width=0.5pt, minimum size=0.3cm] (emp) at (-2, 4) {\textcolor{black}{$\varepsilon^M$}};
            \node[rectangle, draw=gray!10, fill=gray!05, line width=0.5pt, minimum size=0.3cm] (ez) at (-2, 2) {\textcolor{gray!30}{$\varepsilon^{Z}$}};
            \node[rectangle, draw=gray!10, fill=gray!05, line width=0.5pt, minimum size=0.3cm] (ecq) at (-2, 0) {\textcolor{gray!30}{$\varepsilon^{C}$}};
            
            \draw[-latex, line width=0.8pt, black] (C) edge node[midway, fill=white, inner sep=1pt, yshift=3pt]{\footnotesize$\beta_2$} (I); 
            \draw[-latex, line width=0.8pt, black] (C) edge[in=45, out=-45] node[midway, fill=white, inner sep=1pt]{\footnotesize$\beta_1$} (pi); 
            \draw[-latex, line width=0.8pt, black] (I) edge node[midway, fill=white, inner sep=1pt, yshift=3pt]{\footnotesize$\beta_3$} (pi); 
    
            \draw[-latex, line width=0.8pt, lightgray!50] (ez) -- (I); 
            \draw[-latex, line width=0.8pt, lightgray!50] (ez) -- (C); 
            \draw[-latex, line width=0.8pt, lightgray!50] (ez) -- (pi); 
            \draw[-latex, line width=0.8pt, lightgray!50] (ecq) -- (pi); 
            \draw[-latex, line width=0.8pt, lightgray!50] (ecq) -- (C); 
            \draw[-latex, line width=0.8pt, lightgray!50] (ecq) -- (I); 
            \draw[-latex, line width=0.8pt, black] (emp) edge node[midway, fill=white, inner sep=1pt]{\footnotesize$\alpha_1$} (C); 
            \draw[-latex, line width=0.8pt, black] (emp) edge node[midway, fill=white, inner sep=1pt]{\footnotesize$\alpha_2$} (I); 
            \draw[-latex, line width=0.8pt, black] (emp) edge node[midway, fill=white, inner sep=1pt]{\footnotesize$\alpha_3$} (pi);

        \end{tikzpicture}
        \caption{}
        \label{fig:section2-graph-DAG-all}
    \end{subfigure}%
    \hfill
    \begin{subfigure}[b]{0.2\textwidth}
        \centering
        \begin{tikzpicture}[scale=0.8, transform shape]
            \node[circle, draw, minimum size=1cm, line width=0.8pt] (C) at (0, 4) {$C$};
            \node[circle, draw, line width=0.8pt, minimum size=1cm, lightgray!50] (I) at (0, 2) {$I$};
            \node[circle, draw, line width=0.8pt, minimum size=1cm, black] (pi) at (0, 0) {$\pi$};
    
            \node[rectangle, draw=blue!50, fill=blue!20, line width=0.5pt, minimum size=0.3cm] (emp) at (-2, 4) {\textcolor{black}{$\varepsilon^M$}};
            \node[rectangle, draw=gray!10, fill=gray!05, line width=0.5pt, minimum size=0.3cm] (ez) at (-2, 2) {\textcolor{gray!30}{$\varepsilon^{Z}$}};
            \node[rectangle, draw=gray!10, fill=gray!05, line width=0.5pt, minimum size=0.3cm] (ecq) at (-2, 0) {\textcolor{gray!30}{$\varepsilon^{C}$}};

            \draw[-latex, line width=0.8pt, lightgray!50] (C) -- (I); 
            \draw[-latex, line width=0.8pt, black] (C) edge[in=45, out=-45] node[midway, fill=white, inner sep=1pt]{\footnotesize$\beta_1$}(pi); 
            \draw[-latex, line width=0.8pt, lightgray!50] (I) -- (pi); 
    
            \draw[-latex, line width=0.8pt, lightgray!50] (ez) -- (I); 
            \draw[-latex, line width=0.8pt, lightgray!50] (ez) -- (C); 
            \draw[-latex, line width=0.8pt, lightgray!50] (ez) -- (pi); 
            \draw[-latex, line width=0.8pt, black] (emp) edge node[midway, fill=white, inner sep=1pt]{\footnotesize$\alpha_1$} (C); 
            \draw[-latex, line width=0.8pt, lightgray!50] (emp) -- (I); 
            \draw[-latex, line width=0.8pt, lightgray!50] (emp) -- (pi); 
            \draw[-latex, line width=0.8pt, lightgray!50] (ecq) -- (pi); 
            \draw[-latex, line width=0.8pt, lightgray!50] (ecq) -- (C); 
            \draw[-latex, line width=0.8pt, lightgray!50] (ecq) -- (I); 
            
        \end{tikzpicture}
        \caption{}
        \label{fig:section2-graph-DAG-consumption}
    \end{subfigure}%
    \hfill
    \begin{subfigure}[b]{0.2\textwidth}
        \centering
        \begin{tikzpicture}[scale=0.8, transform shape]
            \node[circle, draw, minimum size=1cm, line width=0.8pt] (C) at (0, 4) {$C$};
            \node[circle, draw, line width=0.8pt, minimum size=1cm, black] (I) at (0, 2) {$I$};
            \node[circle, draw, line width=0.8pt, minimum size=1cm, black] (pi) at (0, 0) {$\pi$};
    
            \node[rectangle, draw=blue!50, fill=blue!20, line width=0.5pt, minimum size=0.3cm] (emp) at (-2, 4) {\textcolor{black}{$\varepsilon^M$}};
            \node[rectangle, draw=gray!10, fill=gray!05, line width=0.5pt, minimum size=0.3cm] (ez) at (-2, 2) {\textcolor{gray!30}{$\varepsilon^{Z}$}};
            \node[rectangle, draw=gray!10, fill=gray!05, line width=0.5pt, minimum size=0.3cm] (ecq) at (-2, 0) {\textcolor{gray!30}{$\varepsilon^{C}$}};
            
            \draw[-latex, line width=0.8pt, black] (C) edge node[midway, fill=white, inner sep=1pt, yshift=3pt]{\footnotesize$\beta_2$} (I); 
            \draw[-latex, line width=0.8pt, lightgray!50] (C) edge[in=45, out=-45] (pi); 
            \draw[-latex, line width=0.8pt, black] (I) edge node[midway, fill=white, inner sep=1pt, yshift=3pt]{\footnotesize$\beta_3$} (pi); 
    
            \draw[-latex, line width=0.8pt, lightgray!50] (ez) -- (I); 
            \draw[-latex, line width=0.8pt, lightgray!50] (ez) -- (C); 
            \draw[-latex, line width=0.8pt, lightgray!50] (ez) -- (pi); 
            \draw[-latex, line width=0.8pt, black] (emp) edge node[midway, fill=white, inner sep=1pt]{\footnotesize$\alpha_1$} (C); 
            \draw[-latex, line width=0.8pt, black] (emp) edge node[midway, fill=white, inner sep=1pt]{\footnotesize$\alpha_2$} (I); 
            \draw[-latex, line width=0.8pt, lightgray!50] (emp) -- (pi); 
            \draw[-latex, line width=0.8pt, lightgray!50] (ecq) -- (pi); 
            \draw[-latex, line width=0.8pt, lightgray!50] (ecq) -- (C); 
            \draw[-latex, line width=0.8pt, lightgray!50] (ecq) -- (I); 
            
        \end{tikzpicture}
        \caption{}
        \label{fig:section2-graph-DAG-investment}
    \end{subfigure}%
    \hfill
    \begin{subfigure}[b]{0.2\textwidth}
        \centering
        \begin{tikzpicture}[scale=0.8, transform shape]
            \node[circle, draw, minimum size=1cm, line width=0.8pt] (C) at (0, 4) {$C$};
            \node[circle, draw, line width=0.8pt, minimum size=1cm, black] (I) at (0, 2) {$I$};
            \node[circle, draw, line width=0.8pt, minimum size=1cm, black] (pi) at (0, 0) {$\pi$};
    
            \node[rectangle, draw=blue!50, fill=blue!20, line width=0.5pt, minimum size=0.3cm] (emp) at (-2, 4) {\textcolor{black}{$\varepsilon^M$}};
            \node[rectangle, draw=gray!10, fill=gray!05, line width=0.5pt, minimum size=0.3cm] (ez) at (-2, 2) {\textcolor{gray!30}{$\varepsilon^{Z}$}};
            \node[rectangle, draw=gray!10, fill=gray!05, line width=0.5pt, minimum size=0.3cm] (ecq) at (-2, 0) {\textcolor{gray!30}{$\varepsilon^{C}$}};

            \draw[-latex, line width=0.8pt, black] (C) edge node[midway, fill=white, inner sep=1pt, yshift=3pt]{\footnotesize$\beta_2$} (I); 
            \draw[-latex, line width=0.8pt, black] (C) edge[in=45, out=-45] node[midway, fill=white, inner sep=1pt]{\footnotesize$\beta_1$}(pi); 
            \draw[-latex, line width=0.8pt, black] (I) edge node[midway, fill=white, inner sep=1pt, yshift=3pt]{\footnotesize$\beta_3$} (pi); 
    
            \draw[-latex, line width=0.8pt, lightgray!50] (ez) -- (I); 
            \draw[-latex, line width=0.8pt, lightgray!50] (ez) -- (C); 
            \draw[-latex, line width=0.8pt, lightgray!50] (ez) -- (pi); 
            \draw[-latex, line width=0.8pt, black] (emp) edge node[midway, fill=white, inner sep=1pt]{\footnotesize$\alpha_1$} (C); 
            \draw[-latex, line width=0.8pt, lightgray!50] (emp) -- (I); 
            \draw[-latex, line width=0.8pt, lightgray!50] (emp) -- (pi); 
            \draw[-latex, line width=0.8pt, lightgray!50] (ecq) -- (pi); 
            \draw[-latex, line width=0.8pt, lightgray!50] (ecq) -- (C); 
            \draw[-latex, line width=0.8pt, lightgray!50] (ecq) -- (I); 
            
        \end{tikzpicture}
        \caption{}
        \label{fig:section2-graph-DAG-consumption-alt}
    \end{subfigure}%
    \caption{The DAG of panel~(b) in Figure~\ref{fig:section2-graphs} with path coefficients $\alpha_i$ on shock-to-variable edges and $\beta_j$ on variable-to-variable edges (panel~a). Panel~(b) highlights the consumption channel; panel~(c) highlights the investment channel; panel~(d) highlights the alternative consumption channel.}
    \label{fig:section2-graph-DAG}
\end{figure}

\paragraph{Transmission channels}\hspace{-1.0em} are defined as a collection of paths connecting a shock to an outcome variable in the DAG induced by the reparameterised system. The \emph{transmission effect} is then the sum, across paths in this channel, of the product of the path coefficients along each path. Each path leaving a shock begins with a conditional causal effect on an endogenous variable and continues through conditional predictions between endogenous variables, so a path reads as a chain of conditional predictions triggered by the shock. Since stating each path in this way quickly becomes cumbersome, we will often just say that the effect is transmitted onwards, equivalently that one variable transmits the effect to the next variable. We now illustrate these definitions with the consumption and investment channels of our running example.

Intuitively, the consumption channel should trace the effect of the monetary policy shock on inflation through consumption; the investment channel should trace the effect through investment. One way to achieve this is to define the consumption and investment channels as: 
\begin{align}
    \textbf{Consumption Channel:}\quad & \{\varepsilon^{M} \to C \to \pi\} \label{eq:sec2-cons-channel-definition} \\
    \textbf{Investment Channel:}\quad & \{\varepsilon^{M} \to I \to \pi,\; \varepsilon^{M}\to C \to I \to \pi\}. \label{eq:sec2-invest-channel-definition}
\end{align}
So defined, the consumption channel \eqref{eq:sec2-cons-channel-definition}, also shown in Figure~\ref{fig:section2-graph-DAG-consumption}, consists of a single path and is interpreted as the expected change in inflation predicted by the total causal effect of the monetary policy shock on consumption, conditional on investment. The investment channel \eqref{eq:sec2-invest-channel-definition}, also shown in Figure~\ref{fig:section2-graph-DAG-investment}, is interpreted analogously as the expected change in inflation predicted by the total causal effect of the monetary policy shock on investment, conditional on consumption.

To state the transmission effects of the two channels, we introduce notation for the path coefficients: $\alpha_i$ for $i=1,2,3$ denotes coefficients for shock-to-variable edges and $\beta_j$ for $j=1,2,3$ denotes coefficients for variable-to-variable edges, all shown in Figure~\ref{fig:section2-graph-DAG}. The consumption channel's transmission effect is then $\alpha_1\beta_1$. The investment channel's transmission effect is $\alpha_2\beta_3 + \alpha_1\beta_2\beta_3$. The two channels are non-overlapping but do not exhaust all paths from shock to inflation: the direct path $\varepsilon^{M} \to \pi$, with transmission effect $\alpha_3$, is excluded from both. Thus, although both channels can be quantified, Section~\ref{sec:sec4} shows that their effects only add to the total effect if this direct effect is zero, i.e.\ when $\alpha_3=0$. If this is not the case, the decomposition will contain a remainder term, which remains interpretable as the direct causal effect of the monetary policy shock on inflation.

Note that the path $\varepsilon^{M} \to C \to I \to \pi$ assigned to the investment channel \eqref{eq:sec2-invest-channel-definition} above also passes through consumption. It could thus just as well be part of the consumption channel, leading to an alternative definition of the consumption channel as 
\begin{align}
    \textbf{Consumption Channel II:}\quad & \{\varepsilon^{M} \to C \to \pi, \; \varepsilon^{M}\to C \to I \to \pi\},
    \label{eq:sec2-cons-channel-2-definition}
\end{align}
highlighted in Figure \ref{fig:section2-graph-DAG-consumption-alt}, which is interpreted as the expected change in inflation predicted by the total causal effect of the monetary policy shock on consumption---this time without conditioning on investment. The transmission effect is $\alpha_1\beta_1 + \alpha_1\beta_2\beta_3$. This alternative definition has two consequences. First, the two channels \eqref{eq:sec2-cons-channel-2-definition} and \eqref{eq:sec2-invest-channel-definition} are now overlapping: they share the path $\varepsilon^{M}\to C \to I \to \pi$. As Section~\ref{sec:sec4} shows, the two channels will therefore not generally decompose the total effect even if the direct path is null, even when $\alpha_3=0$. Second, the interpretation of the two channels is asymmetric, with the consumption channel \eqref{eq:sec2-cons-channel-2-definition} being not conditional on investment while the investment channel \eqref{eq:sec2-invest-channel-definition} remains conditional on consumption.

\paragraph{Alternative orderings}\hspace{-1.0em} could be considered. Until now we have worked with the variable ordering $C, I, \pi$. This enabled the definition of the alternative consumption channel \eqref{eq:sec2-cons-channel-2-definition}, whose interpretation was not conditional on investment. A different variable ordering, however, leads to a different set of transmission channels that can be defined and quantified. For example, under the ordering $I, C, \pi$, an investment channel can be defined whose interpretation is not conditional on consumption, while no consumption channel exists whose interpretation is not conditional on investment. The choice of ordering thus determines the possible set of transmission channels and their respective interpretations. We encode this ordering in a permutation matrix---the \emph{transmission matrix}---which selects among the alternative DAGs. Since no true causal ordering exists among the endogenous variables, the transmission matrix is not determined by the model. Instead, it is part of the channel definition and thus a choice of the researcher.

Figure~\ref{fig:section2-all-DAGs} shows all six feasible DAGs and their corresponding transmission matrices. Of the six, only two are sensible for the joint analysis of a consumption and an investment channel: those that order both consumption and investment before inflation, i.e.\ panels (a) and (c). The other four place inflation too early in the ordering, either before both endogenous variables, in which case neither channel can be defined, or before one of them, in which case only one can be defined. The research question therefore narrows the set of admissible DAGs.

This still leaves two sensible orderings for our running example. Under the alternative DAG, shown in Figure~\ref{fig:section2-all-DAGs-I-C-pi}, the two channels can be defined as
\begin{align}
    \textbf{Consumption Channel III:} &\quad \{\varepsilon^M \to I \to C \to \pi, \; \varepsilon^M \to C \to \pi\} \label{eq:sec2-cons-channel-definition-alt-dag} \\
    \textbf{Investment Channel II:} &\quad \{\varepsilon^M \to I \to \pi\} \label{eq:sec2-invest-channel-definition-alt-dag}
\end{align}
The consumption channel \eqref{eq:sec2-cons-channel-definition-alt-dag} is then interpreted as the expected change in inflation predicted by the total causal effect of the monetary policy shock on consumption, conditional on investment. The investment channel \eqref{eq:sec2-invest-channel-definition-alt-dag} is interpreted as the expected change in inflation predicted by the total causal effect of the monetary policy shock on investment, conditional on consumption. These are the same interpretations as under the ordering $C, I, \pi$ with the channels defined as in \eqref{eq:sec2-cons-channel-definition} and \eqref{eq:sec2-invest-channel-definition}. Thus, although the transmission ordering matters in general, under the definitions given here the transmission effects coincide across the two DAGs. This is a special instance of the invariance results derived in Section~\ref{sec:sec4}. More generally, even when the research question does not reduce the number of sensible DAGs to a single one, these invariance results can be used to define channels whose transmission effects are invariant across the remaining choices. We now turn to counterfactual analysis.

\begin{figure}
    \centering
    \begin{subfigure}[b]{0.12\textwidth}
    \centering
    \begin{tikzpicture}[scale=0.7, transform shape]
        \node[circle, draw, line width=1pt, minimum size=0.8cm] (C) at (0, 4) {$C$};
        \node[circle, draw, line width=1pt, minimum size=0.8cm] (I) at (0, 2) {$I$};
        \node[circle, draw, line width=1pt, minimum size=0.8cm] (pi) at (0, 0) {$\pi$};
        
        \draw[-latex, line width=1pt] (C) -- (I);
        \draw[-latex, line width=1pt] (I) -- (pi); 
        \draw[-latex, line width=1pt] (C) edge [in=45, out=-45] (pi);

        \node[rectangle, draw=blue!50, fill=blue!20, line width=0.5pt, minimum size=0.3cm] (emp) at (-1.5, 4) {$\varepsilon^M$};
        \node[rectangle, draw=gray!10, fill=gray!05, line width=0.5pt, minimum size=0.3cm] (ez) at (-1.5, 2) {\textcolor{gray!30}{$\varepsilon^Z$}};
        \node[rectangle, draw=gray!10, fill=gray!05, line width=0.5pt, minimum size=0.3cm] (ecq) at (-1.5, 0) {\textcolor{gray!30}{$\varepsilon^C$}};

        \draw[-latex, line width=1pt, lightgray!50] (ez) -- (C); 
        \draw[-latex, line width=1pt, lightgray!50] (ez) -- (I); 
        \draw[-latex, line width=1pt, lightgray!50] (ez) -- (pi); 
        
        \draw[-latex, line width=1pt, lightgray!50] (ecq) -- (C); 
        \draw[-latex, line width=1pt, lightgray!50] (ecq) -- (I); 
        \draw[-latex, line width=1pt, lightgray!50] (ecq) -- (pi); 
        
        \draw[-latex, line width=1pt] (emp) -- (C); 
        \draw[-latex, line width=1pt] (emp) -- (I); 
        \draw[-latex, line width=1pt] (emp) -- (pi); 
        
        \node (matrix) at (-0.5, -1.5) {\footnotesize$\bT = \begin{bmatrix}
        1 & 0 & 0 \\ 
        0 & 1 & 0 \\ 
        0 & 0 & 1
        \end{bmatrix}$};
    \end{tikzpicture}
    \caption{}
    \label{fig:section2-all-DAGs-C-I-pi}
    \end{subfigure}
    \hfill
    \begin{subfigure}[b]{0.12\textwidth}
    \centering
    \begin{tikzpicture}[scale=0.7, transform shape]
        \node[circle, draw, line width=1pt, minimum size=0.8cm] (C) at (0, 4) {$C$};
        \node[circle, draw, line width=1pt, minimum size=0.8cm] (I) at (0, 2) {$\pi$};
        \node[circle, draw, line width=1pt, minimum size=0.8cm] (pi) at (0, 0) {$I$};
        
        \draw[-latex, line width=1pt] (C) -- (I);
        \draw[-latex, line width=1pt] (I) -- (pi); 
        \draw[-latex, line width=1pt] (C) edge [in=45, out=-45] (pi);

        \node[rectangle, draw=blue!50, fill=blue!20, line width=0.5pt, minimum size=0.3cm] (emp) at (-1.5, 4) {$\varepsilon^M$};
        \node[rectangle, draw=gray!10, fill=gray!05, line width=0.5pt, minimum size=0.3cm] (ez) at (-1.5, 2) {\textcolor{gray!30}{$\varepsilon^Z$}};
        \node[rectangle, draw=gray!10, fill=gray!05, line width=0.5pt, minimum size=0.3cm] (ecq) at (-1.5, 0) {\textcolor{gray!30}{$\varepsilon^C$}};

        \draw[-latex, line width=1pt, lightgray!50] (ez) -- (C); 
        \draw[-latex, line width=1pt, lightgray!50] (ez) -- (I); 
        \draw[-latex, line width=1pt, lightgray!50] (ez) -- (pi); 
        
        \draw[-latex, line width=1pt, lightgray!50] (ecq) -- (C); 
        \draw[-latex, line width=1pt, lightgray!50] (ecq) -- (I); 
        \draw[-latex, line width=1pt, lightgray!50] (ecq) -- (pi); 
        
        \draw[-latex, line width=1pt] (emp) -- (C); 
        \draw[-latex, line width=1pt] (emp) -- (I); 
        \draw[-latex, line width=1pt] (emp) -- (pi); 
        
        \node (matrix) at (-0.5, -1.5) {\footnotesize$\bT = \begin{bmatrix}
        1 & 0 & 0 \\ 
        0 & 0 & 1 \\ 
        0 & 1 & 0
        \end{bmatrix}$};
    \end{tikzpicture}
    \caption{}
    \label{fig:section2-all-DAGs-C-pi-I}
    \end{subfigure}
    \hfill
    \begin{subfigure}[b]{0.12\textwidth}
    \centering
    \begin{tikzpicture}[scale=0.7, transform shape]
        \node[circle, draw, line width=1pt, minimum size=0.8cm] (C) at (0, 4) {$I$};
        \node[circle, draw, line width=1pt, minimum size=0.8cm] (I) at (0, 2) {$C$};
        \node[circle, draw, line width=1pt, minimum size=0.8cm] (pi) at (0, 0) {$\pi$};
        
        \draw[-latex, line width=1pt] (C) -- (I);
        \draw[-latex, line width=1pt] (I) -- (pi); 
        \draw[-latex, line width=1pt] (C) edge [in=45, out=-45] (pi);

        \node[rectangle, draw=blue!50, fill=blue!20, line width=0.5pt, minimum size=0.3cm] (emp) at (-1.5, 4) {$\varepsilon^M$};
        \node[rectangle, draw=gray!10, fill=gray!05, line width=0.5pt, minimum size=0.3cm] (ez) at (-1.5, 2) {\textcolor{gray!30}{$\varepsilon^Z$}};
        \node[rectangle, draw=gray!10, fill=gray!05, line width=0.5pt, minimum size=0.3cm] (ecq) at (-1.5, 0) {\textcolor{gray!30}{$\varepsilon^C$}};

        \draw[-latex, line width=1pt, lightgray!50] (ez) -- (C); 
        \draw[-latex, line width=1pt, lightgray!50] (ez) -- (I); 
        \draw[-latex, line width=1pt, lightgray!50] (ez) -- (pi); 
        
        \draw[-latex, line width=1pt, lightgray!50] (ecq) -- (C); 
        \draw[-latex, line width=1pt, lightgray!50] (ecq) -- (I); 
        \draw[-latex, line width=1pt, lightgray!50] (ecq) -- (pi); 
        
        \draw[-latex, line width=1pt] (emp) -- (C); 
        \draw[-latex, line width=1pt] (emp) -- (I); 
        \draw[-latex, line width=1pt] (emp) -- (pi); 
        
        \node (matrix) at (-0.5, -1.5) {\footnotesize$\bT = \begin{bmatrix}
        0 & 1 & 0 \\ 
        1 & 0 & 0 \\ 
        0 & 0 & 1
        \end{bmatrix}$};
    \end{tikzpicture}
    \caption{}
    \label{fig:section2-all-DAGs-I-C-pi}
    \end{subfigure}
    \hfill 
    \begin{subfigure}[b]{0.12\textwidth}
    \centering
    \begin{tikzpicture}[scale=0.7, transform shape]
        \node[circle, draw, line width=1pt, minimum size=0.8cm] (C) at (0, 4) {$I$};
        \node[circle, draw, line width=1pt, minimum size=0.8cm] (I) at (0, 2) {$\pi$};
        \node[circle, draw, line width=1pt, minimum size=0.8cm] (pi) at (0, 0) {$C$};
        
        \draw[-latex, line width=1pt] (C) -- (I);
        \draw[-latex, line width=1pt] (I) -- (pi); 
        \draw[-latex, line width=1pt] (C) edge [in=45, out=-45] (pi);

        \node[rectangle, draw=blue!50, fill=blue!20, line width=0.5pt, minimum size=0.3cm] (emp) at (-1.5, 4) {$\varepsilon^M$};
        \node[rectangle, draw=gray!10, fill=gray!05, line width=0.5pt, minimum size=0.3cm] (ez) at (-1.5, 2) {\textcolor{gray!30}{$\varepsilon^Z$}};
        \node[rectangle, draw=gray!10, fill=gray!05, line width=0.5pt, minimum size=0.3cm] (ecq) at (-1.5, 0) {\textcolor{gray!30}{$\varepsilon^C$}};

        \draw[-latex, line width=1pt, lightgray!50] (ez) -- (C); 
        \draw[-latex, line width=1pt, lightgray!50] (ez) -- (I); 
        \draw[-latex, line width=1pt, lightgray!50] (ez) -- (pi); 
        
        \draw[-latex, line width=1pt, lightgray!50] (ecq) -- (C); 
        \draw[-latex, line width=1pt, lightgray!50] (ecq) -- (I); 
        \draw[-latex, line width=1pt, lightgray!50] (ecq) -- (pi); 
        
        \draw[-latex, line width=1pt] (emp) -- (C); 
        \draw[-latex, line width=1pt] (emp) -- (I); 
        \draw[-latex, line width=1pt] (emp) -- (pi); 
        
        \node (matrix) at (-0.5, -1.5) {\footnotesize$\bT = \begin{bmatrix}
        0 & 1 & 0 \\ 
        0 & 0 & 1 \\ 
        1 & 0 & 0
        \end{bmatrix}$};
    \end{tikzpicture}
    \caption{}
    \label{fig:section2-all-DAGs-I-pi-C}
    \end{subfigure}
    \hfill 
    \begin{subfigure}[b]{0.12\textwidth}
    \centering
    \begin{tikzpicture}[scale=0.7, transform shape]
        \node[circle, draw, line width=1pt, minimum size=0.8cm] (C) at (0, 4) {$\pi$};
        \node[circle, draw, line width=1pt, minimum size=0.8cm] (I) at (0, 2) {$I$};
        \node[circle, draw, line width=1pt, minimum size=0.8cm] (pi) at (0, 0) {$C$};
        
        \draw[-latex, line width=1pt] (C) -- (I);
        \draw[-latex, line width=1pt] (I) -- (pi); 
        \draw[-latex, line width=1pt] (C) edge [in=45, out=-45] (pi);

        \node[rectangle, draw=blue!50, fill=blue!20, line width=0.5pt, minimum size=0.3cm] (emp) at (-1.5, 4) {$\varepsilon^M$};
        \node[rectangle, draw=gray!10, fill=gray!05, line width=0.5pt, minimum size=0.3cm] (ez) at (-1.5, 2) {\textcolor{gray!30}{$\varepsilon^Z$}};
        \node[rectangle, draw=gray!10, fill=gray!05, line width=0.5pt, minimum size=0.3cm] (ecq) at (-1.5, 0) {\textcolor{gray!30}{$\varepsilon^C$}};

        \draw[-latex, line width=1pt, lightgray!50] (ez) -- (C); 
        \draw[-latex, line width=1pt, lightgray!50] (ez) -- (I); 
        \draw[-latex, line width=1pt, lightgray!50] (ez) -- (pi); 
        
        \draw[-latex, line width=1pt, lightgray!50] (ecq) -- (C); 
        \draw[-latex, line width=1pt, lightgray!50] (ecq) -- (I); 
        \draw[-latex, line width=1pt, lightgray!50] (ecq) -- (pi); 
        
        \draw[-latex, line width=1pt] (emp) -- (C); 
        \draw[-latex, line width=1pt] (emp) -- (I); 
        \draw[-latex, line width=1pt] (emp) -- (pi); 
        
        \node (matrix) at (-0.5, -1.5) {\footnotesize$\bT = \begin{bmatrix}
        0 & 0 & 1 \\ 
        0 & 1 & 0 \\ 
        1 & 0 & 0
        \end{bmatrix}$};
    \end{tikzpicture}
    \caption{}
    \label{fig:section2-all-DAGs-pi-I-C}
    \end{subfigure}
    \hfill 
    \begin{subfigure}[b]{0.12\textwidth}
    \centering
    \begin{tikzpicture}[scale=0.7, transform shape]
        \node[circle, draw, line width=1pt, minimum size=0.8cm] (C) at (0, 4) {$\pi$};
        \node[circle, draw, line width=1pt, minimum size=0.8cm] (I) at (0, 2) {$C$};
        \node[circle, draw, line width=1pt, minimum size=0.8cm] (pi) at (0, 0) {$I$};
        
        \draw[-latex, line width=1pt] (C) -- (I);
        \draw[-latex, line width=1pt] (I) -- (pi); 
        \draw[-latex, line width=1pt] (C) edge [in=45, out=-45] (pi);

        \node[rectangle, draw=blue!50, fill=blue!20, line width=0.5pt, minimum size=0.3cm] (emp) at (-1.5, 4) {$\varepsilon^M$};
        \node[rectangle, draw=gray!10, fill=gray!05, line width=0.5pt, minimum size=0.3cm] (ez) at (-1.5, 2) {\textcolor{gray!30}{$\varepsilon^Z$}};
        \node[rectangle, draw=gray!10, fill=gray!05, line width=0.5pt, minimum size=0.3cm] (ecq) at (-1.5, 0) {\textcolor{gray!30}{$\varepsilon^C$}};

        \draw[-latex, line width=1pt, lightgray!50] (ez) -- (C); 
        \draw[-latex, line width=1pt, lightgray!50] (ez) -- (I); 
        \draw[-latex, line width=1pt, lightgray!50] (ez) -- (pi); 
        
        \draw[-latex, line width=1pt, lightgray!50] (ecq) -- (C); 
        \draw[-latex, line width=1pt, lightgray!50] (ecq) -- (I); 
        \draw[-latex, line width=1pt, lightgray!50] (ecq) -- (pi); 
        
        \draw[-latex, line width=1pt] (emp) -- (C); 
        \draw[-latex, line width=1pt] (emp) -- (I); 
        \draw[-latex, line width=1pt] (emp) -- (pi); 

        \node (matrix) at (-0.5, -1.5) {\footnotesize$\bT = \begin{bmatrix}
        0 & 0 & 1 \\ 
        1 & 0 & 0 \\ 
        0 & 1 & 0
        \end{bmatrix}$};
    \end{tikzpicture}
    \caption{}
    \label{fig:section2-all-DAGs-pi-C-I}
    \end{subfigure}
    \caption{Panels (a) through (f) show all possible permutations of the three variables $C$, $I$ and $\pi$ as encoded in the transmission matrix $\bT$. The associated DAG is shown above the matrix where gray nodes and edges are irrelevant for the transmission of a monetary policy shock.}
    \label{fig:section2-all-DAGs}
\end{figure}

\paragraph{Counterfactual analysis}\hspace{-1em} is a common alternative in macroeconomics to analyse transmission channels. A channel within this view is defined by modifying the structure of the model, that is, a behavioural equation or a deep parameter, in such a way that a specific endogenous variable does not respond to the shock. The difference in total effects, that is in impulse responses, across the two equilibria is then attributed to that channel. 

In our model, the counterfactual approach could study the consumption channel by sending the elasticity of intertemporal substitution to zero, equivalently by sending relative risk aversion $\sigma$ to infinity. This follows from the linearised Euler equation
\begin{equation*}
    C_t = E_t C_{t+1} - \frac{1}{\sigma}(R_t - E_t\pi_{t+1}),
\end{equation*}
where $R_t$ is the short-term interest rate: as $\sigma\to\infty$, consumption becomes unresponsive to any shock, including the monetary policy shock. The investment channel could be studied analogously, by sending investment adjustment costs $\phi_I$ to infinity. The investment equation
\begin{equation*}
    \frac{1}{\phi_I}q_t = (I_t - I_{t-1}) - \beta E_t(I_{t+1}-I_t),
\end{equation*}
where $q_t$ is Tobin's $q$, implies that as $\phi_I\to\infty$, investment becomes unresponsive to the monetary policy shock. Within the counterfactual paradigm, each channel is therefore defined as the difference between the impulse response of inflation in the initial equilibrium and in the counterfactual equilibrium: for the consumption channel, the counterfactual equilibrium is obtained by sending the elasticity of intertemporal substitution to zero; for the investment channel, the counterfactual equilibrium is obtained by sending investment adjustment costs to infinity.

In both cases, the counterfactual shuts down a behavioural mechanism by changing a deep parameter, thereby altering the equilibrium itself. As a counterfactual conditional, this admits a clean causal interpretation: the causal impulse is the change in behaviour induced by the parameter shift, and the difference in impulse responses across equilibria can be attributed to that change. This causality does not run from one endogenous variable to another; it is not, for instance, the response of consumption causing the response of output. It runs instead from a structural change of the model to the equilibrium outcome. Because the causal impulse is always a parameter change that modifies the equilibrium, we call this approach \emph{parameter-driven}. It cannot ask, within a given equilibrium, how the total effect of a shock is distributed across the responses of specific endogenous variables. That is the endogenous-variable-driven view TCA takes. The two are therefore complementary: because TCA works within a given equilibrium while counterfactuals change the equilibrium, TCA can be applied within a counterfactual experiment; Section~\ref{sec:sec5-dsge} shows how this can be done for the running example.

\section{The General Framework}
\label{sec:sec3-general-framework}

We now formalise the intuition developed in the previous section. Section \ref{sec:dynamic-models-total-effects} introduces our general dynamic model and its impulse-response representation which we use to compute total effects. Section \ref{sec:dynamic-graphs} then introduces the associated DAG which we use to decompose the total effect into effects through transmission channels, and formally defines transmission channels and transmission effects. 

\subsection{Dynamic Models and Total Effects}
\label{sec:dynamic-models-total-effects}
We assume $\by_t$ to be a $K$-dimensional stationary stochastic process with a structural Vector Autoregressive Moving Average (VARMA) representation given by
\begin{equation}
    \bA_0\by_t = \sum_{i=1}^{\ell}\bA_i\by_{t-i} + \sum_{j=1}^{q}\bPsi_{j}\bvarepsilon_{t-j} + \bvarepsilon_t ,
    \label{eq:general-model}
\end{equation}
where $\{\bA_i\}_{i=1}^{\ell}$ and $\{\bPsi_j\}_{j=1}^{q}$ are $K\times K$ coefficient matrices which are statistically identified using any common scheme such as the echelon form (e.g., \citealp{poskitt1992}), $\bA_0$ is a contemporaneous coefficient matrix assumed to be (partially) identified using some economic identification scheme, and $\bvarepsilon_t$ is a $K\times 1$ vector of white noise.  

The structural VARMA model \eqref{eq:general-model} encompasses a wide range of popular dynamic models used in macroeconomic research. The SVAR is recovered by setting $\bPsi_j = \bf 0$ (for $j=1, \ldots, q$), while  many linearised DSGE models can be represented as structural VARMAs with cross-equation restrictions \citep{ravennaVectorAutoregressionsReduced2007, fernandez-villaverdeABCsDsUnderstanding2007}. 
We make two standard assumptions on $\bvarepsilon_t$ and $\bA_0$.  
\begin{assumption}
    The white noise vector $\bvarepsilon_t$ consists of $K$ structural shocks satisfying $\E[\bvarepsilon_{t}]=\bm 0_K$, $\E[\bvarepsilon_{t}\bvarepsilon_{t}']=\bI_K$ and $\E[\bvarepsilon_t\bvarepsilon_{t-r}']=\bO_K$ for all $r\geq 1$.
    \label{assump:general-model-structural-shocks}
\end{assumption}
\begin{assumption}
    $\bA_0$ is non-singular. 
    \label{assump:general-models-a0-nonsingular}
\end{assumption}
\noindent Assumption \ref{assump:general-model-structural-shocks} assures that impulse responses with respect to the white noise disturbances $\bvarepsilon_t$ have a causal interpretation. Assumption \ref{assump:general-models-a0-nonsingular}  assures that these causally interpretable impulse responses can be obtained from reduced-form IRFs. Section \ref{sec:sec4-main-results} uses this last implication to show how transmission effects can be obtained from a single identified shock. Hence, $\bA_0$ may be only partially identified. 

While all equivalent equilibrium representations can be used to analyse total effects, decomposing these effects into effects through transmission channels requires a precise definition of which variables depend on which other variables in equilibrium. TCA thus involves choosing the ordering of variables using the transmission matrix $\bT$. 
\begin{definition}
    The \emph{transmission matrix} $\bT$ denotes the $K\times K$ permutation matrix defining the ceteris paribus dependencies of the $K$ variables in the chosen equilibrium representation; it is fixed and specified by the researcher. 
    \label{def:transmission-matrix}
\end{definition}
\noindent Having a fixed and completely specified transmission matrix assures a unique DAG representation of the equilibrium. 
However, Section \ref{sec:sec4-main-results} shows that, under certain conditions when two distinct transmission matrices lead to the same transmission effect, the transmission matrix only needs to be partially specified, so only a partial variable ordering is required in such case. 

Under Assumptions \ref{assump:general-model-structural-shocks}-\ref{assump:general-models-a0-nonsingular} and a fixed $\bT$, 
the QL-decomposition of $\bA_0^*=\bA_0\bT'$ is unique,\footnote{We normalise the diagonal of $\bL$ to be positive, ensuring uniqueness.} such that 
\eqref{eq:general-model} can be written in unique lower-triangular form 
\begin{equation}
    \bL\by^*_t = \sum_{i=1}^\ell \bQ'\bA_i^*\by^*_{t-i}  + \sum_{j=1}^q \bQ'\bPsi_j\bvarepsilon_{t-j} + \bQ'\bvarepsilon_t,
    \label{eq:general-model-ql}
\end{equation}
where $\bL$ is a $K\times K$ lower-triangular matrix, $\by^*_t=\bT\by_t$, $\bA^*_i=\bA_i\bT'$ for $i=1,\dots,\ell$ and $\bQ$ is a $K\times K$ rotation matrix.  

Throughout the paper, we focus on transmission channels of the total effect of shock $\varepsilon_{i, t}$ to $y^*_{j, t+h}$ over the fixed finite impulse-response horizon $h$. Then the lower-triangular form \eqref{eq:general-model-ql} can be  re-written in systems form, capturing all dynamics of the shock $\varepsilon_{i, t}$ up to horizon $h$. Defining $\bx=(\by_t^{*'}, \dots, \by_{t+h}^{*'})'$ and $\bvarepsilon=(\bvarepsilon_t', \dots, \bvarepsilon_{t+h}')'$, the systems form (see Appendix \ref{appendix:rewrite} for the derivation) is given by
\begin{equation}
    \bx = \bB\bx + \bOmega\bvarepsilon,
    \label{eq:general-model-xbx}
\end{equation}
where $\bB$ is lower-triangular with zeros on the diagonal, $\bOmega$ is lower-block-triangular, and both are respectively given by\footnote{For expositional ease, we henceforth omit dimension subscripts, letting $\bI = \bI_K$ and $\bO = \bO_K$.} 
\begin{equation}
\resizebox{0.9\textwidth}{!}{$
    \begin{array}{ccc}
         \bB = \begin{bmatrix}
        \bI - \bD\bL & \bO & \dots & \bO \\
        \bD\bQ'\bA^*_1 & \bI-\bD\bL & \dots & \bO \\
        \vdots & \ddots & \ddots & \vdots \\
        \bD\bQ'\bA^*_h & \dots & \bD\bQ'\bA^*_1 & \bI - \bD\bL
    \end{bmatrix}, & \quad & 
    \bOmega = \begin{bmatrix}
        \bD\bQ' & \bO & \dots & \bO \\
        \bD\bQ'\bPsi_1 & \bD\bQ' & \dots & \bO \\
        \vdots & \ddots & \ddots & \vdots \\
        \bD\bQ'\bPsi_h & \dots & \bD\bQ'\bPsi_1 & \bD\bQ'
   \end{bmatrix},
    \end{array}$}
\label{eq:xbx-B-and-Omega}
\end{equation}
with $\bD=\text{diag}(\bL)^{-1}$, where $\text{diag}(\bX)$ is a diagonal matrix of the diagonal of $\bX$, $\bA^*_i = \bO$ for $i>\ell$ and $\bPsi_j = \bO$ for $j > q$.

The total effects can then be obtained from the impulse-response representation 
\begin{equation}
   \bx = \bPhi \bvarepsilon, \qquad \bPhi=(\bI-\bB)^{-1}\bOmega. 
\label{eq:phi-total-effect}
\end{equation}
Under our assumptions the impulse response matrix $\bPhi$ always exists. The aim of TCA is to decompose these total effects into effects through transmission channels. We next introduce the graph associated with system \eqref{eq:general-model-xbx}. 

\subsection{Dynamic Graphs and Transmission Channels}
\label{sec:dynamic-graphs}

To decompose the total effects given by $\bPhi$ in \eqref{eq:phi-total-effect} into effects through transmission channels, we introduce the graph $\calG(\bB, \bOmega)$ associated with system \eqref{eq:general-model-xbx}. The graph $\calG(\bB, \bOmega)$ is a DAG and has three ingredients: nodes, directed edges and path coefficients. A node exists for each shock $\varepsilon_{i,r}$ at time $r=t$ and for each variable $y_{i,r}$ at time points $r=t, \ldots, t+h$, where $h$ is the fixed horizon. Let  $\bx=(\by_t^{*'}, \ldots, \by_{t+h}^{*'})'$ and $\bvarepsilon=(\bvarepsilon_t', \ldots, \bvarepsilon_{t+h}')'$ as in equation \eqref{eq:general-model-xbx}. A directed edge $x_i \to x_j$ exists between variable $x_i$ and variable $x_j$ whenever $\bB_{j,i}\neq 0$, where $\bB_{j,i}$ denotes the $(j,i)$ element of $\bB$. The path coefficient of this edge is given by $\omega_{x_i, x_j}=\bB_{j,i}$ and is interpreted as the expected change in $x_j$ due to a unit change in $x_i$ and conditional on $x_1, \ldots, x_{j-1}$; it is a conditional prediction. A directed edge $\varepsilon_i \to x_j$ from structural shock $\varepsilon_i$ to variable $x_j$ exists if $\bOmega_{j,i}\neq 0$. The path coefficient of this edge is given by $\omega_{\varepsilon_i, x_j} = \bOmega_{j,i}$; it is the causal effect of $\varepsilon_i$ on $x_j$ conditional on $x_1, \ldots, x_{j-1}$.

\begin{example}  \label{example:tri-variate-SVAR-graph}
    Consider the trivariate SVAR(1) given by
    \begin{equation*}
        \bA_0\by_t = \bA_1\by_{t-1} + \bvarepsilon_t,
    \end{equation*} 
    with $\bA_0$ and $\bA_1$ being $3\times 3$ coefficient matrices, $\by_t = (y_{1,t}, y_{2,t}, y_{3,t})'$ and $\bvarepsilon_t = (\varepsilon_{1,t}, \varepsilon_{2,t}, \varepsilon_{3,t})'$. Fixing the horizon $h=1$ and the transmission matrix $\bT=\bI$, 
    \begin{equation*}
    \begin{split}
    \bx&=(y_{1,t}, y_{2, t}, y_{3,t}, y_{1, t+1}, y_{2, t+1}, y_{3, t+1})' \\
    \bvarepsilon &= (\varepsilon_{1,t}, \varepsilon_{2,t}, \varepsilon_{3,t}, \varepsilon_{1, t+1}, \varepsilon_{2, t+1}, \varepsilon_{3, t+1})',
    \end{split}
    \end{equation*}
    and the corresponding graph $\calG(\bB, \bOmega)$ is shown in Figure \ref{fig:svar-varma-examples}(a). Due to the dynamic nature of the SVAR(1), an edge exists from each component of $\by_t$ (i.e.\ $x_1$ to $x_3$ in the graph) to each component of $\by_{t+1}$ (i.e.\ $x_4$ to $x_6$ in the graph). 
\end{example}

\begin{example}
    Assume a trivariate SVARMA(1, 1) given by
    \begin{equation*}
        \bA_0\by_t = \bA_1\by_{t-1} + \bvarepsilon_t + \bPsi_1\bvarepsilon_{t-1},
    \end{equation*} 
    with $\bA_0$, $\bA_1$ and $\bPsi_1$ being $3\times 3$ coefficient matrices, $\by_t = (y_{1,t}, y_{2,t}, y_{3,t})'$ and $\bvarepsilon_t = (\varepsilon_{1,t}, \varepsilon_{2,t}, \varepsilon_{3,t})'$. Fixing the horizon $h=1$ and the transmission matrix $\bT=\bI$, $\bx$ and $\bvarepsilon$ are the same as in Example \ref{example:tri-variate-SVAR-graph}, but the graph is given in Figure \ref{fig:svar-varma-examples}(b). Similar to the graph corresponding to the SVAR(1), edges exist from each component of $\by_t$ to each component of $\by_{t+1}$. Unlike the SVAR case, edges from structural shocks can now also skip time periods, thereby going from each component of $\bvarepsilon_t$ to each component of  $\by_{t+1}$. One such edge is given by $\varepsilon_1 \to x_4$, which is still interpreted as the conditional causal effect of $\varepsilon_1$ on $x_4$.
\end{example}

\begin{figure}
    \centering
    \begin{subfigure}{0.4\textwidth}
    \centering
    \begin{tikzpicture}[scale=0.8, transform shape]
        \node[circle, draw, minimum size=1cm, line width=0.8pt] (y11) at (0, 4) {$x_1$};
        \node[circle, draw, line width=0.8pt, minimum size=1cm, black] (y21) at (0, 2) {$x_2$};
        \node[circle, draw, line width=0.8pt, minimum size=1cm, black] (y31) at (0, 0) {$x_3$};
        
        \node[circle, draw, minimum size=1cm, line width=0.8pt] (y12) at (3, 4) {$x_4$};
        \node[circle, draw, line width=0.8pt, minimum size=1cm, black] (y22) at (3, 2) {$x_5$};
        \node[circle, draw, minimum size=1cm, line width=0.8pt] (y32) at (3, 0) {$x_6$};

        \node[rectangle, draw=blue!50, fill=blue!20, line width=0.5pt, minimum size=0.3cm] (e11) at (-2, 4) {\textcolor{black}{$\varepsilon_{1}$}};
        \node[rectangle, draw=blue!50, fill=blue!20, line width=0.5pt, minimum size=0.3cm] (e21) at (-2, 2) {\textcolor{black}{$\varepsilon_{2}$}};
        \node[rectangle, draw=blue!50, fill=blue!20, line width=0.5pt, minimum size=0.3cm] (e31) at (-2, 0) {\textcolor{black}{$\varepsilon_{3}$}};

        \draw[-latex, line width=0.8pt, black] (y11) -- (y21);
        \draw[-latex, line width=0.8pt, black] (y11) edge[in=45, out=-45] (y31);
        \draw[-latex, line width=0.8pt, black] (y21) -- (y31);

        \draw[-latex, line width=0.8pt, black] (y12) -- (y22);
        \draw[-latex, line width=0.8pt, black] (y12) edge[in=45, out=-45] (y32);
        \draw[-latex, line width=0.8pt, black] (y22) -- (y32);

        \draw[-latex, line width=0.8pt, black] (y11) -- (y12);
        \draw[-latex, line width=0.8pt, black] (y11) -- (y22);
        \draw[-latex, line width=0.8pt, black] (y11) -- (y32);
        \draw[-latex, line width=0.8pt, black] (y21) -- (y12);
        \draw[-latex, line width=0.8pt, black] (y21) -- (y22);
        \draw[-latex, line width=0.8pt, black] (y21) -- (y32);
        \draw[-latex, line width=0.8pt, black] (y31) -- (y12);
        \draw[-latex, line width=0.8pt, black] (y31) -- (y22);
        \draw[-latex, line width=0.8pt, black] (y31) -- (y32);

        \draw[-latex, line width=0.8pt, black] (e11) -- (y11);
        \draw[-latex, line width=0.8pt, black] (e11) -- (y21);
        \draw[-latex, line width=0.8pt, black] (e11) -- (y31);
        \draw[-latex, line width=0.8pt, black] (e21) -- (y11);
        \draw[-latex, line width=0.8pt, black] (e21) -- (y21);
        \draw[-latex, line width=0.8pt, black] (e21) -- (y31);
        \draw[-latex, line width=0.8pt, black] (e31) -- (y11);
        \draw[-latex, line width=0.8pt, black] (e31) -- (y21);
        \draw[-latex, line width=0.8pt, black] (e31) -- (y31);

        \useasboundingbox (-2, -2) rectangle (4, 4);
        
    \end{tikzpicture}
    \caption{}
    \end{subfigure}
    \hfill
    \begin{subfigure}{0.4\textwidth}
    \centering
    \begin{tikzpicture}[scale=0.8, transform shape]
        \node[circle, draw, minimum size=1cm, line width=0.8pt] (y11) at (0, 4) {$x_1$};
        \node[circle, draw, line width=0.8pt, minimum size=1cm, black] (y21) at (0, 2) {$x_2$};
        \node[circle, draw, line width=0.8pt, minimum size=1cm, black] (y31) at (0, 0) {$x_3$};
        
        \node[circle, draw, minimum size=1cm, line width=0.8pt] (y12) at (3, 4) {$x_4$};
        \node[circle, draw, line width=0.8pt, minimum size=1cm, black] (y22) at (3, 2) {$x_5$};
        \node[circle, draw, minimum size=1cm, line width=0.8pt] (y32) at (3, 0) {$x_6$};

        \node[rectangle, draw=blue!50, fill=blue!20, line width=0.5pt, minimum size=0.3cm] (e11) at (-2, 4) {\textcolor{black}{$\varepsilon_{1}$}};
        \node[rectangle, draw=blue!50, fill=blue!20, line width=0.5pt, minimum size=0.3cm] (e21) at (-2, 2) {\textcolor{black}{$\varepsilon_{2}$}};
        \node[rectangle, draw=blue!50, fill=blue!20, line width=0.5pt, minimum size=0.3cm] (e31) at (-2, 0) {\textcolor{black}{$\varepsilon_{3}$}};

        \draw[-latex, line width=0.8pt, black] (y11) -- (y21);
        \draw[-latex, line width=0.8pt, black] (y11) edge[in=45, out=-45] (y31);
        \draw[-latex, line width=0.8pt, black] (y21) -- (y31);

        \draw[-latex, line width=0.8pt, black] (y12) -- (y22);
        \draw[-latex, line width=0.8pt, black] (y12) edge[in=45, out=-45] (y32);
        \draw[-latex, line width=0.8pt, black] (y22) -- (y32);

        \draw[-latex, line width=0.8pt, black] (y11) -- (y12);
        \draw[-latex, line width=0.8pt, black] (y11) -- (y22);
        \draw[-latex, line width=0.8pt, black] (y11) -- (y32);
        \draw[-latex, line width=0.8pt, black] (y21) -- (y12);
        \draw[-latex, line width=0.8pt, black] (y21) -- (y22);
        \draw[-latex, line width=0.8pt, black] (y21) -- (y32);
        \draw[-latex, line width=0.8pt, black] (y31) -- (y12);
        \draw[-latex, line width=0.8pt, black] (y31) -- (y22);
        \draw[-latex, line width=0.8pt, black] (y31) -- (y32);

        \draw[-latex, line width=0.8pt, black] (e11) -- (y11);
        \draw[-latex, line width=0.8pt, black] (e11) -- (y21);
        \draw[-latex, line width=0.8pt, black] (e11) -- (y31);
        \draw[-latex, line width=0.8pt, black] (e21) -- (y11);
        \draw[-latex, line width=0.8pt, black] (e21) -- (y21);
        \draw[-latex, line width=0.8pt, black] (e21) -- (y31);
        \draw[-latex, line width=0.8pt, black] (e31) -- (y11);
        \draw[-latex, line width=0.8pt, black] (e31) -- (y21);
        \draw[-latex, line width=0.8pt, black] (e31) -- (y31);

        \draw[-latex, line width=0.8pt, black] (e11) edge[in=135, out=45] (y12);
        \draw[-latex, line width=0.8pt, black] (e11) edge[in=135, out=-45] (y22);
        \draw[-latex, line width=0.8pt, black] (e11) edge[in=115, out=-25] (y32);
        \draw[-latex, line width=0.8pt, black] (e21) edge[in=-135, out=45] (y12);
        \draw[-latex, line width=0.8pt, black] (e21) edge[in=135, out=45] (y22);
        \draw[-latex, line width=0.8pt, black] (e21) edge[in=135, out=-45] (y32);
        \draw[-latex, line width=0.8pt, black] (e31) edge[in=-115, out=25] (y12);
        \draw[-latex, line width=0.8pt, black] (e31) edge[in=-135, out=45] (y22);
        \draw[-latex, line width=0.8pt, black] (e31) edge[in=-135, out=-45] (y32);

        \useasboundingbox (-2, -2) rectangle (4, 4);
    \end{tikzpicture}
    \caption{}
    \end{subfigure}
    \caption{Panels (a) and (b) show the DAG corresponding to a SVAR(1) and a SVARMA(1, 1) respectively, where $\bx=(y_{1,t}, y_{2, t}, y_{3,t}, y_{1, t+1}, y_{2, t+1}, y_{3, t+1})'$. Compared to panel (a), panel (b) has edges from all shocks $\varepsilon_{\cdot,t}$ to all variables in period $t+1$, $y_{\cdot, t+1}$, where $\cdot$ is used as a placeholder.}
    \label{fig:svar-varma-examples}
\end{figure}

As edges from shocks are conditional causal effects, while edges between endogenous variables are conditional predictions, paths from a shock to a variable are chains of conditional predictions triggered by a structural shock. The collection $\calP_{\varepsilon_i, x_j}$ of all paths connecting structural shock $\varepsilon_i$ to variable $x_j$ is therefore a set of conditional prediction chains triggered by the shock $\varepsilon_i$; so is any subset $P_{\varepsilon_i, x_j}$ of this collection. We are now ready to define transmission channels.

\begin{definition}
    Let $\calG(\bB, \bOmega)$ be the graph induced by model \eqref{eq:general-model-xbx}. A subset of paths from $\varepsilon_i$ to $x_j$, $P_{\varepsilon_i, x_j}\subset\calP_{\varepsilon_i, x_j}$, is called a \emph{transmission channel}.    
    \label{def:transmission-channel}
\end{definition}

TCA 
quantifies
transmission effects as the effects flowing through transmission channels. To this end, we define the effects along a specific path $p$ and the effects along a collection of paths $P_{\varepsilon_i, x_j}$, the latter leads to the definition of the transmission effect. 

\begin{definition}
    Let $\calG(\bB, \bOmega)$ be a graph induced by model \eqref{eq:general-model-xbx}, $\calP_{\varepsilon_i, x_j}$ be the collection of all paths in $\calG(\bB, \bOmega)$ from $\varepsilon_i$ to $x_j$, and $\xi$ be the shock size. 
    \begin{enumerate}
        \item The \emph{path-specific effect} of a path 
        $p\in\calP_{\varepsilon_i, x_j}$ is 
        \begin{equation*}
        \calQ_{\xi}(p) = \xi\prod_{u \to v \in p}\omega_{u,v}, 
        \end{equation*} 
        with  path coefficient $\omega_{u, v}$ of the edge  from node $u$ to node $v$ on path $p$.
        \item The \emph{total path-specific effect} of a sub-collection of paths $P_{\varepsilon_i, x_j}\subseteq \calP_{\varepsilon_i, x_j}$ is 
        \begin{equation*}
       \calQ_{\xi}(P_{\varepsilon_i, x_j}) = \sum_{p\in P_{\varepsilon_i, x_j}}\calQ_{\xi}(p).
        \end{equation*}
        \item The \emph{transmission effect} of the transmission channel $P_{\varepsilon_i, x_j}\subset \calP_{\varepsilon_i, x_j}$ is the total path-specific effect of $P_{\varepsilon_i, x_j}$.
    \end{enumerate}
    \label{def:transmission-effect}
\end{definition}

Having formally defined TCA in dynamic models for system \eqref{eq:general-model-xbx} and its associated graph $\calG(\bB, \bOmega)$ (Definitions \ref{def:transmission-channel} and \ref{def:transmission-effect}), we provide some intuition for our three main results of Section \ref{sec:sec4-main-results}. For ease of exposition, consider a transmission channel $P_{\varepsilon_i, x_j}$ from $\varepsilon_i$ to $x_j$ through $x_k$ consisting of all paths starting at $\varepsilon_i$, going to $x_k$ and ending in $x_j$. First, note that the total effect of $\varepsilon_i$ on $x_j$ equals the total path-specific effect of $\calP_{\varepsilon_i, x_j}$. This ensures the coherence of our framework.

Second, the transmission channels of the structural shock $\varepsilon_i$ only require $\varepsilon_i$ to be structurally identified. Edges leaving the structural shock $\varepsilon_i$ are conditional causal effects of $\varepsilon_i$. In contrast, edges leaving any intermediate variable $x_m$ on a path $p\in P_{\varepsilon_i, x_j}$ are conditional predictions. Therefore, only the single edge on any path $p\in P_{\varepsilon_i, x_j}$ originating from $\varepsilon_i$ requires structural identification. 

Third, in the SVAR case all transmission effects decomposing the total effect of $\varepsilon_i$ on $x_j$ can be calculated by combining the structural impulse responses of $\varepsilon_i$ with non-structural impulse responses; impulse responses are therefore sufficient to compute any transmission effect in SVARs. All paths $p$ corresponding to the transmission channel $P_{\varepsilon_i, x_j}$ go through the intermediate variable $x_k$, at which the channel can be split. The first part consists of all paths connecting $\varepsilon_i$ to $x_k$. The effect through these paths, the transmission effect, is the total effect of $\varepsilon_i$ on $x_k$ obtained through the structural impulse response $\bPhi_{k,i}$. The second part consists of all paths connecting $x_k$ to $x_j$, whose transmission effect is the expected change in $x_j$ due to a unit increase in $x_k$, conditional on $x_l$ for $l=1, \dots, k-1$; in other words, a conditional prediction and thus a reduced-form impulse response.

Calculating the transmission effect using impulse responses is, however, limited to SVARs and is not possible in general. In SVARMA models, the direct edges from shocks to endogenous variables at later periods can generally not be recovered from structural or reduced form impulse responses. The general SVARMA case therefore requires knowledge of the structural moving average matrices, even though only a single shock still needs to be identified. Appendix~\ref{appendix:computation} provides an algorithm for computing transmission effects that is valid in any SVARMA model and does not rely on impulse responses. Next, we prove these main findings using the potential outcomes framework.

\section{Main Results}
\label{sec:sec4}

We establish our main results on TCA, using the potential outcomes framework as a convenient mathematical tool. Section~\ref{sec:sec4-potential-outcomes} introduces potential outcomes for TCA and shows their equivalence with the graphical definitions and representation. Section~\ref{sec:sec4-main-results} then shows that 
(i) the total effect of $\varepsilon_i$ on $x_j$ is the total path-specific effect of $\calP_{\varepsilon_i, x_j}$, 
(ii) only $\varepsilon_i$, the shock of interest, needs to be structurally identified for transmission effects to be identified, and
(iii) IRFs are sufficient statistics within SVARs to compute transmission effects, since structural IRFs can be combined with IRFs to Cholesky-orthogonalised shocks to obtain any transmission effect. Although transmission channels generally depend on the chosen transmission matrix $\bT$, Section~\ref{sec:sec4-transmission-matrix} gives conditions under which transmission effects are invariant to that choice. All proofs are provided in Appendix~\ref{appendix:proofs}.

\subsection{Transmission Effects as Potential Outcomes}
\label{sec:sec4-potential-outcomes}

To formalise the thought experiments underlying TCA, we introduce potential outcomes, adapting the notation from the mediation analysis literature \citep{danielCausalMediationAnalysis2015}. Unlike the potential outcomes used thus far in macroeconomics, which focus solely on total effects and can therefore be interpreted causally, ours cannot be interpreted as fully causal. Nonetheless, and similar to the interpretation used in mediation analysis, our potential outcomes correspond to clear thought experiments in which some margins of adjustment are allowed to vary while others are held fixed. We retain the term potential outcome throughout the remainder of the paper, although it is used here in a looser sense than in the standard causal inference literature. We show that this potential outcomes framework is equivalent to the more intuitive graphical framework for TCA introduced in Section~\ref{sec:sec3-general-framework}. Throughout this section, we consider the effects of the structural shock $\varepsilon_i$ on the variable $x_j$.

\subsubsection{Notation}

Recall system \eqref{eq:general-model-xbx}, $\bx = \bB\bx + \bOmega\bvarepsilon$. Since $\bB$ is lower-triangular with zeros on the diagonal, $x_j$ directly depends on the realisations of $x_i$ for $i<j$, which, in turn, directly depends on the realisations of $x_k$ for $k<i$. All variables may, additionally, depend on realisations of the shock of interest $\varepsilon_i$. However, shocks themselves are independent of all other variables and shocks. Thus, chains can be traced back from the variable of interest $x_j$ to the structural shock $\varepsilon_i$. Potential outcomes are statements about these chains, with the potential outcome for $x_j$ taking the form
\begin{equation*}
x_j(x_1(\epsilon_1), x_2(x_1(\epsilon_2), \epsilon_3), \dots) = x_j^*(\epsilon_1, \epsilon_2, \epsilon_3, \dots, \epsilon_{\eta(j)}) = x_j(\bepsilon^{(j)}),     
\end{equation*}
where $\bepsilon^{(j)}$ is the potential outcomes assignment vector. The elements of this vector, $\epsilon_l$ for $l=1\dots \eta(j)$, are the realisations of the structural shock experienced by the $l$th nested chain, where $\eta(j)=2^{j-1}$ is the number of chains moving back from $x_j$ to $\varepsilon_i$. 

\begin{example}
    \label{ex:potential-outcome-notation}
    Consider the DSGE model of Section~\ref{sec:sec2-illustrative-example}, reduced to consumption $C$, investment $I$, and inflation $\pi$, with the associated graph shown in panel (a) of Figure~\ref{fig:sec4-po-example}, where all shocks except the monetary policy shock $\varepsilon^M$ are omitted and time subscripts are dropped for clarity. The first step towards introducing potential outcomes is to bring the model into general notation: $\varepsilon_1 = \varepsilon^M_t$, $x_1 = C_t$, $x_2 = I_t$, and $x_3 = \pi_t$, as shown in panel (a'). The potential outcome for the inflation rate $x_3$ is then given by
    \begin{equation*}
        x_3(x_1(\epsilon_1), x_2(x_1(\epsilon_2), \epsilon_3), \epsilon_4) = x_3^*(\epsilon_1, \epsilon_2, \epsilon_3, \epsilon_4) = x_3^*(\bepsilon^{(3)}),
    \end{equation*}
    where $\epsilon_l$ ($l=1, \ldots, 4$) are realisations of the structural shock $\varepsilon_1$. These realisations need not coincide. Panels (b)--(e) of Figure~\ref{fig:sec4-po-example} depict the nested chains corresponding to the realisations $\epsilon_1$ to $\epsilon_4$, respectively. Panel (b) shows that $x_3$ depends on the realisation of $x_1$, which in turn depends on a realisation of the structural shock $\varepsilon_1$, denoted by $\epsilon_1$. Panels (c) and (d) both start from the dependence of $x_3$ on the realisation of $x_2$. Panel (c) considers the chain in which $x_2$ depends on the realisation of $x_1$, which in turn depends on the realisation $\epsilon_2$. Panel (d) considers the chain in which $x_2$ depends directly on the realisation $\epsilon_3$. Lastly, panel (e) shows that $x_3$ depends directly on the realisation $\epsilon_4$.
\end{example}

\begin{figure}
    \centering
    \begin{subfigure}{0.14\textwidth}
       \centering 
        \begin{tikzpicture}[scale=0.55, transform shape]
            \node[circle, draw, line width=1pt, minimum size=0.8cm] (x) at (0, 4) {$C$};
            \node[circle, draw, line width=1pt, minimum size=0.8cm] (pi) at (0, 2) {$I$};
            \node[circle, draw, line width=1pt, minimum size=0.8cm] (i) at (0, 0) {$\pi$};
            
            \draw[-latex, line width=1pt] (x) -- (pi);
            \draw[-latex, line width=1pt] (pi) -- (i); 
            \draw[-latex, line width=1pt] (x) edge [in=45, out=-45] (i);
    
            \node[rectangle, draw=blue!50, fill=blue!20, line width=0.5pt, minimum size=0.3cm] (ed) at (-1.5, 4) {$\varepsilon^M$};
    
            \draw[-latex, line width=1pt] (ed) -- (x); 
            \draw[-latex, line width=1pt] (ed) -- (pi); 
            \draw[-latex, line width=1pt] (ed) -- (i); 
    
            \useasboundingbox (-2, -2) rectangle (6, 6);
        \end{tikzpicture}
        \caption{}
    \end{subfigure}
    \hfill 
    \begin{subfigure}{0.14\textwidth}
       \centering 
        \begin{tikzpicture}[scale=0.55, transform shape]
            \node[circle, draw, line width=1pt, minimum size=0.8cm] (x1) at (0, 4) {$x_1$};
            \node[circle, draw, line width=1pt, minimum size=0.8cm] (x2) at (0, 2) {$x_2$};
            \node[circle, draw, line width=1pt, minimum size=0.8cm] (x3) at (0, 0) {$x_3$};
            
            \draw[-latex, line width=1pt] (x1) -- (x2);
            \draw[-latex, line width=1pt] (x2) -- (x3); 
            \draw[-latex, line width=1pt] (x1) edge [in=45, out=-45] (x3);
    
            \node[rectangle, draw=blue!50, fill=blue!20, line width=0.5pt, minimum size=0.3cm] (e1) at (-1.5, 4) {$\varepsilon_1$};
    
            \draw[-latex, line width=1pt] (e1) -- (x1); 
            \draw[-latex, line width=1pt] (e1) -- (x2); 
            \draw[-latex, line width=1pt] (e1) -- (x3); 
            
            \useasboundingbox (-2, -2) rectangle (6, 6);
        \end{tikzpicture}
        \caption*{(a')}
    \end{subfigure}
    \hfill 
    \begin{subfigure}{0.14\textwidth}
        \centering
        \begin{tikzpicture}[scale=0.55, transform shape]
            \node[circle, draw, line width=1pt, minimum size=0.8cm] (x1) at (0, 4) {$x_1$};
            \node[circle, draw, line width=1pt, minimum size=0.8cm, black!30] (x2) at (0, 2) {$x_2$};
            \node[circle, draw, line width=1pt, minimum size=0.8cm] (x3) at (0, 0) {$x_3$};
            
            \draw[-latex, line width=0.8pt, black!30] (x1) -- (x2);
            \draw[-latex, line width=0.8pt, black!30] (x2) -- (x3); 
            \draw[-latex, line width=1pt, black] (x1) edge [in=45, out=-45] (x3);
    
            \node[rectangle, draw=blue!50, fill=blue!20, line width=0.5pt, minimum size=0.3cm] (e1) at (-1.5, 4) {$\varepsilon_1$};
    
            \draw[-latex, line width=1pt, black] (e1) -- (x1); 
            \draw[-latex, line width=0.8pt, black!30] (e1) -- (x2); 
            \draw[-latex, line width=0.8pt, black!30] (e1) -- (x3); 
    
            \node () at (-0.75, -1) {\footnotesize$x_3(x_1(\epsilon_1), \textcolor{black!30}{x_2(x_1(\epsilon_2), \epsilon_3), \epsilon_4})$};
            \node () at (-0.75, -1.5) {\footnotesize$=x_3^*(\epsilon_1, \textcolor{black!30}{\epsilon_2, \epsilon_3, \epsilon_4})$};
    
            \useasboundingbox (-2, -2) rectangle (6, 6);
        \end{tikzpicture}
        \caption{}
    \end{subfigure}
    \hfill
    \begin{subfigure}{0.14\textwidth}
        \centering
        \begin{tikzpicture}[scale=0.55, transform shape]
            \node[circle, draw, line width=1pt, minimum size=0.8cm] (x1) at (0, 4) {$x_1$};
            \node[circle, draw, line width=1pt, minimum size=0.8cm] (x2) at (0, 2) {$x_2$};
            \node[circle, draw, line width=1pt, minimum size=0.8cm] (x3) at (0, 0) {$x_3$};
            
            \draw[-latex, line width=1pt, black] (x1) -- (x2);
            \draw[-latex, line width=1pt, black] (x2) -- (x3); 
            \draw[-latex, line width=0.8pt, black!30] (x1) edge [in=45, out=-45] (x3);
    
            \node[rectangle, draw=blue!50, fill=blue!20, line width=0.5pt, minimum size=0.3cm] (e1) at (-1.5, 4) {$\varepsilon_1$};
    
            \draw[-latex, line width=1pt, black] (e1) -- (x1); 
            \draw[-latex, line width=0.8pt, black!30] (e1) -- (x2); 
            \draw[-latex, line width=0.8pt, black!30] (e1) -- (x3); 
    
            \node () at (-0.75, -1) {\footnotesize$x_3(\textcolor{black!30}{x_1(\epsilon_1)}, x_2(x_1(\epsilon_2)\textcolor{black!30}{, \epsilon_3})\textcolor{black!30}{, \epsilon_4})$};
            \node () at (-0.75, -1.5) {\footnotesize$=x_3^*(\textcolor{black!30}{\epsilon_1,} \epsilon_2\textcolor{black!30}{, \epsilon_3, \epsilon_4})$};
    
            \useasboundingbox (-2, -2) rectangle (6, 6);
        \end{tikzpicture}
        \caption{}
    \end{subfigure}
    \hfill
    \begin{subfigure}{0.14\textwidth}
        \centering
        \begin{tikzpicture}[scale=0.55, transform shape]
            \node[circle, draw, line width=1pt, minimum size=0.8cm, black!30] (x1) at (0, 4) {$x_1$};
            \node[circle, draw, line width=1pt, minimum size=0.8cm] (x2) at (0, 2) {$x_2$};
            \node[circle, draw, line width=1pt, minimum size=0.8cm] (x3) at (0, 0) {$x_3$};
            
            \draw[-latex, line width=0.8pt, black!30] (x1) -- (x2);
            \draw[-latex, line width=1pt, black] (x2) -- (x3); 
            \draw[-latex, line width=0.8pt, black!30] (x1) edge [in=45, out=-45] (x3);
    
            \node[rectangle, draw=blue!50, fill=blue!20, line width=0.5pt, minimum size=0.3cm] (e1) at (-1.5, 4) {$\varepsilon_1$};
    
            \draw[-latex, line width=0.8pt, black!30] (e1) -- (x1); 
            \draw[-latex, line width=1pt, black] (e1) -- (x2); 
            \draw[-latex, line width=0.8pt, black!30] (e1) -- (x3); 
    
            \node () at (-0.75, -1) {\footnotesize$x_3(\textcolor{black!30}{x_1(\epsilon_1)}, x_2(\textcolor{black!30}{x_1(\epsilon_2),} \epsilon_3)\textcolor{black!30}{, \epsilon_4})$};
            \node () at (-0.75, -1.5) {\footnotesize$=x_3^*(\textcolor{black!30}{\epsilon_1, \epsilon_2}, \epsilon_3\textcolor{black!30}{, \epsilon_4})$};
    
            \useasboundingbox (-2, -2) rectangle (6, 6);
        \end{tikzpicture}
        \caption{}
        \end{subfigure}
        \hfill
        \begin{subfigure}{0.14\textwidth}
            \centering
            \begin{tikzpicture}[scale=0.55, transform shape]
                \node[circle, draw, line width=1pt, minimum size=0.8cm, black!30] (x1) at (0, 4) {$x_1$};
                \node[circle, draw, line width=1pt, minimum size=0.8cm, black!30] (x2) at (0, 2) {$x_2$};
                \node[circle, draw, line width=1pt, minimum size=0.8cm] (x3) at (0, 0) {$x_3$};
                
                \draw[-latex, line width=0.8pt, black!30] (x1) -- (x2);
                \draw[-latex, line width=0.8pt, black!30] (x2) -- (x3); 
                \draw[-latex, line width=0.8pt, black!30] (x1) edge [in=45, out=-45] (x3);
        
                \node[rectangle, draw=blue!50, fill=blue!20, line width=0.5pt, minimum size=0.3cm] (e1) at (-1.5, 4) {$\varepsilon_1$};
        
                \draw[-latex, line width=0.8pt, black!30] (e1) -- (x1); 
                \draw[-latex, line width=0.8pt, black!30] (e1) -- (x2); 
                \draw[-latex, line width=1pt, black] (e1) -- (x3); 
        
                \node () at (-0.75, -1) {\footnotesize$x_3(\textcolor{black!30}{x_1(\epsilon_1), x_2(x_1(\epsilon_2), \epsilon_3), }\epsilon_4)$};
                \node () at (-0.75, -1.5) {\footnotesize$=x_3^*(\textcolor{black!30}{\epsilon_1, \epsilon_2, \epsilon_3, }\epsilon_4)$};

                \useasboundingbox (-2, -2) rectangle (6, 6);
            \end{tikzpicture}
            \caption{}
        \end{subfigure}
    \caption{Panel (a) shows the DAG of the DSGE model introduced in Section~\ref{sec:sec2-illustrative-example} with all except the monetary policy shock omitted, while panel (a') shows the same DAG but in general notation. Panels (b) through (e) show, on top, the paths (in black) corresponding to each element in the potential outcome and its assignment vector (highlighted in black below the graph).}
    \label{fig:sec4-po-example}
\end{figure}

As evident from the example, keeping track of the various dependencies and the corresponding shock realisations in the potential outcomes assignment vector $\bepsilon^{(j)}$ quickly leads to notational clutter. For notational convenience, we therefore introduce additional indexing notation to keep track of elements in the assignment vector that correspond to dependencies involving certain intermediate variables. Defining $H(j)=2^j-1$, the entries $H(k_1-1)+1$ to $H(k_1)$ in the assignment vector $\bepsilon^{(j)}$ correspond to the  dependencies of the variable of interest $x_j$ on variable $x_{k_1}$, which we 
collect in $\bepsilon^{(j\cdot k_1)}=(\epsilon^{(j)}_{H(k_1-1)+1}, ..., \epsilon^{(j)}_{H(k_1)})$. In the subset of  dependencies of $x_j$ on $x_{k_1}$, we can further single out the dependencies on another variable $x_{k_2}$,  $k_2 < k_1$, that is ranked higher up in the system. Then $\bepsilon^{(j\cdot k_1\cdot k_2)}$ collects the entries in the assignment vector corresponding to the  dependencies of $x_j$ on $x_{k_1}$ and $x_{k_1}$ in turn on $x_{k_2}$. Continuing in this way, we thus define
\begin{equation*}
    \begin{split}
        \bepsilon^{(j)} &= (\epsilon_1, ..., \epsilon_{\eta(j)})  \\
        \bepsilon^{(j\cdot k_1)} &= (\epsilon^{(j)}_{H(k_1-1)+1}, ..., \epsilon^{(j)}_{H(k_1)}) \\
        \bepsilon^{(j\cdot k_1\cdot k_2)} &= (\epsilon^{(j\cdot k_1)}_{H(k_2-1)+1}, ..., \epsilon^{(j\cdot k_1)}_{H(k_2)}) \\
        \vdots
    \end{split}
    \label{eq:index-notation}
\end{equation*}

\begin{continueexample}{ex:potential-outcome-notation}
    Figure \ref{fig:sec4-po-example} panels (c) and (d) show the two nested dependencies of $x_3$ on the intermediate variable $x_2$. The 
    elements in the potential outcomes assignment vector $\bepsilon^{(3)}$ corresponding to the shock realisations experienced by these two nested dependencies are $\bepsilon^{(3\cdot 2)}=(\epsilon_2, \epsilon_3)$. We can then further single out the element in the potential outcome assignment vector that corresponds to the dependency of  $x_3$ on $x_2$ and $x_2$ on $x_1$, as given by $\epsilon^{(3\cdot 2\cdot 1)}=\epsilon_2$. Note that we do not introduce additional notation for the final dependency on the structural shock since this is by definition always the last element.
\end{continueexample}

\subsubsection{Potential Transmission Channels}

Each potential outcomes assignment vector $\bepsilon^{(j)}$ defines a thought experiment in which some  chains experience a shock while others may not. The most basic thought experiment 
sets
$\epsilon_i=\xi$ for all $i$, implying that all  chains experience the same shock realisation. A conceptually different thought experiment is one in which $\epsilon_k=0$, for some $1 \leq k \leq \eta(j)$, implying that some of the nested  dependencies may experience a non-zero shock of size $\xi$, while others do not. 
The shock is then carried along a subset of all nested chains---equivalently along a subset of chains of conditional predictions, a subset of paths through the graph---in other words, along a transmission channel. Throughout 
we assume that a nested  chain is either shut off ($\epsilon_k=0$) or is active and together with all other active chains experiences the same realised shock size ($\epsilon_l=\xi)$, as formalised in Assumption \ref{assump:equal-treatment}.

\begin{assumption}
    The potential outcomes assignment vector $\bepsilon^{(j)}$ consists of components $\epsilon_i\in\{0, \xi\}$ for all $1\leq i \leq \eta(j)$.
    \label{assump:equal-treatment}
\end{assumption}

We can now define transmission channels in potential outcome notation.
\begin{definition} \label{def:po-transmission-channel}
     Let $\bepsilon^{(j)}$ be a potential outcomes assignment vector with $\epsilon_i$ satisfying Assumption \ref{assump:equal-treatment}. If $\epsilon_k=0$ and $\epsilon_l=\xi$ for some $k, l$, then $\bepsilon^{(j)}$ is a transmission channel. 
\end{definition}

\begin{continueexample}{ex:potential-outcome-notation}
   Consider the potential outcomes assignment vector $\bepsilon^{(3)}$ in Figure \ref{fig:sec4-po-example}. If $\bepsilon^{(3)}=(0, 1, 0, 0)$ implying $\epsilon^{(3\cdot 2\cdot 1)}=1$, then only the second nested  chain experiences a shock. This chain is shown in Figure \ref{fig:sec4-po-example} panel (c) and corresponds to the transmission channel $\{\varepsilon_1 \to x_1 \to x_2 \to x_3\}$. If $\bepsilon^{(3)}=(0, 1, 1, 0)$, then we obtain a transmission channel that is the combination of the second and third transmission channel shown in Figure \ref{fig:sec4-po-example} panels (c) and (d) respectively.  
\end{continueexample}

\subsubsection{Transmission Effects}

The effect of a potential outcomes assignment vector is measured by comparing it to a baseline. Under Assumption \ref{assump:equal-treatment}, a natural baseline is $\epsilon_i=0$ for all $i=1, \dots, \eta(j)$ such that no chain experiences a shock. If $\bepsilon^{(j)}$ is a transmission channel according to Definition \ref{def:po-transmission-channel}, then the transmission effect compares the reaction of $x_j$ under the thought experiment, $x_j^*(\bepsilon^{(j)})$, to the reaction of $x_j$ under the baseline, $x_j^*(\bm 0)$. This is the effect of $\bepsilon^{(j)}$ on $x_j$. 

\begin{definition} \label{def:po-transmission-effect}
    The effect of 
    $\varepsilon_i$ on 
    $x_j$ is given by 
    \begin{equation*}
        \calC_{\varepsilon_i, x_j}(\bepsilon^{(j)}) = \E[x^*_j(\bepsilon^{(j)}) - x^*_j(\bm 0)],
    \end{equation*}
    where $\bepsilon^{(j)}$ is a potential outcomes assignment vector and $\bm 0$ is a conformable vector of zeros. If $\bepsilon^{(j)}$ is a transmission channel according to Definition \ref{def:po-transmission-channel}, then $\calC_{\varepsilon_i, x_j}(\bepsilon^{(j)})$ is the transmission effect of that transmission channel. 
\end{definition}

\subsubsection{Equivalence to the Graphical Framework}

Finally, in Theorem \ref{thm:framework-equivalence}, we show that transmission channels defined in Definitions \ref{def:transmission-channel} and \ref{def:po-transmission-channel}, and transmission effects defined in Definitions \ref{def:transmission-effect} and \ref{def:po-transmission-effect} are equivalent.

\begin{theorem}
    Assume $\by_t$ is generated by model \eqref{eq:general-model} and Assumptions \ref{assump:general-model-structural-shocks} to \ref{assump:equal-treatment} hold, such that $\by_t$ can be equivalently represented in static form \eqref{eq:general-model-xbx}. Let $\calG(\bB, \bOmega)$ be the graph induced by the model. Consider the effect from the structural shock $\varepsilon_i$ to the variable of interest $x_j$, then we have the following equivalence:
    \begin{enumerate}[label=(\roman*)]
        \item Given a collection of paths $P_{\varepsilon_i, x_j}\subseteq\calP_{\varepsilon_i, x_j}$, there exists a potential outcomes assignment vector $\bepsilon^{(j)}$ such that $\calQ_{\xi}(P_{\varepsilon_i, x_j})=\calC_{\varepsilon_i, x_j}(\bepsilon^{(j)})$. 
        \item Given a potential outcomes assignment vector $\bepsilon^{(j)}$, there exists a collection of paths $P_{\varepsilon_i, x_j}\subseteq\calP_{\varepsilon_i, x_j}$ such that $\calC_{\varepsilon_i, x_j}(\bepsilon^{(j)})=\calQ_{\xi}(P_{\varepsilon_i, x_j})$. 
    \end{enumerate}
    \label{thm:framework-equivalence}
\end{theorem}

\begin{continueexample}{ex:potential-outcome-notation} 
   Figure \ref{fig:sec4-po-example} makes use of the equivalence result between the graphical and the potential outcomes frameworks. 
   Panels (c)-(e) show on top the path that is either shut-off ($\epsilon_i=0$; gray edges) or active ($\epsilon_i=\xi$; black edges) with the potential outcomes underneath highlighting which element in the potential outcomes assignment vector $\bepsilon^{(3)}$ corresponds to the path under consideration. 
\end{continueexample}

\subsection{Properties of TCA}
\label{sec:sec4-main-results}

The previous section introduced potential outcomes and showed in Theorem \ref{thm:framework-equivalence} that the potential outcomes and graphical framework for TCA  are equivalent. In this section, we continue with the potential outcomes framework and use its equivalence  with the graphical framework to derive three main results. We focus throughout on the (transmission) effect of the structural shock $\varepsilon_i$ on $x_j$.

\subsubsection{Total Effects}

First, since all effects flow along paths in the associated graph $\calG(\bB, \bOmega)$, intuitively, the total causal effect should be the effect along all paths connecting $\varepsilon_i$ and $x_j$, $\calP_{\varepsilon_i, x_j}$. Theorem \ref{thm:total-effect-and-decomposition} shows that this is the case. 
\begin{theorem}\label{thm:total-effect-and-decomposition}
    Let $\bm 1_{\eta(j)}$ be a $\eta(j)$ dimensional vector of ones, $P^{(1)}_{\varepsilon_i, x_j}, \dots, P^{(k)}_{\varepsilon_i, x_j}$ be a partition of $\calP_{\varepsilon_i, x_j}$ and $\be_1\dots \be_k$ be $\eta(j)$ dimensional vectors in $\{0, \xi\}^{\eta(j)}$ such that $\sum_{m=1}^k \be_m = \xi\bm 1_{\eta(j)}$. Then, under the conditions of Theorem \ref{thm:framework-equivalence}, 
    \begin{enumerate}[label=(\roman*), ref=\ref{thm:total-effect-and-decomposition}(\roman*)]
        \item\label{thm2-part1} $\calQ_{\xi}(\calP_{\varepsilon_i, x_j})=\calC_{\varepsilon_i, x_j}(\xi \bm 1_{\eta(j)}) = \xi\bPhi_{j,i}$ 
        \item\label{thm2-part2} $\sum_{m=1}^k \calQ_{\xi}(P^{(m)}_{\varepsilon_i, x_j}) = \xi\bPhi_{j,i}$
        \item\label{thm2-part3} $\sum_{m=1}^k \calC_{\varepsilon_i, x_j}(\be_m) = \xi\bPhi_{j,i}$. 
    \end{enumerate}
\end{theorem}

\noindent Part (i) states that the total effect given by entry $\bPhi_{j,i}$ in the impulse response matrix is equal to the total path-specific effect of the collection of all paths $\calQ_1(\calP_{\varepsilon_i, x_j})$, which is in turn equal to the effect of the potential outcomes assignment vector $\bepsilon^{(j)}=\bm 1_{\eta(j)}$ by Theorem \ref{thm:framework-equivalence}. Multiplication by $\xi$ simply accounts for the shock size. Part (ii) states that the total path-specific effect of each set in the partition $P^{(1)}_{\varepsilon_i, x_j}, \dots, P^{(k)}_{\varepsilon_i, x_j}$ of $\calP_{\varepsilon_i, x_j}$ decomposes the total effect. Thus, since each set $P_{\varepsilon_i,x_j}^{(i)}$ in the partition is a transmission channel by Definition \ref{def:transmission-channel}, transmission effects of disjoint transmission channels decompose the total effect. Part (iii) implies the same statement for a set of potential outcomes assignment vectors $\be_1\dots \be_k$ such that $\sum_{m=1}^k \be_m = \xi\bm 1_{\eta(j)}$, which is the potential outcomes equivalent to the partition $P^{(1)}_{\varepsilon_i, x_j}, \dots, P^{(k)}_{\varepsilon_i, x_j}$. 

\subsubsection{Identification Requirements}

The next theorem shows that TCA, and thus the decomposition of total effects, requires only the shock of interest to be identified.
\begin{theorem}
    Under the same conditions as in Theorem \ref{thm:framework-equivalence}, transmission effects can be obtained if the $i$th column of the impulse response matrix, $\bPhi_{\cdot, i}$, is known.
    \label{thm:single-shock}
\end{theorem}
\noindent TCA  provides additional insights into what drives equilibrium dynamics due to a structural shock $\varepsilon_i$, requiring only the identification of $\varepsilon_i$, the structural intervention under investigation. This stands in contrast to counterfactual policy analysis (e.g. \citet{mckayWhatCanTime2023} or \citet{simsDOESMONETARYPOLICY2006}), where in addition to $\varepsilon_i$ further identified shocks are needed to mimic the counterfactual policy rule in a different dynamic equilibrium.

Theorem~\ref{thm:single-shock} implies that TCA is feasible whenever standard macroeconomic analysis of total dynamic causal effects is feasible. In addition, the properties established above provide guidance for computing transmission effects in practice, either by aggregating path effects or using the more efficient algorithms presented in Appendix~\ref{appendix:computation}.%
\footnote{The algorithm is also implemented in the Julia library \href{https://github.com/enweg/TransmissionChannelAnalysis.jl}{TransmissionChannelAnalysis.jl} and the \href{https://github.com/enweg/tca-matlab-toolbox}{TCA MATLAB Toolbox}. The accompanying website \href{https://tca.enricowegner.com}{https://tca.enricowegner.com} documents both packages.} 
We next show that in the special case of an SVAR, another computational approach exists.

\subsubsection{IRF Sufficiency}

As the final step in our theoretical development, Theorem~\ref{thm:irf-sufficiency} shows that, in the special case of an SVAR, the transmission effect of any transmission channel is a combination of structural impulse responses and Cholesky-orthogonalised impulse responses; IRFs are therefore sufficient for the calculation of transmission effects.

\begin{theorem}\label{thm:irf-sufficiency}
    Let $\tilde\bPhi$ collect the IRFs to Cholesky-orthogonalised shocks, with the variables ordered as implied by $\bT$. Suppose that the conditions of Theorem~\ref{thm:framework-equivalence} hold and that, in addition, $\bPsi_j=\bO$ for all $j=1,\ldots,q$. Then, for all $x_r$, $x_s$, $x_j$ and $\varepsilon_i$ with $r<s$,
    \begin{enumerate}[label=(\roman*), ref=\ref{thm:irf-sufficiency}(\roman*)]
        \item\label{thm3-part1} $\calQ_{1}(\calP_{x_r, x_s}) = \tilde\bPhi_{s,r}\tilde\bPhi_{r,r}^{-1}$;
        \item\label{thm3-part2} $\calQ_{\xi}(P_{\varepsilon_i, x_j}) = f(\bPhi, \tilde\bPhi)$ for some function $f: \R^{(h+1)K\times (h+1)K}\times \R^{(h+1)K\times (h+1)K} \to \R$.
    \end{enumerate}
\end{theorem}

Theorem~\ref{thm:irf-sufficiency} extends existing sufficiency results for FEVDs \citep{kilianStructuralVectorAutoregressive2017}, policy counterfactuals \citep{mckayWhatCanTime2023}, and policy optimality \citep{barnichonSufficientStatisticsApproach2023} to TCA. It also implies that if the underlying model is assumed to be an SVAR, the SVAR itself need not be estimated to obtain transmission effects: local projections can be used instead to obtain the IRFs. Because edges may skip time periods, i.e.\ because of the MA dynamics, no similar statement holds for SVARMAs; Appendix~\ref{appendix:counterexample-irf-sufficiency} provides a counterexample.

The specific functional form $f(\cdot)$ is application-dependent; Appendix~\ref{appendix:computation} gives algorithms to obtain the functional form and to compute the transmission effect. Note also that although the sufficiency result makes use of Cholesky-orthogonalised IRFs, the Cholesky orthogonalisation is not used to identify a structural shock; it is simply a computational tool. Only the impulse response matrix $\bPhi$ is structural, while $\tilde\bPhi$ contains impulse responses to shocks orthogonalised with respect to the transmission matrix $\bT$ rather than to a structural ordering. The transmission effects are therefore sensitive to the ordering in $\bT$. However, $f(\cdot)$ often does not involve all columns of $\tilde\bPhi$, so the transmission effect is invariant to the re-ordering of some variables, equivalently to the re-ordering of some rows of $\bT$; the next section makes this precise.

\subsection{Invariance of Transmission Effect}
\label{sec:sec4-transmission-matrix}

TCA requires the researcher to choose a transmission matrix $\bT$, see Definition~\ref{def:transmission-matrix}. With $K$ variables in the VARMA \eqref{eq:general-model}, there are $K!$ choices. We, next, provide conditions under which transmission effects are invariant to the choice of $\bT$, which offers guidance to practitioners on how an appropriate set of transmission matrices for the research question at hand can be specified. 

Our first result identifies when edges in $\calG(\bB, \bOmega)$ remain invariant under variable re-ordering in the transmission matrix.

\begin{lemma}
    \label{lem:invariance-edges}
    Let $\bT$ be a transmission matrix defined by the permutation $\sigma_y(\cdot)$ of $\by_t$ satisfying $\sigma(j) = j$ and $i < j \Leftrightarrow \sigma(i) < \sigma(j)$. Define  $\tilde{\by}_t = \bT\by_t$, $\bx = (\by_t', \by_{t+1}', \ldots, \by_{t+h}')'$, $\tilde{\bx} = (\tilde{\by}_t', \tilde{\by}_{t+1}', \ldots, \tilde{\by}_{t+h}')'$, and $\sigma_x(i) = \sigma_y(\text{mod}(i - 1, K) + 1) + \lfloor\frac{i-1}{K}\rfloor K$, such that $x_i = \tilde{x}_{\sigma_x(i)}$, where $\text{mod}(a, b)$ is the remainder of dividing $a$ by $b$, and $\lfloor a \rfloor$ is the largest integer less-equal $a$. Moreover, let $\bB$ and $\bOmega$ be the systems form matrices for $\bx$ and $\tilde{\bB}$ and $\tilde{\bOmega}$ be the systems form matrices for $\tilde{\bx}$. Then, for $k = j + nK$, $n\in \mathbb{N}^0$ and $l\in \mathbb{N}$,

    \begin{enumerate}[label=(\roman*), ref=\ref{lem:invariance-edges}(\roman*)]
        \item\label{lem-invariance-edges-part1} $\bOmega_{k,l} = \tilde{\bOmega}_{k, l}$
        \item\label{lem-invariance-edges-part2} $\bB_{k,l} = \tilde{\bB}_{k, \sigma_x(l)}$
    \end{enumerate}
\end{lemma}

Lemma \ref{lem-invariance-edges-part1} implies that edges connecting shock $\varepsilon_i$ to variable $x_j$ are invariant to re-ordering variables within the block before and after $x_j$ as long as variables from earlier time periods are not ordered after variables from later time periods. Intuitively, this is because the edge measures the conditional causal effect of $\varepsilon_i$ on $x_j$, conditional on the variables ordered before $x_j$. This only depends on which variables are ordered before and after $x_j$, but not on the ordering within the respective groups.

Similarly, Lemma \ref{lem-invariance-edges-part2} implies that the edge from $x_l$ to $x_k$ remains invariant to re-ordering variables within the block before and after $x_k$, as long as the time ordering remains intact. Thus, although $x_l$ corresponds to $\tilde{x}_{\sigma_x(l)}$ in the alternative ordering, the path coefficient remains unchanged. Intuitively, this is because path coefficients between endogenous variables represent conditional predictions, conditional on all other variables ordered before $x_j$---and past shocks in the VARMA case. Similarly to above, these projection coefficients depend on which variables are ordered before and after $x_j$, but not on the ordering within the groups.

Note the analogy of the previous lemma to the invariance of the ordering in Cholesky identification schemes. There it also only matters which variables are ordered before and after the equation defining the shock, but not on the ordering within those blocks. While the underlying conceptual framework here is different, as it does not relate to identification but to the flow of transmission effects, the analogy can be helpful in guiding the choice of a transmission matrix in practice.

Finally, we present invariance conditions on the IRFs that can be used to compute the transmission effects in the case of an SVAR. This lemma is more general than Lemma~\ref{lem:invariance-edges}: if Lemma~\ref{lem:invariance-edges} holds for each edge on the paths corresponding to the IRF, then Lemma \ref{lemma:invariance-cholesky} holds too, but the reverse need not be true. Since structural IRFs are not subject to the choice of a transmission matrix $\bT$, the invariance conditions in the IRF space only involve the Cholesky IRFs $\tilde\bPhi_{i, j}$. Such conditions have been derived in Proposition 4.1 of \citet{christianoMonetaryPolicyShocks1999}. Lemma \ref{lemma:invariance-cholesky} presents these results in the context of TCA.

\begin{lemma}
    Let $\bT$ be defined by the permutation $\sigma_y$ as in Lemma \ref{lem:invariance-edges}. Let $\sigma_x(i) = \sigma_y(\text{mod}(i - 1, K) + 1) + \lfloor\frac{i-1}{K}\rfloor K$, and let $\tilde\bPhi$ be the Cholesky impulse responses of Theorem \ref{thm:irf-sufficiency}. Then, for $k=j + nK$ with $n\in \mathbb{N}^0$, $\tilde\bPhi_{\cdot, k}$ is invariant to the transmission matrix $\bT$ induced by the permutation $\sigma_y$, respectively, it is invariant to the permutation $\sigma_x$.

    \label{lemma:invariance-cholesky}
\end{lemma}

Intuitively, since $\tilde\bPhi_{\cdot, k}$ measures the effect of a unit increase in $x_k$ without increasing the variables ordered before $x_k$, all that matters is which variables are ordered before and after, but not the ordering within their respective groups. Thus, variables ordered before and after $x_k$ can be re-ordered in their respective groups as long as the time ordering is not broken.

Although Lemma \ref{lemma:invariance-cholesky} is more general than Lemma \ref{lem:invariance-edges}, some transmission channels are easier to analyse for invariance using the latter. Take, for example, the transmission channel consisting of only the direct edge from a shock to a variable. Lemma \ref{lem:invariance-edges} directly provides conditions under which this transmission effect is invariant. Application of Lemma \ref{lemma:invariance-cholesky}, on the other hand, requires representing the single edge as a complicated function of IRFs, making the analysis of invariances much more difficult. Thus, both sets of Lemmas have their merit. 

\section{Empirical Applications}
\label{sec:sec5-applications}

In this section, we empirically demonstrate the versatility of TCA in a variety of settings. Section~\ref{sec:sec5-mckay-wolf} demonstrates how TCA can be used to decompose monetary policy effects estimated using SVARs into direct and indirect effects. Section~\ref{sec:sec5-ramey} applies TCA to government spending within a local projections (LP) framework decomposing the total effect of a government spending news shock into implementation and anticipation effects. 
Section~\ref{sec:sec5-dsge} builds on Section~\ref{sec:sec2-illustrative-example} and shows how TCA can be embedded in counterfactual experiments to investigate how transmission channels change as the characteristics of the model are varied.
Appendix~\ref{appendix:empirical-details} provides additional details, such as graphical representations of the specified transmission channels. All algorithms used for the computation of the transmission effects are implemented in the Julia package \href{https://github.com/enweg/TransmissionChannelAnalysis.jl}{TransmissionChannelAnalysis.jl} and the Matlab \href{https://github.com/enweg/tca-matlab-toolbox}{TCA Toolbox}.%
\footnote{The packages are available on the companion website \href{https://tca.enricowegner.com}{https://tca.enricowegner.com}. Replication code for the empirical applications considered in this section is also available on GitHub: \href{https://github.com/enweg/tca-replication-material}{https://github.com/enweg/tca-replication-material}}

\subsection{Direct and Indirect Effects of Monetary Policy}
\label{sec:sec5-mckay-wolf}

We use TCA to decompose the effects of monetary policy shocks, estimated with an SVAR, into direct and indirect components: direct effects are attributable to contemporaneous changes in the policy rate, indirect effects are not. \citet{mckayWhatCanTime2023} argue that the \citet[][, henceforth RR]{romerNewMeasureMonetary2004} shock series predominantly drives short-term interest rates. In contrast, the \citet[][, henceforth GK]{gertlerMonetaryPolicySurprises2015} series moves long(er)-term rates. The RR series is therefore more likely driven by direct effects, the GK series by indirect effects. TCA provides a framework to quantify and assess this hypothesis by decomposing the total impulse responses to either monetary policy shock into their direct and indirect components. 
To stay as close as possible to \citet{mckayWhatCanTime2023}, we estimate an SVAR(4) in the federal funds rate, the output gap, inflation, and commodity prices over the period 1969Q1 to 2007Q4. We identify the monetary policy shock using either RR or GK as an internal instrument \citep{plagborg-mollerLocalProjectionsVARs2021}. That is, we augment the SVAR with the shock series ordered first and recover the monetary policy shock as the first Cholesky-orthogonalised shock.\footnote{All data was obtained from the replication package of \citet{mckayWhatCanTime2023}.}

We define the direct channel of monetary policy as the part of a shock's effect transmitted through the contemporaneous adjustment of the federal funds rate. The indirect effect is its complement, the part not attributable to that contemporaneous adjustment. These indirect effects can be large for two reasons: the total effect of the shock on the policy rate is small while the total effect on the outcome is large, or the general equilibrium link between the policy rate and the outcome is weak. 
Since we analyse the same economy's reaction to two different types of monetary policy shocks,
the general equilibrium link between the policy rate and the outcome is held constant. All differences between the decompositions are therefore driven 
by the shock series and not by the general equilibrium dynamics.
Thus, indirect effects will be larger if the total effect of the shock on the policy rate is small compared to the total effect on the outcome, a situation common to forward guidance, where the total effect on the policy rate is 
small or even zero while other outcome variables already adjust in anticipation of the future rate change.

To define the channels, we use a transmission matrix that orders the federal funds rate first. This implies that the edge from the monetary policy shock to the federal funds rate captures the total causal effect, while edges from the shock to other variables capture conditional causal effects that condition at least on the federal funds rate. If these conditional causal effects are zero, then the shock captures pure direct effects; zero conditional causal effects are therefore sufficient 
for zero indirect effects. Ordering the federal funds rate first has a further advantage: by the results of Section~\ref{sec:sec4-transmission-matrix}, the transmission effects are invariant to the ordering of the other variables, so no further 
decision on the ordering of the remaining variables needs to be made.
A similar approach is useful when investigating the effect through a single variable, such as a policy variable.

Figure~\ref{fig:instrument-comparison} shows the total, direct, and indirect effects of the RR (top) and GK (bottom) shock, all normalised to a 25bps contemporaneous increase in the federal funds rate. The total effect is shown as a scatter-line, the direct and indirect effects as stacked bars. The GK shock triggers dynamics in the federal funds rate and inflation that are clearly distinct from those induced by the RR shock. Indirect effects explain most of the reactions to the GK shock but only a small part of those to the RR shock. Thus, the results quantitatively confirm the qualitative discussion in \citet{mckayWhatCanTime2023}. The RR shock series measures mostly direct effects of monetary policy, while the GK shock series measures a large indirect component possibly linked to forward guidance. We leave a closer analysis of these indirect effects to future research.

\begin{figure}[t]
    \centering
    \begin{tabular}{cc}
    \textbf{RR} & \includegraphics[align=c, width=0.9\textwidth]{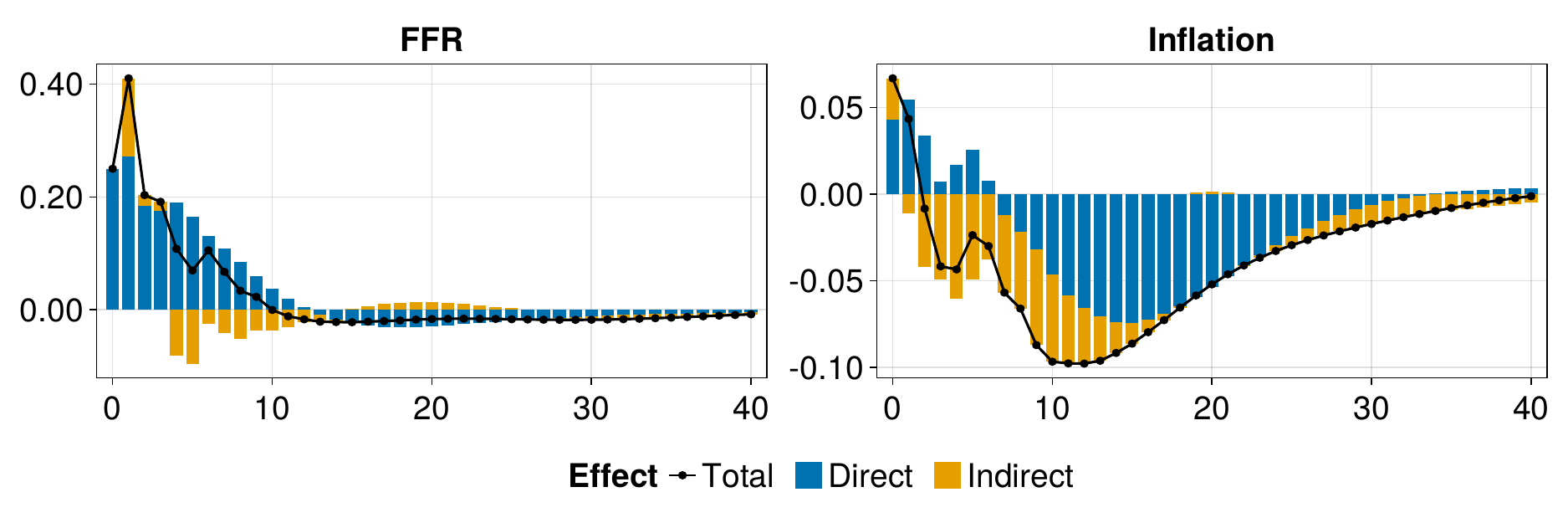} \\
    \textbf{GK} & \includegraphics[align=c, width=0.9\textwidth]{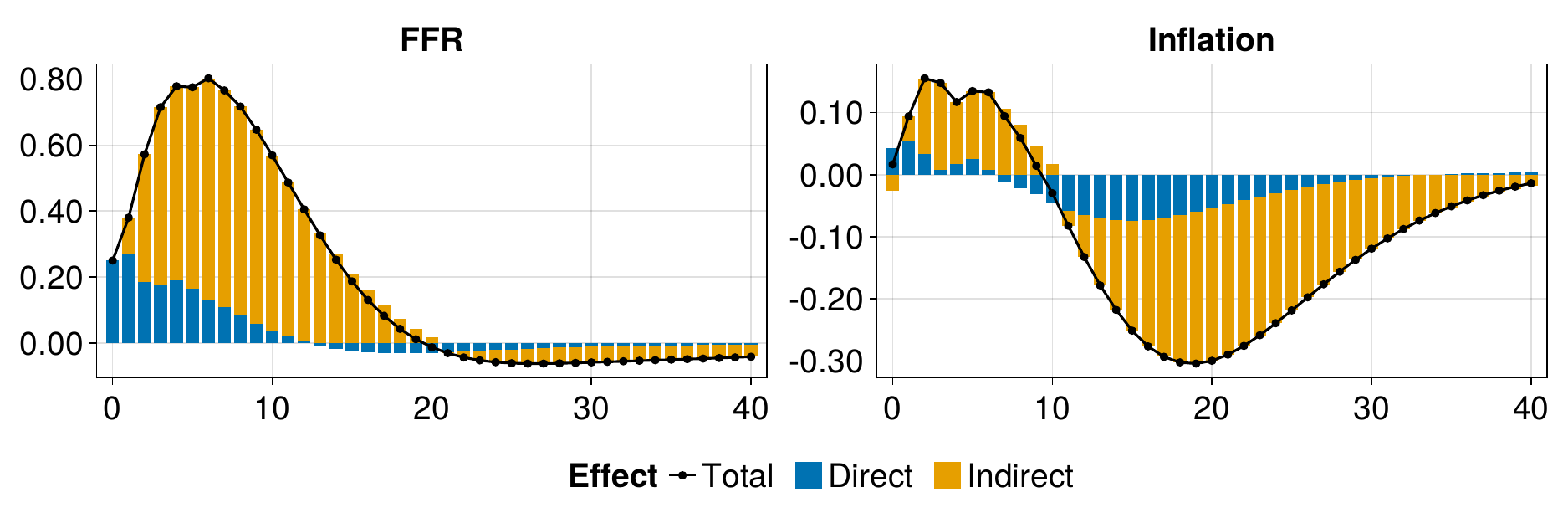}
    \end{tabular}
    \caption{
    Decomposition of the total effect of an RR (top) and GK (bottom) identified monetary policy shock into a direct and indirect channel.
    }
    \label{fig:instrument-comparison}
\end{figure}

\subsection{Anticipation Effects of Government Spending}
\label{sec:sec5-ramey}

In Section~\ref{sec:sec5-mckay-wolf} we defined the direct effect of monetary policy as the effect of a policy that operates through a contemporaneous change in the policy rate. In this section we build on this idea, within a local projections framework, to specify a channel that captures the implementation effects of government spending news shocks. Since we use local projections and therefore rely on the sufficiency of IRFs, we implicitly assume that the underlying model is an SVAR. 

The analysis is motivated by \citet{Ramey2011, rameyGovernmentSpendingMultipliers2018, rameyDefenseNewsShocks2016}, who argue that many series of identified government spending shocks, such as the war-dates series \citep{rameyCostlyCapitalReallocation1998}, fail to capture important anticipation effects.
Since their proposed series captures news about government defence spending, we define the implementation channel as all effects of the news shock that operate through changes in government defence spending. That is, the implementation channel at horizon $h$ comprises all paths that pass through defence spending at least once up to and including horizon $h$. We define the anticipation channel as its complement, capturing all effects that do not run through defence spending. The label ``anticipation channel'' should therefore be read broadly: it captures any effect of the news shock that does not run through defence spending, not anticipation alone.

Our model consists of the news shock series, which we treat as an observed shock, government defence spending, and GDP as our outcome variable. Because we treat the news shock series as an observed shock, we order it first in our transmission ordering. We order government defence spending second, following 
Section~\ref{sec:sec5-mckay-wolf}, where we ordered the policy rate first and the shock was unobserved. Ordering defence spending second lets us decompose the reaction of all other variables into effects that go through defence spending and effects that skip it; the latter we call the anticipation effect. The same pattern, ordering the observed shock first and the policy variable second, can be used in many other TCA applications, especially those investigating shocks to news or expectations. The results of Section~\ref{sec:sec4-transmission-matrix} imply that all impulse response functions are invariant to re-ordering the remaining variables, and that all transmission matrices ordering the news variable first and government defence spending second yield identical transmission effects. This result is trivial in our case, where the only other variable is GDP.

We estimate impulse responses using local projections based on quarterly data from 1890Q1 to 2015Q1. As in \citet{rameyGovernmentSpendingMultipliers2018},
we identify the news shock recursively with the news measure ordered first by estimating the regression
\begin{equation*}
    y_{i, t+h} = \alpha^h + \beta_{i}^h\text{news}_t + \sum_{l=1}^4 \by_{t-l}'\psi_{i, l}^h + \varepsilon_{i,t}^h,
\end{equation*}
where $\by_t=(\text{news}_t, \text{def}_t, \text{gdp}_t)$, $\text{news}_t$ is the news measure, $\text{def}_t$ government defence spending, and $\text{gdp}_t$ is GDP, all expressed in real terms and as percent of real potential GDP. Then, $\beta_i^h$ measures the impulse response of a news shock of one percent of real potential GDP on $y_{i, t+h}$. To estimate transmission channels we need additional, reduced-form impulse responses. These are estimated using local projections of the form
\begin{equation*}
    y_{i, t+h} = \tilde\alpha^h + \tilde\beta_i^h\text{news}_t + \tilde\gamma_i^h\text{def}_t + \sum_{l=1}^4 \by_{t-l}'\tilde\psi_{i,l}^h + \tilde\varepsilon_{i,t}^h,
\end{equation*}
where $\tilde\gamma_i^h$ is the Cholesky impulse response of a government defence spending shock on $y_{i, t+h}$, following the ordering defined in the transmission matrix. Anticipation and implementation effects 
can now be computed as in Appendix \ref{appendix:computation}.

Figure~\ref{fig:sec5-ramey} shows the total, implementation, and anticipation effects of the government (defence) spending news shock on GDP and government defence spending. The total effect is shown as a black scatter-line, while the anticipation and implementation effects appear as stacked bars in blue and yellow, respectively. By construction, the entire response of government defence spending is due to implementation effects. The response of GDP is predominantly driven by persistent, expansionary anticipation effects. Implementation effects appear with a delay and peak around ten quarters after the news shock. At longer horizons, implementation effects turn negative, consistent with crowding out, and drive the reduction in the GDP response. 

A crucial difference from the related literature, which uses defence spending news as an exogenous measure to identify overall government spending news shocks, is that we can attribute observed (implementation) effects solely to changes in defence spending. This may suggest that the implementation effects of defence and civilian government spending differ (cf.~\citealp{Perotti2014}). In contrast, anticipation effects are much less tangibly defined and, by construction, could reflect broader economic expectations, not only those concerning future defence spending. We leave further investigation to future research; a sub-sample analysis would be a natural starting point for uncovering the drivers of this finding (e.g.~\citealp{AcaraiEtAl2023}).

\begin{figure}
    \centering
    \includegraphics[width=0.9\textwidth]{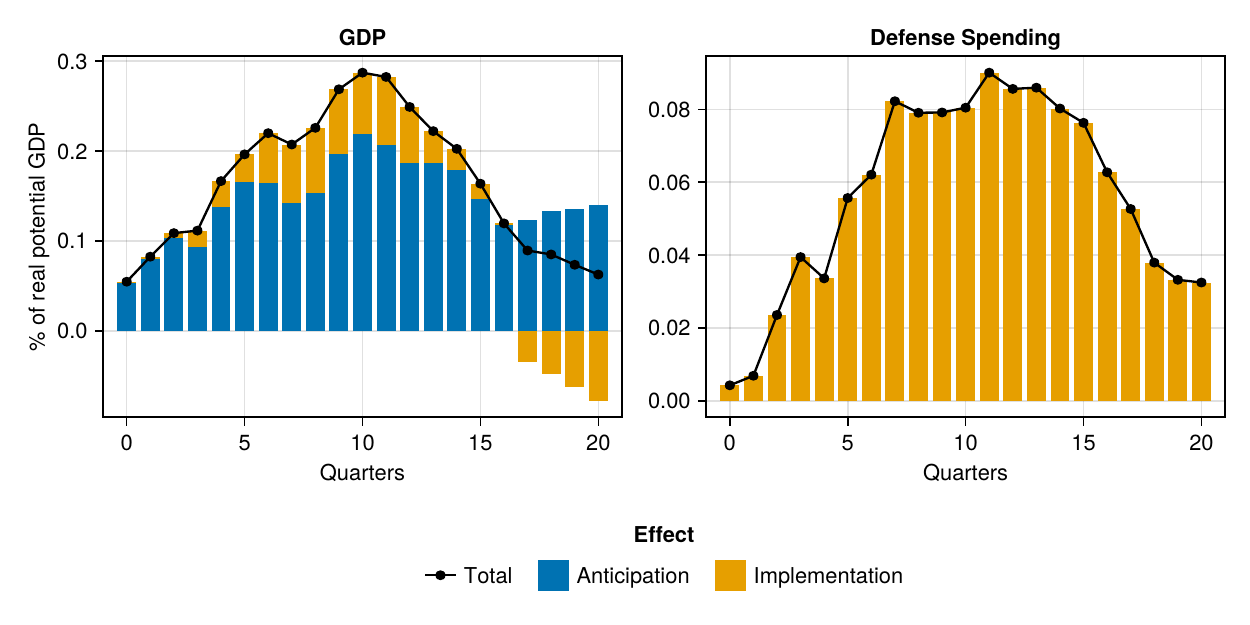}
    \caption{
    Decomposition of the total effect of a defense-news shock of one percent of potential GDP into an anticipation and implementation channel.
    }
    \label{fig:sec5-ramey}
\end{figure}

\subsection{Consumption and Investment Channels of Monetary Policy}
\label{sec:sec5-dsge}

In this section we return to the DSGE model of Section~\ref{sec:sec2-illustrative-example} and embed TCA in counterfactual experiments to quantify how the transmission channels of the model change as the model’s characteristics vary. Recall from Section~\ref{sec:sec2-illustrative-example} that the DSGE model is based on \citet{galiMonetaryPolicyInflation2015} and extended with investment owned by households and used in production. Appendix~\ref{appendix:empirical-details-dsge} gives the details and calibration of the model. Partial-equilibrium reasoning suggests two demand channels of monetary policy: a consumption channel and an investment channel. In Section~\ref{sec:sec2-illustrative-example} we defined these two channels as \eqref{eq:sec2-cons-channel-definition} and \eqref{eq:sec2-invest-channel-definition} at the impact horizon, repeated below for convenience. 
\begin{align*}
    \textbf{Consumption Channel:}\quad & \{\varepsilon^{M} \to C \to \pi\} \\
    \textbf{Investment Channel:}\quad & \{\varepsilon^{M} \to I \to \pi,\; \varepsilon^{M}\to C \to I \to \pi\}. 
\end{align*}
We next extend these channels dynamically and then show how they vary as we change the characteristics of the model.

We illustrate the dynamic extension of the investment channel at horizon $h=1$ in Figure~\ref{fig:section-5-3-dynamic-extension}; analogous definitions apply to the consumption channel. The extensions to horizons $h>1$ are given in Appendix~\ref{appendix:empirical-details-dsge}. At $h=1$, the investment channel splits into the three cases of Figure~\ref{fig:section-5-3-dynamic-extension}, where the first column of endogenous nodes represents horizon $h=0$ and the second column represents horizon $h=1$. In the first (panel a), the monetary policy shock follows all paths to investment contemporaneously; the effect transmits directly to contemporaneous inflation and propagates to inflation one period later. In the second (panel b), the shock again propagates through all paths to investment contemporaneously, but price adjustment frictions delay its transmission to inflation by one period. In the third (panel c), the shock reaches investment only one period later and then transmits to inflation. This last case requires an additional restriction: to keep it distinct from the first, the shock must travel only among real variables before reaching investment at $h=1$. The investment channel at $h=1$ is the union of the three cases.

\begin{figure}
    \begin{subfigure}{0.3\textwidth}
    \centering
        \begin{tikzpicture}[scale=0.8, transform shape]
            \node[circle, draw, minimum size=1cm, line width=0.8pt, black] (C) at (0, 4) {$C$};
            \node[circle, draw, line width=0.8pt, minimum size=1cm, black] (I) at (0, 2) {$I$};
            \node[circle, draw, line width=0.8pt, minimum size=1cm, black] (pi) at (0, 0) {$\pi$};

            \node[circle, draw, minimum size=1cm, line width=0.8pt] (C1) at (2, 4) {$C$};
            \node[circle, draw, line width=0.8pt, minimum size=1cm, black] (I1) at (2, 2) {$I$};
            \node[circle, draw, line width=0.8pt, minimum size=1cm, black] (pi1) at (2, 0) {$\pi$};
    
            \node[rectangle, draw=blue!50, fill=blue!20, line width=0.5pt, minimum size=0.3cm] (emp) at (-2, 4) {\textcolor{black}{$\varepsilon^{M}$}};
            
            \draw[-latex, line width=0.8pt, black] (C) -- (I); 
            \draw[-latex, line width=0.8pt, black] (I) -- (pi); 

            \draw[-latex, line width=0.8pt, lightgray!50] (C) -- (C1); 
            \draw[-latex, line width=0.8pt, lightgray!50] (C) -- (I1); 
            \draw[-latex, line width=0.8pt, lightgray!50] (C) -- (pi1); 
            \draw[-latex, line width=0.8pt, lightgray!50] (I) -- (C1); 
            \draw[-latex, line width=0.8pt, lightgray!50] (I) -- (I1); 
            \draw[-latex, line width=0.8pt, lightgray!50] (I) -- (pi1);

            \draw[-latex, line width=0.8pt, black] (emp) edge (I); 
            \draw[-latex, line width=0.8pt, lightgray!50] (emp) edge (pi);

            \draw[-latex, line width=0.8pt, lightgray!50] (emp.north) to[out=60, in=120, looseness=1.1] (C1.north);
            \draw[-latex, line width=0.8pt, lightgray!50] (emp) -- (I1);
            \draw[-latex, line width=0.8pt, lightgray!50] (emp.south) .. controls (-2, -0.6) and (-0.4, -1.3) .. (1, -1.2) .. controls (1.6, -1.15) and (2, -0.9) .. (pi1.south);

            \draw[-latex, line width=0.8pt, lightgray!50] (C) edge[in=45, out=-45] (pi); 
            \draw[-latex, line width=0.8pt, black] (emp) edge (C); 
            \draw[-latex, line width=0.8pt, black] (C1) -- (I1); 
            \draw[-latex, line width=0.8pt, black] (C1) edge[in=45, out=-45] (pi1); 
            \draw[-latex, line width=0.8pt, black] (I1) -- (pi1); 
            \draw[-latex, line width=0.8pt, black] (pi) -- (C1); 
            \draw[-latex, line width=0.8pt, black] (pi) -- (I1); 
            \draw[-latex, line width=0.8pt, black] (pi) -- (pi1); 
        \end{tikzpicture}
    \caption{}
    \label{fig:section-5-3-dynamic-extension-a}
    \end{subfigure}%
    \hfill
    \begin{subfigure}{0.3\textwidth}
    \centering
        \begin{tikzpicture}[scale=0.8, transform shape]
            \node[circle, draw, minimum size=1cm, line width=0.8pt, black] (C) at (0, 4) {$C$};
            \node[circle, draw, line width=0.8pt, minimum size=1cm, black] (I) at (0, 2) {$I$};
            \node[circle, draw, line width=0.8pt, minimum size=1cm, lightgray!50] (pi) at (0, 0) {$\pi$};

            \node[circle, draw, minimum size=1cm, line width=0.8pt, lightgray!50] (C1) at (2, 4) {$C$};
            \node[circle, draw, line width=0.8pt, minimum size=1cm, lightgray!50] (I1) at (2, 2) {$I$};
            \node[circle, draw, line width=0.8pt, minimum size=1cm, black] (pi1) at (2, 0) {$\pi$};
    
            \node[rectangle, draw=blue!50, fill=blue!20, line width=0.5pt, minimum size=0.3cm] (emp) at (-2, 4) {\textcolor{black}{$\varepsilon^{M}$}};
            
            \draw[-latex, line width=0.8pt, black] (C) -- (I); 
            \draw[-latex, line width=0.8pt, lightgray!50] (C) edge[in=45, out=-45] (pi); 
            \draw[-latex, line width=0.8pt, lightgray!50] (I) -- (pi); 

            \draw[-latex, line width=0.8pt, lightgray!50] (C) -- (C1); 
            \draw[-latex, line width=0.8pt, lightgray!50] (C) -- (I1); 
            \draw[-latex, line width=0.8pt, lightgray!50] (C) -- (pi1); 
            \draw[-latex, line width=0.8pt, lightgray!50] (I) -- (C1); 
            \draw[-latex, line width=0.8pt, lightgray!50] (I) -- (I1); 
            \draw[-latex, line width=0.8pt, lightgray!50] (pi) -- (C1); 
            \draw[-latex, line width=0.8pt, lightgray!50] (pi) -- (I1); 
            \draw[-latex, line width=0.8pt, lightgray!50] (pi) -- (pi1); 
            
            \draw[-latex, line width=0.8pt, lightgray!50] (C1) -- (I1); 
            \draw[-latex, line width=0.8pt, lightgray!50] (C1) edge[in=45, out=-45] (pi1); 
            \draw[-latex, line width=0.8pt, lightgray!50] (I1) -- (pi1); 
    
            \draw[-latex, line width=0.8pt, lightgray!50] (emp) edge (pi);
            \draw[-latex, line width=0.8pt, black] (emp) edge (C); 
            \draw[-latex, line width=0.8pt, black] (emp) edge (I); 
            \draw[-latex, line width=0.8pt, lightgray!50] (emp.north) to[out=60, in=120, looseness=1.1] (C1.north);
            \draw[-latex, line width=0.8pt, lightgray!50] (emp) -- (I1);
            \draw[-latex, line width=0.8pt, lightgray!50] (emp.south) .. controls (-2, -0.6) and (-0.4, -1.3) .. (1, -1.2) .. controls (1.6, -1.15) and (2, -0.9) .. (pi1.south);
            
            \draw[-latex, line width=0.8pt, black] (I) -- (pi1); 
        \end{tikzpicture}
    \caption{}
    \label{fig:section-5-3-dynamic-extension-b}
    \end{subfigure}%
    \hfill
    \begin{subfigure}{0.3\textwidth}
    \centering
        \begin{tikzpicture}[scale=0.8, transform shape]
            \node[circle, draw, minimum size=1cm, line width=0.8pt] (C) at (0, 4) {$C$};
            \node[circle, draw, line width=0.8pt, minimum size=1cm, black] (I) at (0, 2) {$I$};
            \node[circle, draw, line width=0.8pt, minimum size=1cm, lightgray!50] (pi) at (0, 0) {$\pi$};

            \node[circle, draw, minimum size=1cm, line width=0.8pt, black] (C1) at (2, 4) {$C$};
            \node[circle, draw, line width=0.8pt, minimum size=1cm, black] (I1) at (2, 2) {$I$};
            \node[circle, draw, line width=0.8pt, minimum size=1cm, black] (pi1) at (2, 0) {$\pi$};
    
            \node[rectangle, draw=blue!50, fill=blue!20, line width=0.5pt, minimum size=0.3cm] (emp) at (-2, 4) {\textcolor{black}{$\varepsilon^{M}$}};
            
            \draw[-latex, line width=0.8pt, lightgray!50] (C) edge[in=45, out=-45] (pi); 
            \draw[-latex, line width=0.8pt, lightgray!50] (I) -- (pi); 
            \draw[-latex, line width=0.8pt, lightgray!50] (C) -- (pi1); 
            \draw[-latex, line width=0.8pt, lightgray!50] (I) -- (pi1); 
            \draw[-latex, line width=0.8pt, lightgray!50] (pi) -- (C1); 
            \draw[-latex, line width=0.8pt, lightgray!50] (pi) -- (I1); 
            \draw[-latex, line width=0.8pt, lightgray!50] (pi) -- (pi1); 
            \draw[-latex, line width=0.8pt, lightgray!50] (C1) edge[in=45, out=-45] (pi1); 
            \draw[-latex, line width=0.8pt, lightgray!50] (emp) edge (pi);
            \draw[-latex, line width=0.8pt, lightgray!50] (emp.south) .. controls (-2, -0.6) and (-0.4, -1.3) .. (1, -1.2) .. controls (1.6, -1.15) and (2, -0.9) .. (pi1.south);
            
            \draw[-latex, line width=0.8pt, black] (emp.north) to[out=60, in=120, looseness=1.1] (C1.north);
            \draw[-latex, line width=0.8pt, black] (emp) -- (I1);
            \draw[-latex, line width=0.8pt, black] (C) -- (I1); 
            \draw[-latex, line width=0.8pt, black] (I) -- (I1); 
            \draw[-latex, line width=0.8pt, black] (C1) -- (I1); 
            \draw[-latex, line width=0.8pt, black] (I1) -- (pi1); 
            \draw[-latex, line width=0.8pt, black] (C) -- (I); 
            \draw[-latex, line width=0.8pt, black] (C) -- (C1); 
            \draw[-latex, line width=0.8pt, black] (I) -- (C1); 
            \draw[-latex, line width=0.8pt, black] (emp) edge (C); 
            \draw[-latex, line width=0.8pt, black] (emp) edge (I); 
        \end{tikzpicture}
    \caption{}
    \label{fig:section-5-3-dynamic-extension-c}
    \end{subfigure}
    \caption{The three cases composing the investment channel at horizon $h=1$ under the dynamic extension of the investment channel \eqref{eq:sec2-invest-channel-definition}. Black nodes and edges are active in each case; greyed-out nodes and edges are not. All shocks other than the monetary policy shock $\varepsilon^M$ are omitted for clarity. The left column of nodes corresponds to the contemporaneous horizon; the right column corresponds to horizon one. The investment channel at $h=1$ is the union of the three cases.}    
    \label{fig:section-5-3-dynamic-extension}
\end{figure}

Partial-equilibrium reasoning, however, suggests that the investment channel captures two offsetting forces not visible in the DAG: one disinflationary and one inflationary. First, the contractionary monetary policy shock lowers household investment, reducing aggregate demand and placing downward pressure on marginal costs and prices. Second, lower investment leaves the capital stock below its no-shock baseline, and this shortfall raises marginal costs and prices over time. A higher depreciation rate should accelerate the second force, since lower investment together with higher depreciation depletes capital faster and thus raises marginal costs and prices sooner. Partial-equilibrium reasoning therefore predicts that, as the depreciation rate rises, the early disinflationary contribution of the investment channel weakens and shortens while its later inflationary contribution strengthens and arrives earlier. To test whether this intuition carries over to general equilibrium, we embed TCA within a counterfactual experiment that varies the depreciation rate locally around its baseline value. This is possible only because TCA quantifies transmission channels within any given equilibrium, with the different equilibria corresponding to the different settings of the depreciation rate.

Figure~\ref{fig:sec5-tca-in-counterfactual} shows the total responses of consumption, investment, and inflation to a contractionary monetary policy shock, together with the transmission effects on inflation through the consumption and investment channels, when the depreciation rate is varied locally around its baseline value of 2.5\%. Throughout, we use the dynamic extension of the investment channel defined above, with an analogous definition for the consumption channel. The total responses remain largely unchanged across the range of the depreciation parameter. Across the entire range, the consumption channel is the predominant contributor to the disinflationary effect of the shock, while the investment channel contributes negatively to inflation up to horizon two and positively from there onwards. As the depreciation rate increases, however, the investment channel becomes more disinflationary in the first periods and less inflationary in later periods, opposite to 
the partial-equilibrium prediction. 
TCA thus shows quantitatively that partial-equilibrium reasoning can be misleading about general equilibrium dynamics. A similar strategy could be used to investigate whether the consumption channel becomes stronger as the elasticity of intertemporal substitution increases; we leave this for future research.

\begin{figure}
    \centering
    \includegraphics[width=\linewidth]{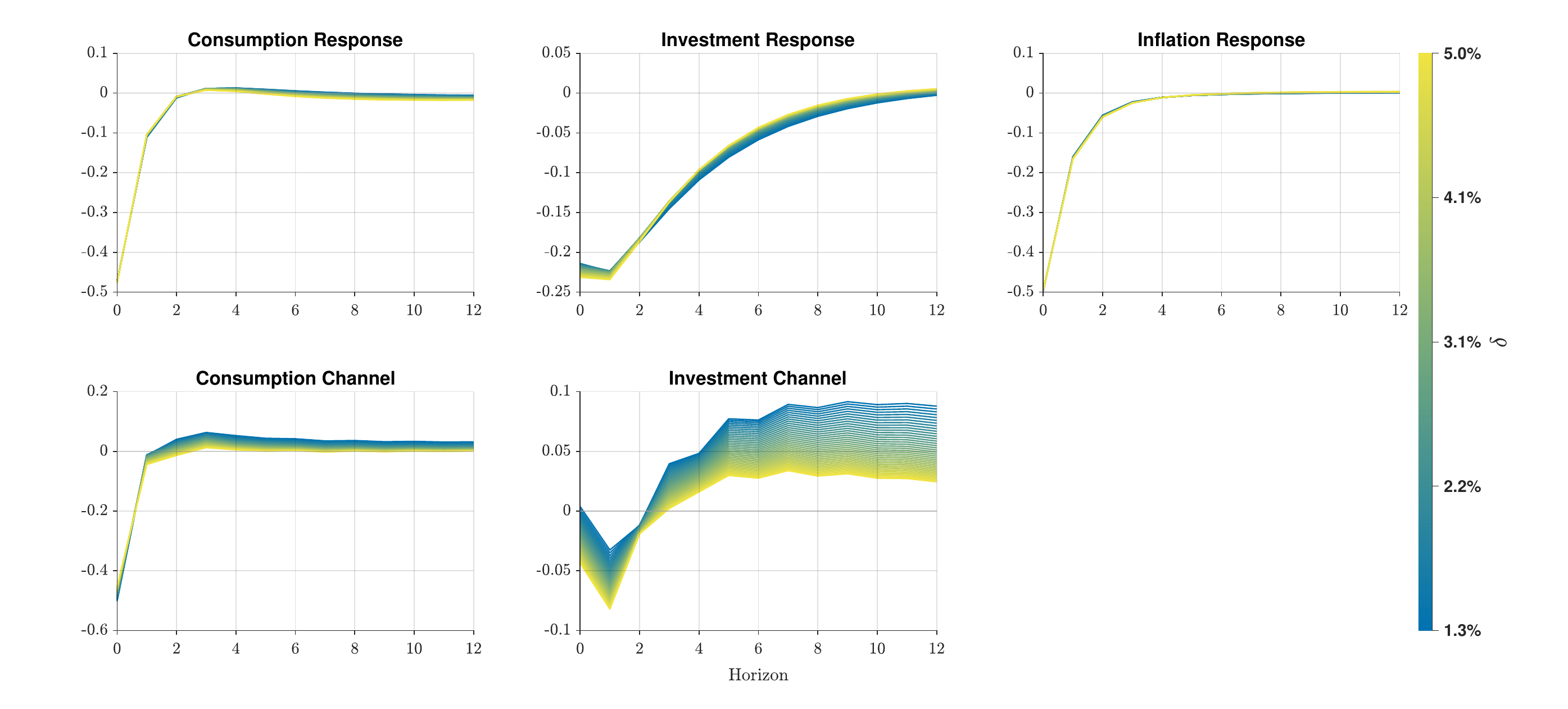}
    \caption{Total responses and transmission effects as the depreciation rate $\delta$ is varied locally around its baseline of 2.5\%. The top row shows the total responses of consumption, investment, and inflation; the bottom row shows the transmission effects on inflation through the consumption and investment channels. Colour denotes the value of $\delta$, from 1.25\% (dark) to 5\% (light).}
    \label{fig:sec5-tca-in-counterfactual}
\end{figure}

\section{Conclusion}
\label{sec:conclusion}

We develop a new coherent framework for analysing transmission channels in a large class of dynamic models by connecting them to a graphical representation. Within the graphical representation, transmission channels are defined as paths from the shock of interest to the outcome variable; transmission effects are the effects flowing along these paths. We show that this graphical representation is equivalent to a potential outcomes representation, and prove that impulse responses are sufficient to compute transmission effects in SVAR models, so that local projections can be used instead.

Current methods, including partial-equilibrium reasoning, counterfactual analysis, and direct-indirect decompositions in the HANK literature, are either qualitative, compare two equilibria rather than analysing channels within a given equilibrium, or are limited to theoretical models. Our framework provides a quantitative decomposition of IRFs into the contribution of each channel, does so within a given equilibrium, and applies to theoretical and empirical models alike. Of particular relevance for applied work is that TCA requires no identification beyond that of the structural shock of interest.

A crucial component of TCA is a precise definition of the transmission channel of interest. We show how this can be formalised through the transmission matrix, which defines the ceteris paribus ordering of the variables in equilibrium; this ordering is determined entirely by the research question and is distinct from a causal ordering. We also show that a partial ordering is sufficient in some cases, as in many of our empirical applications, where we show how TCA can shed quantitative light on the channels through which policy operates in a variety of macroeconom(etr)ic settings.

Since TCA quantifies channels within a given equilibrium, it can be applied inside counterfactual experiments to investigate how the strength of a channel varies as a structural parameter changes. We demonstrate this in Section~\ref{sec:sec5-dsge} and show that partial-equilibrium reasoning can be misleading about the general equilibrium dynamics: the quantitative estimates we obtain run contrary to what partial-equilibrium reasoning would predict. TCA can thus be used to compare channel estimates across theoretical and empirical models, and to assess how accurate partial-equilibrium intuition is about general-equilibrium dynamics.

While our proposed framework applies to a large set of commonly used dynamic macroeconomic models, including SVARs, DSGEs, and local projections, it does not currently apply to general state space models, which model dynamics through a system of latent variables. Although some state space representations, such as those used for DSGE models, can be rewritten in SVARMA form, this rewriting cannot accommodate all such models; in particular, our framework does not apply directly to dynamic factor models or to DSGEs with news shocks. In state space models, transmission channels and effects can still be defined in an equivalent way using either the graphical or the potential outcomes representation. However, this definition comes with statistical identification problems that are beyond the scope of this paper. We leave this for future research. 

Although we prove that only a single shock needs to be identified, the decomposition relies implicitly on the variances of all shocks, because edges between endogenous variables are conditional predictions and paths are interpreted as chains of conditional predictions triggered by a structural shock. This reliance is consistent with the interpretation we provide throughout and is the most practical implementation of TCA. For theoretical work, however, a decomposition that depends on the variances of all shocks may at times be undesirable. In DSGEs this can be partially addressed by setting the variances of the other shocks to a small number, as done in Section~\ref{sec:sec5-dsge}; empirically, a similar strategy is not available. Fully addressing this problem requires different theory and a different conceptual approach, and is thus beyond the scope of this paper. 

Our proposed framework relies on linearity of the dynamic model and is thus inherently subject to the potential shortcomings of linear models. This might be considered a limitation, as some evidence points to nonlinear effects of policies. For example, the effects of monetary expansions and contractions appear to be distinct \citep{coverAsymmetricEffectsPositive1992, angristSemiparametricEstimatesMonetary2018, tenreyroPushingStringUS2016}. This may, for instance, be driven by differences in the transmission of expansionary and contractionary shocks. Under the current setup, TCA cannot differentiate between the two: it analyses average transmission effects irrespective of the sign. However, most of macroeconomic analysis still relies on linear dynamic models, and TCA is applicable in a very wide class of these. In its current form, TCA can therefore already complement the vast majority of empirical macroeconomic analyses. In addition, linear estimands are often robust to nonlinearities, and moving towards nonlinearity can often do more harm than good \citep{kolesarDynamicCausalEffects2024}; whether this robustness extends to TCA is an open question. Furthermore, extending TCA to nonlinear models, particularly when allowing for interaction effects, requires a complete rethink of how to define transmission channels. While this is another exciting avenue for future research, it is clearly far outside the scope of the paper.

{
\footnotesize
\paragraph{\footnotesize Acknowledgements.} We thank Sumanta Basu, Jonathan Benchimol, Itamar Caspi, Martin Ellison, Tobias Hartl, James McNeil, Geert Mesters, Jose Luis Montiel Olea, and seminar participants at Cornell University, the Deutsche Bundesbank Vo-Seminar, EC$^2$ 2024 conference, European Seminar on Bayesian Econometrics 2023, Maastricht MILE seminar, the NESG conference, the Virtual Israel Macro Meeting (VIMM), and VTSS Workshop for Junior Researchers (2025) for their helpful comments and discussions.
Remaining errors are our own.

\paragraph{\footnotesize Funding.} Ines Wilms was financially supported by the Dutch Research Council (NWO) under grant number VI.Vidi.211.032.

\paragraph{\footnotesize Conflict of Interest.}The authors declare no conflict of interest.
}

\begingroup
\singlespacing
\printbibliography

@article{poskitt1992,
	title        = {Identification of Echelon canonical forms for vector linear processes using least squares},
	author       = {Poskitt, D.},
	year         = 1992,
	journal      = {The Annals of Statistics},
	volume       = 20,
	number       = 1,
	pages        = {195--215}
}

@article{Ramey2011,
	title        = {Identifying government spending shocks: It's all in the timing},
	author       = {Ramey, Valerie A},
	year         = 2011,
	journal      = {The Quarterly Journal of Economics},
	publisher    = {MIT Press},
	volume       = 126,
	number       = 1,
	pages        = {1--50},
	%doi         = {10.1093/qje/qjq008}
}

@book{kilianStructuralVectorAutoregressive2017,
	title        = {Structural Vector Autoregressive Analysis},
	shorttitle   = {Structural Vector Autoregressive Analysis},
	author       = {Kilian, Lutz and L\"utkepohl, Helmut},
	publisher    = {{Cambridge University Press}},
	isbn         = {978-1-107-19657-5},
	date         = 2017,
	% doi        = {10.1017/9781108164818}
}

@article{imaiIdentificationInferenceSensitivity2010,
	title        = {Identification, Inference and Sensitivity Analysis for Causal Mediation Effects},
	author       = {Imai, Kosuke and Keele, Luke and Yamamoto, Teppei},
	volume       = 25,
	number       = 1,
	pages        = {51--71},
	date         = 2010,
	journaltitle = {Statistical Science},
	% doi        = {10.1214/10-sts321}
}

@book{pearlCausalityModelsReasoning2009,
	title        = {Causality: Models, Reasoning, and Inference},
	author       = {Pearl, Judea},
	year         = 2009,
	location     = {{Cambridge}},
	publisher    = {{Cambridge University Press}},
	isbn         = {978-0-521-89560-6},
	edition      = {2nd ed.},
	% doi        = {10.1017/cbo9780511803161}
}

@article{pearlCausalMediationFormula2012,
	title        = {The Causal Mediation Formula -- A Guide to the Assessment of Pathways and Mechanisms},
	author       = {Pearl, Judea},
	year         = 2012,
	volume       = 13,
	number       = 4,
	pages        = {426--436},
	journaltitle = {Prevention Science}
}

@article{angristSemiparametricEstimatesMonetary2018,
	title        = {Semiparametric Estimates of Monetary Policy Effects: String Theory Revisited},
	author       = {Angrist, Joshua D. and Jord\`a, \`Oscar and Kuersteiner, Guido M.},
	year         = 2018,
	volume       = 36,
	number       = 3,
	pages        = {371--387},
	journaltitle = {Journal of Business \& Economic Statistics},
	% doi        = {10.1080/07350015.2016.1204919}
}

@techreport{asheshrambachanWhenCommonTime2021,
	title        = {When Do Common Time Series Estimands Have Nonparametric Causal Meaning?},
	author       = {Rambachan, Ashesh and Shephard, Neil},
	date         = 2021,
	type         = {Manuscript, Harvard University},
	% url        = {https://scholar.harvard.edu/files/shephard/files/causalmodelformacro20211012.pdf}
}

@techreport{cloyneStateDependentLocalProjections2023,
	title        = {State-Dependent Local Projections: Understanding Impulse Response Heterogeneity},
	author       = {Cloyne, James and Jord\`a, \`Oscar and Taylor, Alan M.},
	year         = 2023,
	series       = {Working Paper Series},
	number       = 30971,
	institution  = {{National Bureau of Economic Research}},
	type         = {Working Paper},
	%doi         = {10.3386/w30971}
}

@article{danielCausalMediationAnalysis2015,
	title        = {Causal Mediation Analysis with Multiple Mediators},
	author       = {Daniel, R. M. and De Stavola, B. L. and Cousens, S. N. and Vansteelandt, S.},
	volume       = 71,
	number       = 1,
	pages        = {1--14},
	date         = {2015-03},
	journaltitle = {Biometrics},
	% doi        = {10.1111/biom.12248}
}

@book{chiswellMathematicalLogic2007,
	title        = {Mathematical Logic},
	author       = {Chiswell, Ian and Hodges, Wilfrid},
	location     = {{London ; New York}},
	publisher    = {{Oxford University Press}},
	series       = {Oxford Texts in Logic},
	number       = 3,
	isbn         = {978-0-19-921562-1},
	date         = 2007,
	pagetotal    = 250
}

@book{rautenbergConciseIntroductionMathematical2010,
	title        = {A Concise Introduction to Mathematical Logic},
	author       = {Rautenberg, Wolfgang},
	location     = {{New York, NY}},
	publisher    = {{Springer}},
	series       = {Universitext},
	isbn         = {978-1-4419-1221-3},
	date         = 2010,
	edition      = {3rd ed.},
	% doi        = {10.1007/978-1-4419-1221-3}
}

@article{simsDOESMONETARYPOLICY2006,
	title        = {Does Monetary Policy Generate Recessions?},
	author       = {Sims, Christopher A. and Zha, Tao},
	year         = 2006,
	volume       = 10,
	number       = 2,
	pages        = {231--272},
	journaltitle = {Macroeconomic Dynamics},
	% doi        = {10.1017/s136510050605019x}
}

@article{plagborg-mollerLocalProjectionsVARs2021,
	title        = {Local Projections and {VARs} Estimate the Same Impulse Responses},
	author       = {{Plagborg-M{\o}ller}, Mikkel and Wolf, Christian K.},
	year         = 2021,
	journal      = {Econometrica},
	volume       = 89,
	number       = 2,
	pages        = {955--980},
	% doi        = {10.3982/ecta17813}
}

@article{ravennaVectorAutoregressionsReduced2007,
	title        = {Vector autoregressions and reduced form representations of {DSGE} models},
	author       = {Ravenna, Federico},
	year         = 2007,
	journal      = {Journal of Monetary Economics},
	volume       = 54,
	number       = 7,
	pages        = {2048--2064},
	% doi        = {10.1016/j.jmoneco.2006.09.002}
}

@article{mckayWhatCanTime2023,
	title        = {What Can Time-Series Regressions Tell Us About Policy Counterfactuals?},
	author       = {McKay, Alisdair and Wolf, Christian K.},
	year         = 2023,
	journal      = {Econometrica},
	volume       = 91,
	number       = 5,
	pages        = {1695--1725},
	% doi        = {10.3982/ecta21045}
}

@article{romerNewMeasureMonetary2004,
	title        = {A New Measure of Monetary Shocks: Derivation and Implications},
	author       = {Romer, Christina D and Romer, David H},
	year         = 2004,
	journal      = {American Economic Review},
	volume       = 94,
	number       = 4,
	pages        = {1055--1084},
	% doi        = {10.1257/0002828042002651}
}

@article{gertlerMonetaryPolicySurprises2015,
	title        = {Monetary Policy Surprises, Credit Costs, and Economic Activity},
	author       = {Gertler, Mark and Karadi, Peter},
	year         = 2015,
	journal      = {American Economic Journal: Macroeconomics},
	volume       = 7,
	number       = 1,
	pages        = {44--76},
	% doi        = {10.1257/mac.20130329}
}

@article{fernandez-villaverdeABCsDsUnderstanding2007,
	title        = {{ABCs (and Ds)} of Understanding {VARs}},
	author       = {Fern\'andez-Villaverde, Jes\'us and Rubio-Ram\'irez, Juan F and Sargent, Thomas J and Watson, Mark W},
	volume       = 97,
	number       = 3,
	pages        = {1021--1026},
	date         = 2007,
	journaltitle = {American Economic Review},
	% doi        = {10.1257/aer.97.3.1021}
}

@book{galiMonetaryPolicyInflation2015,
	title        = {Monetary policy, inflation, and the business cycle: an introduction to the {New Keynesian} framework and its applications},
	author       = {Gal\'i, Jordi},
	location     = {Princeton ; Oxford},
	publisher    = {Princeton University Press},
	isbn         = {978-0-691-16478-6},
	date         = 2015,
	edition      = {2nd ed.},
	pagetotal    = 279
}

@incollection{christianoMonetaryPolicyShocks1999,
	title        = {Monetary Policy Shocks: What have we learned and to what end?},
	author       = {Christiano, Lawrence J and Eichenbaum, Martin and Evans, Charles~L.},
	booktitle    = {Handbook of Macroeconomics},
	publisher    = {Elsevier},
	volume       = {1A},
	pages        = {65--148},
	isbn         = {978-0-444-50156-1},
	date         = 1999,
	chapter      = 2,
	editor       = {John B. Taylor and Michael Woodford}
}

@incollection{RameyMacroeconomicShocks2016,
	title        = {Macroeconomic Shocks and Their Propagation},
	author       = {Ramey, Valerie A},
	year         = 2016,
	booktitle    = {Handbook of Macroeconomics},
	publisher    = {Elsevier},
	volume       = 2,
	pages        = {71--162},
	chapter      = 2,
	editor       = {John B. Taylor and Harald Uhlig}
}

@article{nakamuraIdentification2018,
	title        = {Identification in Macroeconomics},
	author       = {Nakamura, Emi and Steinsson, J\'{o}n},
	year         = 2018,
	journal      = {Journal of Economic Perspectives},
	volume       = 32,
	number       = 3,
	pages        = {59--86},
	% doi        = {10.1257/jep.32.3.59}
}

@article{jordaEstimationAndInference2005,
	title        = {Estimation and inference of impulse responses by Local Projections},
	author       = {Jord{\`a}, {\`O}scar},
	year         = 2005,
	journal      = {American Economic Review},
	volume       = 95,
	number       = 1,
	pages        = {161--182},
	% doi        = {10.1257/0002828053828518}
}

@article{coverAsymmetricEffectsPositive1992,
	title        = {Asymmetric Effects of Positive and Negative Money-Supply Shocks},
	author       = {Cover, J. P.},
	volume       = 107,
	number       = 4,
	pages        = {1261--1282},
	date         = 1992,
	journaltitle = {The Quarterly Journal of Economics},
	% doi        = {10.2307/2118388}
}

@article{tenreyroPushingStringUS2016,
	title        = {Pushing on a String: {US} Monetary Policy Is Less Powerful in Recessions},
	author       = {Tenreyro, Silvana and Thwaites, Gregory},
	volume       = 8,
	number       = 4,
	pages        = {43--74},
	date         = 2016,
	journaltitle = {American Economic Journal: Macroeconomics},
	% doi        = {10.1257/mac.20150016}
}

@article{barnichonSufficientStatisticsApproach2023,
	title        = {A Sufficient Statistics Approach for Macro Policy},
	author       = {Barnichon, R\'egis and Mesters, Geert},
	volume       = 113,
	number       = 11,
	pages        = {2809--2845},
	date         = 2023,
	journaltitle = {American Economic Review},
	% doi        = {10.1257/aer.20220581}
}

@article{barnichonIdentifyingModernMacro2020,
	title        = {Identifying Modern Macro Equations with Old Shocks},
	author       = {Barnichon, Regis and Mesters, Geert},
	volume       = 135,
	number       = 4,
	pages        = {2255--2298},
	date         = 2020,
	journaltitle = {The Quarterly Journal of Economics},
	% doi        = {10.1093/qje/qjaa022}
}

@article{rameyGovernmentSpendingMultipliers2018,
	title        = {Government Spending Multipliers in Good Times and in Bad: Evidence from {US} Historical Data},
	author       = {Ramey, Valerie A. and Zubairy, Sarah},
	volume       = 126,
	number       = 2,
	pages        = {850--901},
	date         = 2018,
	journaltitle = {Journal of Political Economy},
	% doi        = {10.1086/696277}
}

@article{rameyCostlyCapitalReallocation1998,
	title        = {Costly capital reallocation and the effects of government spending},
	author       = {Ramey, Valerie A. and Shapiro, Matthew D.},
	volume       = 48,
	pages        = {145--194},
	date         = 1998,
	journaltitle = {Carnegie-Rochester Conference Series on Public Policy}
}

@techreport{rameyDefenseNewsShocks2016,
	title        = {Defense News Shocks, 1889 - 2015: Estimates Based on News Sources},
	author       = {Ramey, Valerie A.},
	date         = 2016,
	type         = {Manuscript, Univ. California San Diego},
	% url          = {https://econweb.ucsd.edu/~vramey/research/Defense_News_Narrative.pdf},
}

@article{wielandFinancialDampening2020,
	title        = {Financial Dampening},
	author       = {Wieland, Johannes F. and Yang, Mu-Jeung},
	volume       = 52,
	number       = 1,
	pages        = {79--113},
	date         = 2020,
	journaltitle = {Journal of Money, Credit and Banking},
	% doi        = {10.1111/jmcb.12681}
}

@article{kolesarDynamicCausalEffects2024,
	title        = {Dynamic Causal Effects in a Nonlinear World: the Good, the Bad, and the Ugly},
	author       = {Koles\'ar, Michal and Plagborg-M\o ller, Mikkel},
	year         = 2025,
	volume       = 43,
	number       = 4,
	pages        = {737--754},
	journaltitle = {Journal of Business \& Economic Statistics},
	% doi        = {10.1080/07350015.2025.2539478}
}

@article{AcaraiEtAl2023,
	title        = {Fiscal foresight and the effects of government spending: It’s all in the monetary-fiscal mix},
	author       = {Guido Ascari and Peder Beck-Friis and Anna Florio and Alessandro Gobbi},
	year         = 2023,
	journal      = {Journal of Monetary Economics},
	volume       = 134,
	pages        = {1--15},
	% doi        = {10.1016/j.jmoneco.2022.11.008}
}

@techreport{Perotti2014,
	title        = {Defense Government Spending Is Contractionary, Civilian Government Spending Is Expansionary},
	author       = {Roberto Perotti},
	year         = 2014,
	series       = {Working Paper Series},
	number       = 20179,
	institution  = {{National Bureau of Economic Research}},
	type         = {Working Paper}
}

@techreport{dograParadoxesProblemsCausal2024a,
	title        = {Paradoxes and Problems in the Causal Interpretation of Equilibrium Economics},
	author       = {Dogra, Keshav},
	year         = 2024,
	series       = {Staff Reports Number 1093},
	type         = {Staff Reports Number 1093},
	institution  = {Federal Reserve Bank of New York},
	% doi        = {10.59576/sr.1093}
}

@article{kaplanMonetaryPolicyAccording2018,
	title        = {Monetary Policy According to HANK},
	author       = {Kaplan, Greg and Moll, Benjamin and Violante, Giovanni L.},
	year         = 2018,
	journal      = {American Economic Review},
	volume       = 108,
	number       = 3,
	pages        = {697--743},
	% doi        = {10.1257/aer.20160042}
}

@article{auclertMonetaryPolicyRedistribution2019,
	title        = {Monetary Policy and the Redistribution Channel},
	author       = {Auclert, Adrien},
	year         = 2019,
	journal      = {American Economic Review},
	volume       = 109,
	number       = 6,
	pages        = {2333--2367},
	% doi        = {10.1257/aer.20160137}
}

@article{bilbiieNewKeynesianCross2020,
	title        = {The New Keynesian cross},
	author       = {Bilbiie, Florin O.},
	year         = 2020,
	journal      = {Journal of Monetary Economics},
	volume       = 114,
	pages        = {90--108},
	% doi        = {10.1016/j.jmoneco.2019.03.003}
}

@book{woodfordInterestPricesFoundations2004,
	title        = {Interest and Prices: Foundations of a Theory of Monetary Policy},
	author       = {Woodford, Michael},
	location     = {Princeton},
	publisher    = {Princeton University Press},
	isbn         = {978-0-691-01049-6 978-1-4008-3016-9},
	date         = 2004
}

@incollection{bernankeChapter21Financial1999,
	title        = {The Financial Accelerator in a Quantitative Business Cycle Framework},
	author       = {Bernanke, Ben S. and Gertler, Mark and Gilchrist, Simon},
	year         = 1999,
	booktitle    = {Handbook of Macroeconomics},
	publisher    = {Elsevier},
	volume       = {1C},
	pages        = {1341--1393},
	editor       = {Taylor, John B. and Woodford, Michael},
	chapter      = 21,
	% doi        = {10.1016/S1574-0048(99)10034-X}
}

@article{christianoNominalRigiditiesDynamic2005,
	title        = {Nominal Rigidities and the Dynamic Effects of a Shock to Monetary Policy},
	author       = {Christiano, Lawrence~J. and Eichenbaum, Martin and Evans, Charles~L.},
	year         = 2005,
	volume       = 113,
	number       = 1,
	pages        = {1--45},
	journaltitle = {Journal of Political Economy},
	% doi        = {10.1086/426038}
}

@article{iacovielloHousePricesBorrowing2005,
	title        = {House Prices, Borrowing Constraints, and Monetary Policy in the Business Cycle},
	author       = {Iacoviello, Matteo},
	year         = 2005,
	volume       = 95,
	number       = 3,
	pages        = {739--764},
	journaltitle = {American Economic Review},
	% doi        = {10.1257/0002828054201477}
}

@article{wrightCorrelationCausation1921,
	title        = {Correlation and Causation},
	author       = {Wright, Sewell},
	volume       = 20,
	pages        = {557--585},
	date         = 1921,
	journaltitle = {Journal of Agricultural Research}
}

@article{wrightMethodPathCoefficients1934,
	title        = {The Method of Path Coefficients},
	author       = {Wright, Sewall},
	volume       = 5,
	number       = 3,
	pages        = {161--215},
	date         = 1934,
	journaltitle = {The Annals of Mathematical Statistics},
	% publisher  = {Institute of Mathematical Statistics},
	% url        = {https://www.jstor.org/stable/2957502}
}

@article{gertlerModelUnconventionalMonetary2011,
	title        = {A model of unconventional monetary policy},
	author       = {Gertler, Mark and Karadi, Peter},
	year         = 2011,
	volume       = 58,
	number       = 1,
	pages        = {17--34},
	journaltitle = {Journal of Monetary Economics},
	% doi        = {10.1016/j.jmoneco.2010.10.004}
}
\endgroup

\clearpage
\numberwithin{lemma}{section}
\numberwithin{equation}{section}
\numberwithin{figure}{section}
\numberwithin{table}{section}

\begin{appendices}
\section{Proofs of Main Theorems}
\label{appendix:proofs}
Section \ref{appendix:Boolean} first introduces some notation and preliminary results based on Boolean algebra needed in the proofs of the main theorems. 
The proofs of Theorems \ref{thm:framework-equivalence} to \ref{thm:irf-sufficiency} and Lemmas \ref{lem:invariance-edges} to \ref{lemma:invariance-cholesky}
are then presented in Section \ref{appendix:proofs-main-theorems}. 

\subsection{Boolean Algebra for Transmission Channels} \label{appendix:Boolean}
Definition \ref{def:transmission-channel} defines transmission channels as a sub-collection of paths. In practice, it is more convenient to define transmission channels by statements about which variables lie on the paths and which do not. Boolean algebra plays an important role in translating such statements of transmission channels into sub-collections of paths. Here we introduce the basics of Boolean algebra\footnote{For an in-depth overview of Boolean algebra and mathematical logic, see \citet{chiswellMathematicalLogic2007, rautenbergConciseIntroductionMathematical2010}.} needed in the context of TCA. The main result of this section is Lemma \ref{lem:properties} which provides an algebra for transmission channels and is key to proving Theorems \ref{thm:total-effect-and-decomposition}, \ref{thm:irf-sufficiency}, and \ref{thm:single-shock}. 

Let $b$ be a Boolean formula that is true if and only if a path satisfies it, and false otherwise. We denote $\calP_{u, v}$ the collection of all paths in $\calG(\bB, \bOmega)$ from variable/shock $u$ to variable $v$. Then $P_{u,v}(b)\subseteq \calP_{u,v}$ represents the sub-collection of paths of $\calP_{u,v}$ that satisfy $b$. By Definition \ref{def:transmission-channel}, $P_{u,v}(b)$ is a transmission channel if $u$ is a shock and at least one path does not satisfy $b$. 

A Boolean formula can consist of the following elements:
\begin{enumerate}
    \item $u$: denotes that variable $u$ must be on the paths.
    \item $b' \land u$: denotes that the paths must satisfy the Boolean formula $b'$, and variable $u$ must be on the paths.
    \item $b' \land b''$: denotes that both Boolean formulas $b'$ and $b''$ must be satisfied by the paths.
    \item $b' \lor u$: denotes that the paths must either satisfy the boolean formula $b'$ or $u$ must be on the paths, or both.
    \item $b' \lor b''$: denotes that the paths must either satisfy the boolean formula $b'$ or $b''$ or both.
    \item $\neg u$: denotes that variable $u$ cannot be on the paths. 
    \item $\neg b'$: denotes that the Boolean formula $b$ cannot be satisfied by the paths.
\end{enumerate}

Denote the total path-specific effect of the collection $P_{u,v}(b)$ by
\begin{equation*}
    \calQ_{\xi}[P_{u,v}(b)] = \xi\sum_{p \in P_{u,v}(b)}\prod_{(k\to l) \in p}\omega_{kl},
\end{equation*}
where $\xi$ is the shock size, and $\omega_{k, l}$ the path coefficient of the edge going from node $k$ to node $l$ on path $p$, in line with Definition \ref{def:transmission-effect}.

Employing standard set theory arguments and Boolean algebra, several properties of the total path-specific effect can be demonstrated. These properties are collected in the next lemma. The proofs can be found in Appendix \ref{appendix:proof-algebra}.
\begin{lemma} \label{lem:properties}
    Let $\by_t$ be generated by \eqref{eq:general-model} and Assumptions \ref{assump:general-model-structural-shocks} through \ref{assump:equal-treatment} be satisfied. Let $\calG(\bB, \bOmega)$ be the graph induced by \eqref{eq:general-model-xbx}. Denote with $\calP_{u,v}$ the set of all paths connecting variable or shock $u$ to variable $v$ in $\calG(\bB, \bOmega)$ and by $P_{u,v}(b)\subseteq \calP_{u,v}$ the subset of paths satisfying a Boolean formula $b$. Further, denote $\calQ_{\xi}[P_{u,v}(b)]$ the total path-specific effect of the paths in $P_{u,v}(b)$. Let $b'$ be another Boolean formula. Then, the following properties hold:

\begin{enumerate}[label=(\roman*), ref=\ref{lem:properties}(\roman*)]
    \item\label{enum:prop-1} Let $b=x_k$. Then $\calQ_{\xi}[P_{\varepsilon_i, x_j}(b)]=\calQ_{\xi}[\calP_{\varepsilon_i, x_k}]\calQ_{1}[\calP_{x_k, x_j}]$ with $k\neq j$. 
    
    \item\label{enum:prop-2} Let $i_1 < i_2 < ... < i_k$ and $i_k \neq  j$. Let $b = \wedge_{m=1}^k x_{i_m}$. Then $$\calQ_{\xi}[P_{\varepsilon_i,x_j}(b)] = \calQ_{\xi}[\calP_{\varepsilon_i, x_{i_1}}]\calQ_{1}[\calP_{x_{i_1}, x_{i_2}}]\dots\calQ_{1}[\calP_{x_{i_k}, x_j}].$$ 
    
    \item\label{enum:prop-3} Let $P_{\varepsilon_i,x_j}\subset \mathcal{P}_{\varepsilon_i,x_j}$ and $P_{\varepsilon_i,x_j}' \subset \mathcal{P}_{\varepsilon_i,x_j}$ be two disjoint sets of paths. Then $\calQ_{\xi}[P_{\varepsilon_i,x_j} \cup P_{\varepsilon_i,x_j}'] = \calQ_{\xi}[P_{\varepsilon_i,x_j}] + \calQ_{\xi}[P_{\varepsilon_i,x_j}']$. 
    
    \item\label{enum:prop-4} Let $P_{\varepsilon_i,x_j}\subset \mathcal{P}_{\varepsilon_i,x_j}$ and $P_{\varepsilon_i,x_j}' \subset \mathcal{P}_{\varepsilon_i,x_j}$. Then $\calQ_{\xi}[P_{\varepsilon_i,x_j}\backslash P_{\varepsilon_i,x_j}'] = \calQ_{\xi}[P_{\varepsilon_i,x_j}] - \calQ_{\xi}[P_{\varepsilon_i,x_j} \cap P_{\varepsilon_i,x_j}']$. 
    
    \item\label{enum:prop-5} $P_{\varepsilon_i,x_j}(b) \cap P_{\varepsilon_i,x_j}(b') = P_{\varepsilon_i, x_j}(b\wedge b')$. 
    
    \item\label{enum:prop-6} 
    $P_{\varepsilon_i,x_j}(b \vee b') = P_{\varepsilon_i,x_j}(b)\cup (P_{\varepsilon_i,x_j}(b')\backslash P_{\varepsilon_i,x_j}(b \wedge b'))$. 
    
    \item\label{enum:prop-7} 
    $P_{\varepsilon_i,x_j}(b \wedge \neg b') = P_{\varepsilon_i,x_j}(b)\backslash P_{\varepsilon_i,x_j}(b \wedge b')$. 
    
    \item\label{enum:prop-8} Let $b=\neg x_k$. 
    Then $\calQ_{\xi}[P_{\varepsilon_i,x_j}(b)] = \calQ_{\xi}[\mathcal{P}_{\varepsilon_i,x_j}] - \calQ_{\xi}[P_{\varepsilon_i,x_j}(x_k)]$ with $k\neq j$. 
\end{enumerate}
\end{lemma}

\subsection{Proofs of Main Results} \label{appendix:proofs-main-theorems}
\begin{proof}[\textbf{Proof of Theorem \ref{thm:framework-equivalence}}]
Under the conditions of Theorem \ref{thm:framework-equivalence},
$\calC_{\varepsilon_i, x_j}(\bepsilon^{(j)}) = \E[x_j^*(\bepsilon^{(j)}) - x_j^*(\bm 0)]$
is the effect of $\varepsilon_i$ on $x_j$
which can be re-written as 
\begin{equation*}
\begin{split}
\calC_{\varepsilon_i, x_j}(\bepsilon^{(j)}) &= \sum_{k_1=1}^{j-1}\left\{\bB_{j,k_1}\E[x^*_{k_1}(\bepsilon^{(j\cdot k_1)}) - x^*_{k_1}(\bm 0_{\eta(k_1)})]\right\} + \bOmega_{j,i} \bepsilon^{(j)}_{\eta(j)} \\
&= \sum_{k_1=1}^{j-1}\left\{\bB_{j,k_1}\calC_{\varepsilon_i, x_{k_1}}(\bepsilon^{(j\cdot k_1)})\right\} + \bOmega_{j,i} \bepsilon^{(j)}_{\eta(j)} \\
&= \bOmega_{j,i} \bepsilon^{(j)}_{\eta(j)} + \sum_{k_1=1}^{j-1}\bB_{j, k_1}\bOmega_{{k_1},i}\bepsilon^{(j\cdot k_1)}_{\eta(k_1)} + \sum_{k_1=1}^{j-1}\sum_{k_2=1}^{k_1-1}\left\{\bB_{j,k_1}\bB_{k_1,k_2}\calC_{\varepsilon_i, x_{k_2}}(\bepsilon^{(j\cdot k_1\cdot k_2)})\right\}.
\end{split}
\end{equation*}

Denote for $l=0, ..., j-1$,
\begin{equation*}
S(l) = \begin{cases}
\bOmega_{j,i} \bepsilon^{(j)}_{\eta(j)} & l=0 \\
\sum_{k_1=1}^{j-1}\sum_{k_2=1}^{k_1-1}\ldots\sum_{k_l=1}^{k_{l-1}-1}\bB_{j,k_1}\bB_{k_1,k_2}\ldots \bB_{k_{l-1}, k_l}\bOmega_{{k_l},i}\bepsilon^{(j\cdot k_1\cdot\ldots \cdot k_{l-1}\cdot k_l)}_{\eta(k_l)} & l \geq 1.
\end{cases}
\end{equation*}
It then follows that
\begin{equation}
    \calC_{\varepsilon_i, x_j}(\bepsilon^{(j)}) = \sum_{l=0}^{j-1}S(l).
    \label{eq:causal-simplify-3}
\end{equation}

The following observations now lead to the equivalence proof:
\begin{enumerate}
    \item $S(l)$ enumerates all paths that include $l$ intermediate variables. 
    \item Due to the structure of $\bB$, there never exists an edge from node $x_r \to x_s$ if $r > s$. Thus, there exists no path from $\epsilon_{i}$ to $x_j$ that has more than $j-1$ intermediate variables. 
    \item This last observation implies that the right-hand-side of equation \eqref{eq:causal-simplify-3} enumerates all paths that exist in $\calG(\bB, \bOmega)$ from $\epsilon_{i}$ to $x_j$. 
    \item The product $\bB_{j k_1}\bB_{k_1 k_2}...\bB_{k_{l-1}k_l}\bOmega_{{k_l},i}$ in $S(l)$ is the product of path coefficients along the path $\epsilon_{i} \to x_{k_l} \to \ldots \to x_{k_1} \to x_j$ and thus equals the path-specific effect of this path. 
    \item $\epsilon^{(j \cdot k_1 \cdot \ldots \cdot k_{l-1}\cdot k_l)}_{\eta(k_l)}$ acts as a path selector. If $\epsilon^{(j \cdot k_1 \cdot \ldots \cdot k_{l-1}\cdot k_l)}_{\eta(k_l)}=0$, such that the path is not selected, then the term $\bB_{j k_1}\bB_{k_1 k_2}...\bB_{k_{l-1}k_l}\bOmega_{{k_l},i}$ drops out of the right-hand-side of equation \eqref{eq:causal-simplify-3}. 
\end{enumerate}

According to observations 1 to 3, the potential outcomes assignment vector has as many entries as there exist paths in the graph $\calG(\bB, \bOmega)$. Next, take a set of paths $P_{\varepsilon_i, x_j}$ as given. By observations 4 and 5, we can set the elements in the potential outcomes assignment vector $\bepsilon^{(j)}$ selecting the paths in $P_{\varepsilon_i, x_j}$ equal to $\xi$ and all other elements to zero. This then implies $\calQ_\xi(P_{\varepsilon_i, x_j})=\calC_{\varepsilon_i, x_j}(\bepsilon^{(j)})$. Finally, take a potential outcomes assignment vector $\bepsilon^{(j)}$ with $\epsilon^{(j)}_k\in \{\xi, 0\}$ as given. Then, by observations 3 and 5, this potential outcomes assignment vector selects a set of paths $P_{\varepsilon_i, x_j}$ such that $\calC_{\varepsilon_i, x_j}(\bepsilon^{(j)})=\calQ_\xi(P_{\varepsilon_i, x_j})$.
\end{proof}

\begin{proof}[\textbf{Proof of Theorem \ref{thm2-part1}}]
Recall from Theorem \ref{thm:framework-equivalence} that $\calC_{\varepsilon_i, x_j}(\bepsilon^{(j)})=\calQ_\xi[P_{\varepsilon_i, x_j}]$. Taking $\bepsilon^{(j)} = \xi \bm 1_{\eta(j)}$ immediately results in $P_{\varepsilon_i, x_j} =\calP_{\varepsilon_i, x_j}$. This, in turn, implies by Theorem \ref{thm:framework-equivalence} that $\calC_{\varepsilon_i, x_j}(\xi\bm 1_{\eta(j)}) = \calQ_\xi(\calP_{\varepsilon_i, x_j})$. Next, note that $\calC_{\varepsilon_i, x_j}(\xi\bm 1_{\eta(j)}) = \E[x_j | \varepsilon_i=\xi] - \E[x_j | \varepsilon_i=0] = \xi \bPhi_{j,i}$, which completes the proof. 
\end{proof}
\begin{proof}[\textbf{Proof of Theorem \ref{thm2-part2}}]
As $P_{\varepsilon_i, x_j}^{(1)}, \dots, P_{\varepsilon_i, x_j}^{(k)}$ is a partition, $P_{\varepsilon_i, x_j}^{(1)}, \cup_{r=2}^k P_{\varepsilon_i, x_j}^{(r)}$ is a partition as well. By the first part we therefore have 
\begin{equation*}
    \xi\bPhi_{j,i} = \calQ_\xi(\calP_{\varepsilon_i, x_j}) = \calQ_\xi[P_{\varepsilon_i, x_j}^{(1)} \cup (\cup_{r=2}^k P_{\varepsilon_i, x_j}^{(r)})].
\end{equation*}
Since $P_{\varepsilon_i, x_j}^{(1)}$ and $\cup_{r=2}^k P_{\varepsilon_i, x_j}^{(r)}$ are disjoint, Lemma \ref{enum:prop-3} implies
\begin{equation*}
    \calQ_\xi[P_{\varepsilon_i, x_j}^{(1)} \cup (\cup_{r=2}^k P_{\varepsilon_i, x_j}^{(r)})] = \calQ_\xi[P_{\varepsilon_i, x_j}^{(1)}] + \calQ_\xi[\cup_{r=2}^k P_{\varepsilon_i, x_j}^{(r)}].
\end{equation*}
Applying the same logic recursively to $\calQ_\xi[\cup_{r=2}^k P_{\varepsilon_i, x_j}^{(r)}]$ results in 
\begin{equation*}
    \xi\bPhi_{j,i} = \sum_{r=1}^k \calQ_\xi(P_{\varepsilon_i, x_j}^{(r)}),
\end{equation*}
completing the proof of the second part. 
\end{proof}
\begin{proof}[\textbf{Proof of Theorem \ref{thm2-part3}}]
Each element in the potential outcomes assignment vector $\bepsilon^{(j)}$ selects whether a specific path is active or not. Thus, for each collection $P_{\varepsilon_i, x_j}^{(r)}$ there exists a potential outcomes assignment vector $\be_r$ that has a $\xi$ at those elements selecting the paths present in $P_{\varepsilon_i, x_j}^{(r)}$ and $0$ at all other elements. Since no path can be present in two collections of the partition, $\sum_{r=1}^k \be_r = \xi \bm 1_{\eta(j)}$. Applying Theorems \ref{thm:framework-equivalence} and \ref{thm2-part2} yields $\sum_{r=1}^k \calC_{\varepsilon_i, x_j}(\be_m) = \xi\bPhi_{j,i}$.
\end{proof}

\begin{proof}[\textbf{Proof of Theorem \ref{thm3-part1}}]
Under the assumption that $\bPsi_j=\bO$ for all $j=1,\ldots,q$ and after substituting $\bQ'\bA_i^*=\bar\bA_i$ and $\bgamma_t = \bQ'\bvarepsilon_t$, equation \eqref{eq:general-model-ql} can be written as 
\begin{equation}
    \bL\by_t^* = \sum_{i=1}^\ell \bar\bA_i\by_{t-i}^* + \bgamma_t.
    \label{eq:appendix-proof-lemma-irf-suff-1}
\end{equation}
With $\bL$ lower-triangular, this system is obtained from the reduced-form system using a Cholesky-orthogonalisation scheme with the ordering determined by $\bT$. 

As in Appendix \ref{appendix:rewrite}, the system \eqref{eq:appendix-proof-lemma-irf-suff-1} can be written as
\begin{equation}
\bx = \tilde\bB\bx + \tilde\bOmega\bgamma, 
\label{eq:appendix-proof-lemma-irf-suff-xbx}
\end{equation}
where $\bgamma=(\bgamma_t', \dots, \bgamma_{t+h}')'$, $\bx=(\by_t^{*'}, \dots, \by_{t+h}^{*'})'$ and $\tilde\bB=\bB$, and $\tilde\bOmega = \bI_{h+1} \otimes \bD$ with $\bD=\text{diag}(\bL)^{-1}$, where $\text{diag}(\bX)$ is a diagonal matrix of the diagonal of $\bX$. 

Since $\bB = \tilde\bB$, the system \eqref{eq:appendix-proof-lemma-irf-suff-xbx} induces a graph $\calG(\tilde\bB, \tilde\bOmega)$ that has the same path-coefficients for any edge $x_j \to x_i$ as the graph $\calG(\bB, \bOmega)$.

Because the shocks in $\bgamma$ are uncorrelated, we can use the graph $\calG(\tilde\bB, \tilde\bOmega)$ to investigate the effect of the shock $\gamma_i$ on $x_j$ along all paths in the graph, $\calP_{\gamma_i, x_j}$. This effect is given by the total path-specific effect
\begin{equation}
    \sum_{p \in \calP_{\gamma_i, x_j}}\prod_{u\to v \in p}\tilde\omega_{u,v}
    \label{eq:appendix-proof-lemma-irf-suff-2},
\end{equation}
where $\tilde\omega_{u, v}$ is the path coefficient of the edge connecting $u$ to $v$. 

Since $\tilde\bOmega$ is diagonal, any path in $\calP_{\gamma_i, x_j}$ must always first go into $x_i$ before it can go through any other variables. Thus equation \eqref{eq:appendix-proof-lemma-irf-suff-2} can also be written as 
\begin{equation}
   \tilde\omega_{\gamma_i, x_i}\sum_{p \in \calP_{x_i, x_j}}\prod_{u\to v \in p}\tilde\omega_{u, v}. 
   \label{eq:appendix-proof-lemma-irf-suff-3}
\end{equation}

Because, $\tilde\bB=\bB$, all path coefficients along paths connecting $x_k$ to $x_l$ for all $k, l$, are the same as in the graph $\calG(\bB, \bOmega)$, $\tilde\omega_{u,v}=\omega_{u,v}$ for all $u\neq \gamma_i$. Thus, \eqref{eq:appendix-proof-lemma-irf-suff-3} can also be written as
\begin{equation}
   \tilde\omega_{\gamma_i, x_i}\sum_{p \in \calP_{x_i, x_j}}\prod_{u\to v \in p}\omega_{u, v}.
   \label{eq:appendix-proof-lemma-irf-suff-4}
\end{equation}

Letting $\gamma_i$ take the role of $\varepsilon_i$ in Theorem \ref{thm:total-effect-and-decomposition}, one can show that the total path-specific effect of $\calP_{\gamma_i, x_j}$ in the graph $\calG(\tilde\bB, \tilde\bOmega)$ equals the total effect of $\gamma_i$ on $x_j$, which is given by the impulse response matrix
\begin{equation}
    \tilde\bPhi_{j,i} = [(\bI - \tilde\bB)^{-1}\tilde\bOmega]_{j,i} = [(\bI -\bB)^{-1}\tilde\bOmega]_{j,i}.
   \label{eq:appendix-proof-lemma-irf-suff-5}
\end{equation} 

Similarly, the effect of $\gamma_i$ on $x_i$, which consists of only the direct edge $\gamma_i \to x_i$ equalling $\tilde\omega_{\gamma_i, x_i}$, is a total effect equal to the impulse response of $\gamma_i$ on $x_i$,
\begin{equation}
    \tilde\bPhi_{i, i} = [(\bI - \bB)^{-1}\tilde\Omega]_{i, i}.
   \label{eq:appendix-proof-lemma-irf-suff-6}
\end{equation}
By combining \eqref{eq:appendix-proof-lemma-irf-suff-4}, \eqref{eq:appendix-proof-lemma-irf-suff-5} and \eqref{eq:appendix-proof-lemma-irf-suff-6}, 
\begin{equation}
   \tilde\bPhi_{i, i}\sum_{p \in \calP_{x_i, x_j}}\prod_{u\to v \in p}\omega_{u, v} = \tilde\bPhi_{j, i}.
   \label{eq:appendix-proof-lemma-irf-suff-7}
\end{equation}
Rearranging,  \eqref{eq:appendix-proof-lemma-irf-suff-7} gives the final result,
\begin{equation*}
   \calQ_{1}(P_{x_i, x_j}) = \sum_{p \in \calP_{x_i, x_j}}\prod_{u\to v \in p}\omega_{u, v} = \tilde\bPhi_{j, i}\tilde\bPhi_{i, i}^{-1}.
\end{equation*}

To complete the proof, observe that since \eqref{eq:appendix-proof-lemma-irf-suff-1} is the system obtained from a Cholesky-orthogonalisation scheme with the ordering determined by $\bT$, the impulse response matrix $\tilde\bPhi=(\bI - \bB)^{-1}\tilde\bOmega$ is also given by a Cholesky-orthogonalisation scheme with the ordering determined by $\bT$.
\end{proof}

\begin{proof}[\textbf{Proof of Theorem \ref{thm3-part2}}]
Let $b$ be a Boolean formula. We say $b$ is in disjunctive normal form (DNF) if $b$ can be written as 
$b = b_1 \vee b_2 \vee ... \vee b_n,$
where $b_i$ includes only conjunctions, $\wedge$, and negations, $\neg$. Further, $b$ can be written in $\text{DNF}_n$ if $b$ can be written in disjunctive normal form involving only $n-1$ disjunctions, $\vee$. 

According to Theorem 3.8.8 in \citet{chiswellMathematicalLogic2007} and Theorem 2.1 in \citet{rautenbergConciseIntroductionMathematical2010}, every Boolean formula $b$ can be written in DNF. Thus, for every Boolean formula there exists a $k\in \mathbb{N}$ such that $b$ can be written in $\text{DNF}_k$. Additionally, for any collection of paths $P_{\varepsilon_i, x_j}\subseteq \calP_{\varepsilon_i, x_j}$ there exists a boolean formula $b$ such that $P_{\varepsilon_i, x_j}(b) = P_{\varepsilon_i, x_j}$. To prove part two of Theorem \ref{thm:irf-sufficiency} it suffices to show that for any $DNF_k$ with $k\in \mathbb{N}$, there exists a function $f: \R^{(h+1)K\times (h+1)K}\times \R^{(h+1)k\times (h+1)k} \to \R$ such that $\calQ_\xi(P_{\varepsilon_i, x_j}(b)) = f(\bPhi, \tilde\bPhi)$. 

We prove this by induction. 
Let $k=1$ such that the Boolean formula $b$ consists of only conjunctions and negations; the Boolean formula $b$ is $\text{DNF}_1$. If $b$ has no negations, then Lemma \ref{enum:prop-2} with Theorems \ref{thm:total-effect-and-decomposition} and \ref{thm3-part1} imply that there exists a function $f$ such that $\calQ_\xi(P_{\varepsilon_i, x_j}(b)) = f(\bPhi, \tilde\bPhi)$. If $b$ has a single negation then we may write $b = \tilde b \land \neg x_r$ for some variable $x_r$. Then \begin{equation*}
    \begin{split}
        \calQ_{\xi}[P_{\varepsilon_i, x_j}(b)] &= \calQ_{\xi}[P_{\varepsilon_i, x_j}(\tilde{b} \land \neg x_r)] 
        = \calQ_{\xi}[P_{\varepsilon_i, x_j}(\tilde{b})\setminus P_{\varepsilon_i, x_j}(\tilde{b} \land x_r)] \\
        &= \calQ_{\xi}[P_{\varepsilon_i, x_j}(\tilde{b})] - \calQ_{\xi}[P_{\varepsilon_i, x_j}(\tilde{b} \land x_r)].
    \end{split}
\end{equation*}
The second line above follows by Lemma \ref{enum:prop-7} while the last line follows by Lemma \eqref{enum:prop-4} and \eqref{enum:prop-5}. Since $\tilde b$ and $\tilde b \land x_r$ are Boolean formulas consisting of only conjunctions, there exist functions $f_1$ and $f_2$ such that $\calQ_\xi(P_{\varepsilon_i, x_j}(\tilde b)) = f_1(\bPhi, \tilde\bPhi)$ and $\calQ_\xi(P_{\varepsilon_i, x_j}(\tilde b \land x_r)) = f_2(\bPhi, \tilde\bPhi)$. Thus, let $f(\bPhi, \tilde\bPhi) = f_1(\bPhi, \tilde\bPhi) - f_2(\bPhi, \tilde\bPhi)$. Then $\calQ_\xi(P_{\varepsilon_i, x_j}(b))=f(\bPhi, \tilde\bPhi)$.

Finally, if $b$ consists of $m$ negations, then we may write $b=\bar b \land \neg x_r$ for some variable $x_r$, where $\bar b$ consists of $m-1$ negations. We can thus again write 
\begin{equation*}
    \begin{split}
        \calQ_{\xi}[P_{\varepsilon_i, x_j}(b)] &= \calQ_{\xi}[P_{\varepsilon_i, x_j}(\bar{b} \land \neg x_r)] 
        = \calQ_{\xi}[P_{\varepsilon_i, x_j}(\bar{b})\setminus P_{\varepsilon_i, x_j}(\bar{b} \land x_r)] \\
        &= \calQ_{\xi}[P_{\varepsilon_i, x_j}(\bar{b})] - \calQ_{\xi}[P_{\varepsilon_i, x_j}(\bar{b} \land x_r)].
    \end{split}
\end{equation*}
By the induction hypothesis, there exists a function $f_1$ sucht that $\calQ_\xi(P_{\varepsilon_i, x_j}(\bar b))=f_1(\bPhi, \tilde\bPhi)$. Additionally, since $\bar b \land x_r$ consists of a only $m-1$ negations, the induction hypothesis also implies that there exists a function $f_2$ such that $\calQ_\xi(P_{\varepsilon_i, x_j}(\bar b \land x_r))=f_2(\bPhi, \tilde\bPhi)$. Thus, let $f(\bPhi, \tilde\bPhi) = f_1(\bPhi, \tilde\bPhi) - f_2(\bPhi, \tilde\bPhi)$. Then $\calQ_\xi(P_{\varepsilon_i, x_j}(b))=f(\bPhi, \tilde\bPhi)$. Thus, if $b$ is $\text{DNF}_1$ with $m\in\mathbb{N}$ negations, then there exists a function $f$ sucht that $\calQ_\xi(P_{\varepsilon_i, x_j}(b)) = f(\bPhi, \tilde\bPhi)$. 

Next, suppose there exists a function $f$ such that $\calQ_\xi(P_{\varepsilon_i, x_j}(\bar b)) = f(\bPhi, \tilde\bPhi)$ if $\bar b$ is $\text{DNF}_k$ and let the Boolean formula $b$ be $\text{DNF}_{k+1}$. We can then write $b = b_1 \lor \bar b$ where $b_1$ is $\text{DNF}_1$. By Lemma \ref{enum:prop-6}, it follows that
\begin{equation*}
    P_{\varepsilon_i, x_j}(b) = P_{\varepsilon_i, x_j}(b_1) \cup (P_{\varepsilon_i, x_j}(\bar b)\backslash P_{\varepsilon_i, x_j}(b_1 \land \bar b)).
\end{equation*}
Note that $P_{\varepsilon_i, x_j}(b_1)$ and $P_{\varepsilon_i, x_j}(\bar b)\backslash P_{\varepsilon_i, x_j}(b_1 \land \bar b)$ are disjoint. Thus, by Lemma \ref{enum:prop-3} and \ref{enum:prop-4},
\begin{equation*}
    \calQ_\xi(P_{\varepsilon_i, x_j}(b)) = 
    \calQ_\xi(P_{\varepsilon_i, x_j}(\underbrace{b_1}_{\text{DNF}_1})) + 
    \calQ_\xi(P_{\varepsilon_i, x_j}(\underbrace{\bar b}_{\text{DNF}_k})) - 
    \calQ_\xi(P_{\varepsilon_i, x_j}(\underbrace{b_1 \land \bar b}_{\text{DNF}_k})).
\end{equation*}
Since there exists a function $f_1$ such that $\calQ_\xi(P_{\varepsilon_i, x_j}(b_1))=f_1(\bPhi, \tilde\bPhi)$, and the induction hypothesis implies that there exist functions $f_2$ and $f_3$ such that $\calQ_\xi(P_{\varepsilon_i, x_j}(\bar b))=f_2(\bPhi, \tilde\bPhi)$ and $\calQ_\xi(P_{\varepsilon_i, x_j}(b_1 \land \bar b))=f_3(\bPhi, \tilde\bPhi)$, there exists a function $f(\bPhi, \tilde\bPhi)=f_1(\bPhi, \tilde\bPhi) + f_2(\bPhi, \tilde\bPhi) - f_3(\bPhi, \tilde\bPhi)$ such that $\calQ_\xi(P_{\varepsilon_i, x_j}(b))=f(\bPhi, \tilde\bPhi)$.
We can then conclude that, for any $n\in \mathbb{N}$ and any Boolean formula $b$ being $\text{DNF}_n$, there exists a function $f$ such that $\calQ_\xi(P_{\varepsilon_i, x_j}(b))=f(\bPhi, \tilde\bPhi)$.
This concludes the proof.
\end{proof}

\begin{proof}[\textbf{Proof of Theorem \ref{thm:single-shock}}]
After substituting $\bQ'\bA_i^*=\bar\bA_i$, $\bQ'\bPsi_j\bQ=\bar\bPsi_j$, and $\bgamma_t = \bQ'\bvarepsilon_t$, equation \eqref{eq:general-model-ql} can be written as 
\begin{equation}
    \bL\by_t^* = \sum_{i=1}^l \bar\bA_i\by_{t-i}^* + \bgamma_t + \sum_{j=1}^q\bar\bPsi_j\bgamma_{t-j}.
    \label{eq:proof-theorem-4-cholesky-system}
\end{equation}
With $\bL$ lower-triangular, this system is obtained from the reduced form system using a Cholesky-othogonalisation scheme with the ordering determined by $\bT$. 

As in Appendix~\ref{appendix:rewrite}, the system \eqref{eq:proof-theorem-4-cholesky-system} can be written as $\bx = \tilde\bB\bx + \tilde\bOmega\bgamma$, where $\bgamma = (\bgamma_t', \ldots, \bgamma_{t+h}')', \bx = (\by_t^{*'}, \ldots, \by_{t+h}^{*'})'$, $\tilde\bB = \bB$, and
\begin{equation*}
    \tilde\bOmega = 
    \begin{bmatrix}
    \bD & \bO & \ldots & \bO \\ 
    \bD\bar\bPsi_1 & \bD & \ldots & \bO \\
    \vdots & \ddots & \ddots & \vdots \\
    \bD\bar\bPsi_h & \ldots & \bD\bar\bPsi_1 & \bD
    \end{bmatrix},
\end{equation*}
with $\bar\bA_i=\bO$ for $i>l$, $\bar\bPsi_j=\bO$ for $j>q$, and $\bD = \text{diag}(\bL)^{-1}$, where $\text{diag}(\bX)$ is a diagonal matrix of the diagonal of $\bX$.

Since $\tilde\bB = \bB$, path coefficients of edges connecting two endogenous variables can be obtained from the Cholesky system \eqref{eq:proof-theorem-4-cholesky-system} and therefore require no identification of any shock.

Next, suppose $\varepsilon_{i,t}=\varepsilon_i$ is identified such that $[\bPhi_0]_{\cdot, i}$ is identified, where $\bPhi_0$ is the contemporaneous impulse response matrix. Since $\bPhi_0 = \bL^{-1}\bQ'$ and $\bL$ can be found from the Cholesky system \eqref{eq:proof-theorem-4-cholesky-system}, $[\bQ']_{\cdot, i}=\bL[\bPhi_0]_{\cdot, i}$ is identified. This immediately implies that $[\bD\bQ']_{\cdot, i} = \bD[\bQ']_{\cdot, i}$ is identified. The Cholesky system also provides $\bD\bar\bPsi_k=\bD\bQ'\bPsi_k\bQ$. Post multiplying by $[\bQ']_{\cdot, i}$ and using $\bQ[\bQ']_{\cdot, i}=[\bI]_{\cdot, i}$ yields $\bD\bQ'[\bPsi_k]_{\cdot, i}$.

Thus, if $[\bPhi_0]_{\cdot, i}$ is identified, the $i$-th column of $\bOmega$ in equation \eqref{eq:general-model-xbx} is identified and thus all path coefficients for edges leaving $\varepsilon_i$. Since only $\bB$ and $[\bOmega]_{\cdot, i}$ are needed to compute any transmission effects for $\varepsilon_i$, this proves the statement.
\end{proof}

\begin{proof}[\textbf{Proof of Lemma~\ref{lem:invariance-edges}}]
    
    Let $\bS$ be the permutation matrix that inverses the order and $\bQ^*\bR=\bA_0\bS=\bar\bA_0$ be the QR-decomposition of $\bar\bA_0$. Since $[\bQ^*]_{\cdot, j}$ is obtained as the residual of the projection of the $j$-th column of $\bar\bA_0$ on all previous columns, it is invariant to the permutation $\sigma_y$. Let $\bQ\bL=\bA_0$ be the QL-decomposition of $\bA_0$. It holds $\bQ=\bQ^*\bS$. Thus $[\bQ]_{\cdot, j}$ is invariant to the permutation $\sigma_y$. Letting $\tilde{\bQ}\tilde{\bL}=\bA_0\bT'$ be the QL-decomposition of $\bA_0\bT'$, $[\bQ']_{j, \cdot} = [\tilde{\bQ}']_{j, \cdot}$.  

    Next, from $\tilde{\bQ}\tilde{\bL}=\bA_0\bT'$ it follow $\tilde{\bL}=\tilde{\bQ}'\bA_0\bT'$ and therefore $[\tilde{\bL}]_{j, \cdot} = [\tilde{\bQ}']_{j, \cdot}\bA_0\bT'$. From the result above, it immediately follows that $[\tilde{\bL}]_{j, \cdot} = [\bQ']_{j, \cdot}\bA_0\bT' = [\bL]_{j, \cdot}\bT'$ which is equivalent to $[\bL]_{j, i} = [\tilde{\bL}]_{j, \sigma_y(i)}$ for $i=1,\ldots,K$.

    The previous two results lead to the proof of the lemma. Let $\tilde{\by}_t = \bT\by_t$, $\tilde{\bA}_i = \bA_i\bT'$ for $i=1,\ldots, p$, and $\tilde{\bD}=\text{diag}(\tilde{\bL})^{-1}$. The lower-triangular system of $\tilde{y}_{j,t}$, where $\sigma_y(j) = j$, is then given by 
    \begin{equation*}
    \begin{split}
    \tilde{y}_{j,t} &= [\bI - \tilde{\bD}\tilde{\bL}]_{j, \cdot}\tilde{\by}_t + \sum_{i=1}^p [\tilde{\bD}\tilde{\bQ}'\tilde{\bA}_i]_{j, \cdot}\tilde{\by}_{t-i} + \sum_{i=1}^q [\tilde{\bD}\tilde{\bQ}'\bPsi_i]_{j, \cdot}\bvarepsilon_{t-i} + [\tilde{\bD}\tilde{\bQ}']_{j, \cdot}\bvarepsilon_t \\
    &= ([\bI]_{j, \cdot} - [\tilde{\bL}]_{j,j}^{-1}[\tilde{\bL}]_{j, \cdot})\tilde{\by}_t + \sum_{i=1}^p [\tilde{\bL}]_{j,j}^{-1}[\tilde{\bQ}']_{j, \cdot}\bA_i\bT'\tilde{\by}_{t-i} \\ &\quad + \sum_{i=1}^q [\tilde{\bL}]_{j, j}^{-1}[\tilde{\bQ}']_{j, \cdot}\bPsi_i\bvarepsilon_{t-i} + [\tilde{\bL}]_{j, j}^{-1}[\tilde{\bQ}']_{j, \cdot}\bvarepsilon_t \\
    &= ([\bI]_{j, \cdot} - [\bL]_{j,j}^{-1}[\bL]_{j, \cdot})\by_{t} + \sum_{i=1}^p [\bL]_{j,j}^{-1}[\bQ']_{j, \cdot}\bA_i\by_{t-i} \\ &\quad + \sum_{i=1}^q [\bL]_{j,j}^{-1}[\bQ']_{j, \cdot}\bPsi_i \bvarepsilon_{t-i} + [\bL]_{j, j}^{-1}[\bQ']_{j, \cdot}\bvarepsilon_t \\
    &= [\bI - \bD\bL]_{j, \cdot}\by_t + \sum_{i=1}^p [\bD\bQ'\bA_i]_{j, \cdot}\by_{t-i} + \sum_{i=1}^q [\bD\bQ'\bPsi_i]_{j, \cdot}\bvarepsilon_{t-i} + [\bD\bQ']_{j, \cdot}\bvarepsilon,
    \end{split}
    \end{equation*}
    with $\bQ\bL=\bA_0$ and $\bD=\text{diag}(\bL)^{-1}$, and where the second equation follows from the definition of $\tilde{D}$, the third equation from the results above and $\sigma_y(j) = j$, and the last equation follows by definition of the matrices. Thus, all coefficients remain unchanged, implying that the transmission matrix $\bT$ simply re-orders the coefficients by $\sigma_y$. The construction of $\bB$, $\bOmega$ and $\sigma_x$ then imply the two statements in the lemma.
\end{proof}

\begin{proof}[\textbf{Proof of Lemma \ref{lemma:invariance-cholesky}}]
    The result immediately follows from Proposition 4.1 in \citet{christianoMonetaryPolicyShocks1999}.
\end{proof}

\end{appendices}

\newpage
\pdfbookmark[0]{Online Appendix}{online-appendix}
\begin{center}
{\Large\textbf{Online Appendix}}
\end{center}

\vspace{12pt}
\noindent
This Online Appendix accompanies the paper ``\papertitle''. It contains Appendices B to F; Appendix A, with the proofs of the main results, is part of the main paper. Sections, equations, lemmas, figures, and tables that are not numbered with the letters B to F refer to the main paper.

\appendixtitleoff
\begin{appendices}
\section{Additional Proofs}
\label{appendix:proof-algebra}

\begin{proof}[\textbf{Proof of Lemma \ref{enum:prop-1}}]
Each path $p\in P_{\varepsilon_{i},x_j}(b)$ can be split into two parts: $p_1$ which goes from $\varepsilon_{i}$ to $x_k$, and $p_2$ which goes from $x_k$ to $x_j$. So we have $P_{\varepsilon_{i},x_j}(b) = \mathcal{P}_{\varepsilon_{i},x_k} \otimes \mathcal{P}_{x_k,x_j}$ where $\otimes$ denotes the cross product of paths: For $p_1 \in \mathcal{P}_{\varepsilon_{i}, x_k}$ and $p_2 \in \mathcal{P}_{x_k,x_j}$, $p_1 \to p_2 \in \mathcal{P}_{\varepsilon_{i},x_k} \otimes \mathcal{P}_{x_k, x_j}$. The total path specific effect can then be written as 
\begin{equation*}
\begin{split}
\calQ_{\xi}[P_{\varepsilon_i,x_j}(b)] &= \xi\sum_{p \in P_{\varepsilon_i,x_j}(b)}\prod_{(u \to v) \in p} \omega_{uv} \\
&= \xi\sum_{p_1 \in \mathcal{P}_{\varepsilon_i,x_k}}\sum_{p_2 \in \mathcal{P}_{x_k,x_j}}\left(\prod_{(u \to v) \in p_1}\omega_{uv}\right)\left(\prod_{(u \to v) \in p_2}\omega_{uv}\right) \\
&= \xi\sum_{p_1 \in \mathcal{P}_{\varepsilon_i,x_k}}\left\{\prod_{(u \to v) \in p_1}\omega_{uv}\right\}\sum_{p_2 \in \mathcal{P}_{x_k,x_j}}\left\{\prod_{(u \to v) \in p_2}\omega_{uv}\right\} \\
&= \calQ_{\xi}[\mathcal{P}_{\varepsilon_{i},x_k}]\calQ_{1}[\mathcal{P}_{x_k,x_j}]. \qedhere
\end{split}
\label{eq:proof-prop-1-1}
\end{equation*}    
\end{proof}

\begin{proof}[\textbf{Proof of Lemma \ref{enum:prop-2}}]
We prove this by induction. We know from Lemma \ref{enum:prop-1} that it holds for $k=1$. Now suppose it holds for some $k$. We then need to show that it holds for $k+1$. To do so, first define $b_k = \wedge_{m=1}^k x_{i_m}$ and $b_{k+1} = \wedge_{m=1}^{k+1} x_{i_m}$. Any path $p \in P_{\varepsilon_i, x_j}(b_{k+1})$ can then be split into three parts: $p_1 \in P_{\varepsilon_i,x_{i_k}}(b_k)$, $p_2 \in \mathcal{P}_{x_{i_k}, x_{i_{k+1}}}$, and $p_3 \in \mathcal{P}_{x_{i_{k+1}}, x_j}$. The following is then an extension of the proof to \ref{enum:prop-1}: 
\begin{equation*}
\begin{split}
&\calQ_{\xi}[P_{\varepsilon_i,x_j}(b_{k+1})] = \xi\sum_{p \in P_{\varepsilon_i,x_j}(b_{k+1})}\prod_{(u\to v) \in p}\omega_{u,v} \\
&= \xi\sum_{p_1 \in P_{\varepsilon_i,x_{i_k}}(b_k)}\sum_{p_2 \in \mathcal{P}_{x_{i_k}, x_{i_{k+1}}}}\sum_{p_3 \in \mathcal{P}_{x_{i_{k+1}}, x_j}}\left(\prod_{(u\to v) \in p_1}\omega_{u,v}\right)\left(\prod_{(u\to v) \in p_2}\omega_{u,v}\right)\left(\prod_{(u\to v) \in p_3}\omega_{u,v}\right) \\
&= \xi\sum_{p_1 \in P_{\varepsilon_i,x_{i_k}}(b_k)}\left\{\prod_{(u\to v) \in p_1}\omega_{u,v}\right\}\sum_{p_2 \in \mathcal{P}_{x_{i_k}, x_{i_{k+1}}}}\left\{\prod_{(u\to v) \in p_2}\omega_{u,v}\right\}\sum_{p_3 \in \mathcal{P}_{x_{i_{k+1}}, x_j}}\left\{\prod_{(u\to v) \in p_3}\omega_{u,v}\right\} \\
&=\calQ_{\xi}[P_{\varepsilon_i,x_{i_k}}(b_k)]\calQ_{1}[\mathcal{P}_{x_{i_k}, x_{i_{k+1}}}]\calQ_{1}[\mathcal{P}_{x_{i_{k+1}}, x_j}] .
\end{split}
\end{equation*}

By the induction hypothesis, 
\begin{equation*}
    \calQ_{\xi}[P_{\varepsilon_i, x_{i_k}}(b_k)] = \calQ_{\xi}[\calP_{\varepsilon_i, x_{i_1}}]\calQ_1[\calP_{x_{i_1}, x_{i_2}}]\dots \calQ_1[\calP_{x_{i_{k-1}}, x_{i_k}}].
\end{equation*}
Thus, the statement holds for $k+1$. This implies that it holds for all $k\in \mathbb{N}$, concluding the proof. 
\end{proof}
\begin{proof}[\textbf{Proof of Lemma \ref{enum:prop-3}}]
Since for any $p \in P_{\varepsilon_i,x_j}\cup P_{\varepsilon_i,x_j}'$, we have $p\in P_{\varepsilon_i,x_j}$ or $p\in P_{\varepsilon_i,x_j}'$ but never in both, we can write 
\begin{align*}
    \calQ_{\xi}[P_{\varepsilon_i,x_j}\cup P_{\varepsilon_i,x_j}'] &= \xi\sum_{p \in P_{\varepsilon_i,x_j}\cup P_{\varepsilon_i,x_j}'}\prod_{(u\to v) \in p}\omega_{u,v} \\
    &= \xi\sum_{p \in P_{\varepsilon_i, x_j}}\prod_{(u\to v) \in p}\omega_{u,v} + \xi\sum_{p \in P_{\varepsilon_i,x_j}'}\prod_{(u\to v) \in p}\omega_{u,v} \\
    &= \calQ_{\xi}[P_{\varepsilon_i,x_j}] + \calQ_{\xi}[P_{\varepsilon_i,x_j}']. \qedhere
\end{align*}
\end{proof}

\begin{proof}[\textbf{Proof of Lemma \ref{enum:prop-4}}]
Make the following two observations: 
\begin{enumerate}
    \item $P_{\varepsilon_i, x_j}\backslash P_{\varepsilon_i,x_j}'$ and $P_{\varepsilon_i,x_j} \cap P_{\varepsilon_i,x_j}'$ are disjoint. 
    \item $P_{\varepsilon_i,x_j}\backslash P_{\varepsilon_i,x_j}' \cup (P_{\varepsilon_i,x_j} \cap P_{\varepsilon_i,x_j}')=P_{\varepsilon_i,x_j}$.
\end{enumerate}
Using these observations and \ref{enum:prop-3}, we have 
\begin{equation*}
\calQ_{\xi}[P_{\varepsilon_i,x_j}] = \calQ_{\xi}[P_{\varepsilon_i,x_j}\backslash P_{\varepsilon_i,x_j}' \cup (P_{\varepsilon_i,x_j} \cap P_{\varepsilon_i,x_j}')] = \calQ_{\xi}[P_{\varepsilon_i,x_j}\backslash P_{\varepsilon_i,x_j}'] + \calQ_{\xi}[P_{\varepsilon_i,x_j}\cap P_{\varepsilon_i,x_j}'].
\end{equation*}
The result follows immediately.  
\end{proof}
\begin{proof}[\textbf{Proof of Lemma \ref{enum:prop-5}}]
First, let $p \in P_{\varepsilon_i, x_j}(b) \cap P_{\varepsilon_i,x_j}(b')$. Then $p$ is in both $P_{\varepsilon_i,x_j}(b)$ and $P_{\varepsilon_i,x_j}(b')$. This can only be if $p$ satisfies both $b$ and $b'$. But then $p$ is also in $P_{\varepsilon_i,x_j}(b \wedge b')$. Second, let $p \in P_{\varepsilon_i, x_j}(b \wedge b')$. Then $p$ satisfies both $b$ and $b'$. Thus $p\in P_{\varepsilon_i,x_j}(b)$ and $p\in P_{\varepsilon_i,x_j}(b')$. But then $p\in P_{\varepsilon_i,x_j}(b)\cap P_{\varepsilon_i,x_j}(b')$. The result now follows from the above two observations. 
\end{proof}

\begin{proof}[\textbf{Proof of Lemma \ref{enum:prop-6}}]
First, let $p\in P_{\varepsilon_i,x_j}(b \vee b')$. Then $p$ satisfies either $b$, or $b'$ or both. 
\begin{itemize}
        \item If $p$ satisfies $b$, then $p\in P_{\varepsilon_i,x_j}(b)$ and thus $p \in P_{\varepsilon_i,x_j}(b)\cup (P_{\varepsilon_i,x_j}(b')\backslash P_{\varepsilon_i,x_j}(b \wedge b'))$.
        \item If $p$ satisfies $b$ and $b'$, then $p$ satisfies $b$ and thus $p\in P_{\varepsilon_i,x_j}(b)$ and therefore $p\in P_{\varepsilon_i,x_j}(b)\cup (P_{\varepsilon_i,x_j}(b')\backslash P_{\varepsilon_i,x_j}(b \wedge b'))$.
        \item If $p$ satisfies $b'$ but not $b$, then $p\in P_{\varepsilon_i,x_j}(b')\backslash P_{\varepsilon_i,x_j}(b \wedge b')$ and thus $p \in P_{\varepsilon_i,x_j}(b)\cup (P_{\varepsilon_i,x_j}(b')\backslash P_{\varepsilon_i,x_j}(b \wedge b'))$. 
    \end{itemize}
Therefore, $p\in P_{\varepsilon_i, x_j}(b\vee b') \Rightarrow p\in P_{\varepsilon_i,x_j}(b)\cup (P_{\varepsilon_i,x_j}(b')\backslash P_{\varepsilon_i,x_j}(b \wedge b'))$.

Second, let $p\in P_{\varepsilon_i,x_j}(b)\cup (P_{\varepsilon_i,x_j}(b')\backslash P_{\varepsilon_i,x_j}(b \wedge b'))$.
    \begin{itemize}
        \item If $p\in P_{\varepsilon_i,x_j}(b)$, then $p$ satisfies $b$ or $p$ satisfies $b$ and $b'$.
        \item If $p \in P_{\varepsilon_i,x_j}(b')\backslash P_{\varepsilon_i,x_j}(b \wedge b')$, then $p$ satisfies $b'$ but not $b$.
        \item $P_{\varepsilon_i,x_j}(b)$ and $P_{\varepsilon_i,x_j}(b')\backslash P_{\varepsilon_i,x_j}(b \wedge b')$ are disjoint and thus $p$ never falls into both. 
    \end{itemize}
Thus, $p$ always satisfies $b$ or $b'$ or both, and thus $p\in P_{\varepsilon_i,x_j}(b)\cup (P_{\varepsilon_i,x_j}(b')\backslash P_{\varepsilon_i,x_j}(b \wedge b')) \Rightarrow p\in P_{\varepsilon_i,x_j}(b\vee b')$. Both points together prove the statement.
\end{proof}
\begin{proof}[\textbf{Proof of Lemma \ref{enum:prop-7}}]
First, let $p \in P_{\varepsilon_i,x_j}(b \wedge \neg b')$. Then $p$ satisfies $b$ but not $b'$. Thus $p \in P_{\varepsilon_i,x_j}(b)$ but $p \notin P_{\varepsilon_i,x_j}(b \wedge b')$. Therefore, $p\in P_{\varepsilon_i,x_j}(b)\backslash P_{\varepsilon_i,x_j}(b\wedge b')$. Second, let $p\in P_{\varepsilon_i,x_j}(b)\backslash P_{\varepsilon_i,x_j}(b\wedge b')$. Then $p$ satisfies $b \wedge \neg (b\wedge b')$. A truth table  then shows that $p$  satisfies $b\wedge \neg b'$ and therefore $p\in P_{\varepsilon_i,x_j}(b\wedge \neg b')$. Both points together prove the statement.
\end{proof}
\begin{proof}[\textbf{Proof of Lemma \ref{enum:prop-8}}]
Note that 
$P_{\varepsilon_i,x_j}(\neg x_k) = \mathcal{P}_{\varepsilon_i,x_j}\backslash P_{\varepsilon_i,x_j}(x_k)$. Then by Property \ref{enum:prop-4} we have 
\begin{equation*}
\begin{split}
\calQ_{\xi}[P_{\varepsilon_i,x_j}(\neg x_k)] &=  \calQ_{\xi}[\mathcal{P}_{\varepsilon_i,x_j}\backslash P_{\varepsilon_i,x_j}(x_k)] = \calQ_{\xi}[\mathcal{P}_{\varepsilon_i,x_j}] - \calQ_{\xi}[\mathcal{P}_{\varepsilon_i,x_j}\cap P_{\varepsilon_i,x_j}(x_k)] \\&= \calQ_{\xi}[\mathcal{P}_{\varepsilon_i,x_j}] - \calQ_{\xi}[P_{\varepsilon_i,x_j}(x_k)].\qedhere
\end{split}
\end{equation*}
\end{proof}

\section{From VARMA to Systems Form}
\label{appendix:rewrite}
In this appendix, we provide details how the lower-triangular from of the structural VARMA given by
\begin{equation}
    \bL\by^*_t = \sum_{i=1}^\ell \bQ'\bA_i^*\by^*_{t-i}  + \sum_{j=1}^q \bQ'\bPsi_j\bvarepsilon_{t-j} + \bQ'\bvarepsilon_t,
    \label{app:eq:general-model-ql}
\end{equation}
as presented in equation \eqref{eq:general-model-ql} can be re-written into the general systems form 
\begin{equation}
    \bx = \bB\bx + \bOmega\bvarepsilon,
    \label{app:eq:general-model-xbx}
\end{equation}
as presented in equation \eqref{eq:general-model-xbx}.

Let $h$ be the fixed finite horizon of the transmission channels under consideration. Equation \eqref{app:eq:general-model-ql} for time points $t, \dots, t+h$ can then be represented as the  system of equations
\begin{equation*}
\begin{cases}
    \bL\by^*_t = \sum_{i=1}^\ell \bQ'\bA^*_i\by^*_{t-i} + \sum_{j=1}^q \bQ'\bPsi_j\bvarepsilon_{t-j}   + \bQ'\bvarepsilon_t \\ 
    \vdots  \\
    \bL\by^*_{t+h} = \sum_{i=1}^\ell \bQ'\bA^*_i\by^*_{t+h-i} + \sum_{j=1}^q \bQ'\bPsi_j\bvarepsilon_{t+h-j} + \bQ'\bvarepsilon_{t+h} ,  
\end{cases}  
\end{equation*}
which contains all information of the dynamic system up to horizon $h$. Multiplying all equations on the left-hand-side and right-hand-side by $\bD=\text{diag}(\bL)^{-1}$, where $\text{diag}(\bX)$ is a diagonal matrix containing the diagonal of $\bX$, and moving $(\bI - \bD\bL)\by^*_{t+i}$ to the right-hand-side in equations $i=0, \dots, h$, we obtain
\begin{equation}
\resizebox{0.9\textwidth}{!}{$
\begin{cases}
    \by^*_t = (\bI - \bD\bL)\by^*_t +  \sum_{i=1}^\ell \bD\bQ'\bA^*_i\by^*_{t-i} + \sum_{j=1}^q \bD\bQ'\bPsi_j\bvarepsilon_{t-j} + \bD\bQ'\bvarepsilon_t  \\ 
    \vdots \\
    \by^*_{t+h} = (\bI - \bD\bL)\by^*_{t+h} + \sum_{i=1}^\ell \bD\bQ'\bA^*_i\by^*_{t+h-i}  + \sum_{j=1}^q \bD\bQ'\bPsi_j\bvarepsilon_{t+h-j} + \bD\bQ'\bvarepsilon_{t+h}.
\end{cases}
$}
\label{eq:appendix-rewrite-2}
\end{equation}

Given a shock in period $t$, then variables $\by^*_{t-i}$ for $i>0$ are unaffected by the shock, and thus do not play a role in the propagation of the shock. We may therefore ignore them by setting their corresponding coefficients equal to zero. System \eqref{eq:appendix-rewrite-2} then becomes 
\begin{equation}
\resizebox{0.9\textwidth}{!}{
$
\begin{cases}
    \by^*_t = (\bI - \bD\bL)\by^*_t + \bD\bQ'\bvarepsilon_t \\ 
    \vdots \\
    \by^*_{t+h} = (\bI - \bD\bL)\by^*_{t+h} + \sum_{i=1}^{\min(h, p)} \bD\bQ'\bA^*_i\by^*_{t+h-i}  + \sum_{j=1}^{\min(h, q)}\bD\bQ'\bPsi_j\bvarepsilon_{t+h-j} + \bD\bQ'\bvarepsilon_{t+h}.
\end{cases}
$
}
\label{eq:appendix-rewrite-3}
\end{equation}
System  \eqref{eq:appendix-rewrite-3} is then equivalent to the general form \eqref{app:eq:general-model-xbx}, with $\bx = (\by_t^{*'}, \dots, \by_{t+h}^{*'})'$, $\bvarepsilon = (\bvarepsilon_t', \dots, \bvarepsilon_{t+h}')'$, and 
\begin{equation*}
\resizebox{0.9\textwidth}{!}{$
    \begin{array}{ccc}
         \bB = \begin{bmatrix}
        \bI - \bD\bL & \bO & \dots & \bO \\
        \bD\bQ'\bA^*_1 & \bI-\bD\bL & \dots & \bO \\
        \vdots & \ddots & \ddots & \vdots \\
        \bD\bQ'\bA^*_h & \dots & \bD\bQ'\bA^*_1 & \bI - \bD\bL
    \end{bmatrix}, & \quad & 
    \bOmega = \begin{bmatrix}
        \bD\bQ' & \bO & \dots & \bO \\
        \bD\bQ'\bPsi_1 & \bD\bQ' & \dots & \bO \\
        \vdots & \ddots & \ddots & \vdots \\
        \bD\bQ'\bPsi_h & \dots & \bD\bQ'\bPsi_1 & \bD\bQ'
   \end{bmatrix},
    \end{array}$}
\end{equation*}
as given in equation \eqref{eq:xbx-B-and-Omega} of Section \ref{sec:sec3-general-framework}.

\section{Computing Transmission Effects}
\label{appendix:computation}

In this section we discuss how transmission effects can be computed in practice. Section \ref{appendix:computation-method} discusses computational issues and proposes an efficient algorithm to compute transmission effects. Section \ref{appendix:computation-correctness} establishes the correctness of the proposed approach. Section \ref{appendix:computation-uncertainty} discusses how uncertainty around point estimates can be quantified using a frequentist or Bayesian approach. 

\subsection{Efficient Computational Approaches}
\label{appendix:computation-method}

Section \ref{sec:sec3-general-framework} shows that given a collection of paths $P_{\varepsilon_i, x_j}$ constituting a transmission channel, transmission effects can be calculated as the total path-specific effect of $P_{\varepsilon_i, x_j}$, $\calQ_\xi(P_{\varepsilon_i, x_j})$. This is an efficient method if the collection of paths is known. However, often times it is more intuitive to define transmission channels in terms of which variable may or may not lie on a path, i.e.\ it is more natural to define transmission channels in terms of Boolean statements (see Appendix \ref{appendix:Boolean}). Given such a definition, the collection of paths must first be found before the total path-specific effect can be calculated. Finding all paths corresponding to a Boolean statement is computationally costly, especially for models with many variables or for effects over many horizons---in either case the number of paths grows exponentially. Thus, first finding all paths and then calculating the total path-specific effect is inefficient in those cases. 

An alternative approach existent for SVARs is to directly use impulse response functions; Theorem \ref{thm:irf-sufficiency} states that for SVARs all transmission effects can be computed using a combination of IRFs. Given a transmission channel $P_{\varepsilon_i, x_j}$ implicitly defined using a Boolean formula $b$, this can be made operational, by first using Lemma \ref{lem:properties} to split the transmission effect of $P_{\varepsilon_i, x_j}$ into terms involving the total path-specific effects of simpler collections of paths---those that use all paths connecting two variables or a shock and a variable. Each term can then be calculated using Theorem \ref{thm:total-effect-and-decomposition} and \ref{thm:irf-sufficiency} as an impulse response. 

The above approach has the advantage that it allows for alternative methods to compute IRFs such as local projections \citep{jordaEstimationAndInference2005}, but is limited to SVARs and requires a full simplification of the Boolean statement, i.e.\ the transmission effect must be simplified to terms involving collections of all paths connecting two variables or a variable and a shock. This is computationally expensive and a more efficient algorithm could be obtained if full simplification was not required---if NOTs ($\neg$) would not have to be simplified. This can be achieved by combining Lemma \ref{lem:properties} with insights obtained from the graphical representation $\calG(\bB, \bOmega)$. We therefore introduce an algorithm next that is computationally more efficient and applicable to both SVARs and SVARMAs. 

Suppose the Boolean formula is $b=x_k$ ($k<j$). Then all paths of the transmission channel $P_{\varepsilon_i, x_j}$ implied by the Boolean statement $b$ must go through $x_k$.  Thus, if we were to remove all paths $p$ in the graph $\calG(\bB, \bOmega)$ that did not go through $x_k$, then the resulting graph $\calG(\bar\bB, \bar\bOmega)$ would only contain paths going through $x_k$. The total effect of the structural shock $\varepsilon_i$ on $x_j$ in the graph $\calG(\bar\bB, \bar\bOmega)$ would then use all paths connecting $\varepsilon_i$ and $x_j$ and going through $x_k$ but no paths that do not go through $x_k$. By Theorem \ref{thm:total-effect-and-decomposition} this total effect is given by $\bar\Phi_{j,i}=(\bI - \bar\bB)^{-1}_{j, \cdot}\bar\bOmega_{\cdot,i}$. Because moving from $\calG(\bB, \bOmega)$ to $\calG(\bar\bB, \bar\bOmega)$ only removes edges and does not change path coefficients, the total effect $\bar\Phi_{j,i}$ equals the transmission effect of the transmission channel $P_{\varepsilon_i, x_j}$. A similar logic can be applied to Boolean formulas $b=\neg x_k$ ($k<j$) implying that paths cannot go through $x_k$; this time, all paths that go through $x_k$ are removed. We can generalise this logic to any Boolean formula of the form $b = \land_{k\in N_1}x_k \land_{k\in N_2}\neg x_k$ with $k<j\forall k\in N_1\cup N_2$, resulting in Algorithm \ref{alg:edge-deletion}. 

\RestyleAlgo{ruled}
\begin{algorithm}
    \caption{Calculating the transmission effect of $b = \land_{k\in N_1}x_k \land_{k\in N_2}\neg x_k$ from shock $\varepsilon_i$ to $x_j$}
    \label{alg:edge-deletion}
    
    \KwIn{Boolean statement $b = \land_{k\in N_1}x_k \land_{k\in N_2}\neg x_k$ with $k<j\forall k\in N_1\cup N_2$, matrices $\bB$ and $\bOmega_{\cdot, i}$, shock size $\xi$}
    \KwOut{Transmission effect $\calQ_{\xi}(P_{\varepsilon_i, x_j})$ where $P_{\varepsilon_i, x_j}$ is implicitly defined by $b$}

    $\bar\bB \gets \bB$ \\
    $\bar\bOmega_{\cdot, i} \gets \bOmega_{\cdot, i}$ \\
    \For{$k\in N_1$}{
        $\bar\bB_{r,s} \gets 0$ if $r>k$ and $s<k$ \\
        $\bar\bOmega_{r,i} \gets 0$ if $r>k$ \\
    }
    \For{$k\in N_2$}{
        $\bar\bB_{k,s} \gets 0$ if $s < k$ \\
        $\bar\bOmega_{k,i} \gets 0$\\
    }
    \Return{$\xi(\bI - \bar\bB)^{-1}_{j, \cdot}\bar\bOmega_{\cdot,i}$}
\end{algorithm}

Algorithm \ref{alg:edge-deletion} requires Boolean formulas of a specific form; however, not all transmission channels can be defined by such a Boolean formula. Thus, Algorithm \ref{alg:edge-deletion} must be combined with Lemma \ref{lem:properties}. Lemma \ref{lem:properties} is first applied to a Boolean formula $b$ until each term involves Boolean formulas amendable to Algorithm \ref{alg:edge-deletion}. Algorithm \ref{alg:edge-deletion} can then be used to compute each term. Since Boolean formulas need not be fully simplified---terms can consist of collections of paths that do not contain all paths connecting two variables or a variable and a shock---the combination of Lemma \ref{lem:properties} and Algorithm \ref{alg:edge-deletion} is often more efficient than using IRFs. Additionally, Lemma \ref{lem:alg-single-shock} shows that this method can be applied even if the full graph is not known---only the structural shock whose effects the researcher wants to decompose is known. Thus, unless the underlying model is an SVAR and local projections are being used, the combination of Lemma \ref{lem:properties} and Algorithm \ref{alg:edge-deletion} has no clear shortcomings compared to using Lemma \ref{lem:properties} with IRFs; it is, therefore, oftentimes the preferred approach. 

\begin{lemma}
    Let $\bL$ be the contemporaneous matrix obtained when using a Cholesky-orthogonalisation scheme with the ordering determined by $\bT$, $\bar\bA_i$ be the AR coefficient matrices and $\bar\bPsi_j$ be the MA coefficient matrices obtained using the same scheme, and suppose $(\bA_0^*)^{-1}_{\cdot, i}=\bPhi_{1:K, i}$ is identified. Then $\bQ'_{\cdot, i} = \bL(\bA_0^*)^{-1}_{\cdot, i}$, $(\bQ'\bPsi_j)_{\cdot, i}=\bar\bPsi_j\bQ'_{\cdot, i}$ for all $j=1,\dots,q$, and 
    \begin{equation*}
        \begin{array}{ccc}
             \bB = \begin{bmatrix}
           \bI - \bD\bL & \bO & \dots & \bO \\
           \bD\bar\bA_1 & \bI - \bD\bL & \dots & \bO \\ 
           \vdots & \ddots & \ddots & \vdots \\
           \bD\bar\bA_h & \dots & \bD\bar\bA_1 & \bI - \bD\bL
        \end{bmatrix} & \text{ and } &
             \bOmega_{\cdot, i} = \begin{bmatrix}
        \bD\bQ'_{\cdot, i}\\
        \bD(\bQ'\bPsi_1)_{\cdot, i}\\
        \vdots\\
        \bD(\bQ'\bPsi_h)_{\cdot, i} 
   \end{bmatrix},\\
        \end{array}
    \end{equation*}
    with $\bD = \text{diag}(\bL)^{-1}$, $\bar\bA_i=\bO$ for $i>\ell$, and $\bar\bPsi_j=\bO$ for all $j>q$.
    \label{lem:alg-single-shock}
\end{lemma}
\begin{proof}[\textbf{Proof of Lemma \ref{lem:alg-single-shock}}]
    After substituting $\bar\bA_i=\bQ'\bA_i^*$ ($i\geq 1)$, $\bar\bPsi_j=\bQ'\bPsi_j\bQ$ ($j\geq 1$) and $\bgamma_t=\bQ'\bvarepsilon_t$, \ref{eq:general-model-ql} can be written as 
    \begin{equation*}
        \bL\by_{t}^* = \sum_{i=1}^{\ell}\bar\bA_i\by_{t-i}^* + \sum_{j=1}^q\bar\bPsi_j\bgamma_{t-j} + \bgamma_t.
    \end{equation*}
    This is the system obtained using a Cholesky-orthogonalisation scheme with the ordering determined by $\bT$. We then have $(\bA_0^*)^{-1}_{\cdot, i} = (\bQ\bL)^{-1}_{\cdot, i}=\bL^{-1}\bQ'_{\cdot, i}$ implying $\bQ'_{\cdot, i}=\bL(\bA_0^*)^{-1}_{\cdot, i}$. Similarly, we have $(\bQ'\bPsi_j)_{\cdot, i}=\bQ'\bPsi_j\bQ\bQ'_{\cdot, i}=\bar\bPsi_j\bQ'_{\cdot, i}$. The form of $\bB$ and $\bOmega_{\cdot, i}$ can then be obtained by substituting the above into equation \eqref{eq:xbx-B-and-Omega}. 
\end{proof}

\subsection{Properties of the Computational Algorithm}
\label{appendix:computation-correctness}

In this section we prove the correctness of the computational approach suggested in the previous section. More precisely, we will show that combining Lemma \ref{lem:properties} with Algorithm \ref{alg:edge-deletion} correctly computes the transmission effect of any transmission channel. Given that Lemma \ref{lem:properties} can be used to split the computation of the transmission effect of a transmission channel $P_{\varepsilon_i, x_j}$ implicitly defined by a Boolean statement $b$ into terms involving Boolean formulas (transmission channels) that can be computed using Algorithm \ref{alg:edge-deletion}, it suffices to show that Algorithm \ref{alg:edge-deletion} correctly computes the transmission effect of a transmission channel indirectly defined by a Boolean formula of the form $b = \land_{k\in N_1}x_k \land_{k\in N_2}\neg x_k$ with $k<j\forall k\in N_1\cup N_2$. 

We first prove that the first loop of Algorithm \ref{alg:edge-deletion} removes all and only paths $p$ from $\calG(\bB, \bOmega)$ that do not satisfy $\land_{k\in N_1} x_k$ with $k<j$ for all $k\in N_1$. This is implied by Lemma \ref{lem:comptuations-1} and \ref{lem:comptuations-2} below. 

\begin{lemma}
    Let $b=x_k$ and $\bar\bB=\bB$, $\bar\bOmega=\bOmega$ except that $\bar\bB_{r,s}=0$ and $\bar\bOmega_{r,i}=0$ whenever $s<k$ and $r>k$. Then $\calG(\bar\bB, \bar\bOmega)$ contains all and only paths $p$ in $\calG(\bB, \bOmega)$ that satisfy $b$. 
    \label{lem:comptuations-1}
\end{lemma}
\begin{proof}[\textbf{Proof of Lemma \ref{lem:comptuations-1}}]
    Suppose there was a path $p$ in $\calG(\bar\bB, \bar\bOmega)$ that did not satisfy $b$. Then $p$ must have an edge of the kind $x_s\to x_r$ or $\varepsilon_i\to x_r$ with $s<k$ and $r>k$. This is a contradiction, because all such edges have been removed. Next, suppose there was a path $p$ in $\calG(\bB, \bOmega)$ not in $\calG(\bar\bB, \bar\bOmega)$ that does satisfy $b$. Since moving from $\calG(\bB, \bOmega)$ to $\calG(\bar\bB, \bar\bOmega)$ only removes edges of the form $x_s\to x_r$ and $\varepsilon_i\to x_r$ with $s<k$ and $r>k$, $p$ must have such an edge. This is a contradiction, because it implies that $p$ cannot satisfy $b$ in the first place.
\end{proof}

\begin{lemma}
    Suppose all paths in $\calG(\bB, \bOmega)$ satisfy the Boolean formula $b'$ and let $b=b' \land x_k$. Let $\bar\bB=\bB$ and $\bar\bOmega=\bOmega$ except that $\bar\bB_{r,s}=0$ and $\bar\bOmega_{r,i}=0$ whenever $s<k$ and $r>k$. Then $\calG(\bar\bB, \bar\bOmega)$ contains all and only paths in $\calG(\bB, \bOmega)$ that satisfy $b$. 
    \label{lem:comptuations-2}
\end{lemma}
\begin{proof}[\textbf{Proof of Lemma \ref{lem:comptuations-2}}]
    Suppose there was a path $p$ in $\calG(\bar\bB, \bar\bOmega)$ that did not satisfy $b$. Then $p$ violates $b'$ or does not go through $x_k$. By Lemma \ref{lem:comptuations-1}, the latter cannot be the case. Thus, $p$ must violate $b'$. This is a contradiction, because $\calG(\bar\bB, \bar\bOmega)$ is obtained from $\calG(\bB, \bOmega)$ by removing paths, and all paths in $\calG(\bB, \bOmega)$ satisfy $b'$. Next, suppose there was a path $p$ in $\calG(\bB, \bOmega)$ not in $\calG(\bar\bB, \bar\bOmega)$ that satisfies $b$. Since all paths in $\calG(\bB, \bOmega)$ already satisfy $b'$, Lemma \ref{lem:comptuations-1} implies that this is not possible. 
\end{proof}

Lemma \ref{lem:comptuations-1} shows that the first iteration of the first loop in Algorithm \ref{alg:edge-deletion} ends with a graph that contains all and only paths in $\calG(\bB, \bOmega)$ that go through $x_{k_1}$ where we write $N_1=\{k_1, \dots, k_{N_1}\}$. Lemma \ref{lem:comptuations-2} then shows that all remaining iterations $m$ end with a graph that contains all and only edges in $\calG(\bB, \bOmega)$ that satisfy $\land_{i=1}^{m}x_{k_m}$. Thus, at the end of the first loop, we are left with a graph that contains all and only paths of $\calG(\bB, \bOmega)$ that satisfy $\land_{k\in N_1}x_k$.

We will next show in Lemmas \ref{lem:comptuations-3} and \ref{lem:comptuations-4} that the second loop removes all and only paths from the graph obtained at the end of the first loop that do not satisfy $\land_{k\in N_2}\neg x_k$. 

\begin{lemma}
    Let $b=\neg x_k$ and $\bar\bB=\bB$, $\bar\bOmega=\bOmega$ except that $\bar\bB_{k,s}=0$ and $\bar\bOmega_{k,i}=0$ whenever $s<k$. Then $\calG(\bar\bB, \bar\bOmega)$ contains all and only paths $p$ in $\calG(\bB, \bOmega)$ that satisfy $b$. 
    \label{lem:comptuations-3}
\end{lemma}
\begin{proof}[\textbf{Proof of Lemma \ref{lem:comptuations-3}}]
    Suppose there was a path $p$ in $\calG(\bar\bB, \bar\bOmega)$ that did not satisfy $b$. Then $p$ must have an edge of the kind $x_s\to x_k$ or $\varepsilon_i\to x_k$ with $s<k$. This is a contradiction, because all such edges have been removed. Next, suppose that there was a path $p$ in $\calG(\bB, \bOmega)$ not in $\calG(\bar\bB, \bar\bOmega)$ that satisfies $b$. Since moving from $\calG(\bB, \bOmega)$ only removes edges of the kind $x_s\to x_k$ or $\varepsilon_i\to x_k$ with $s<k$, $p$ must have such an edge. This is a contradiction, because it would imply that $p$ goes through $x_k$. 
\end{proof}
\begin{lemma}
    Suppose all paths in $\calG(\bB, \bOmega)$ satisfy the Boolean formula $b'$ and let $b=b' \land \neg x_k$. Let $\bar\bB=\bB$ and $\bar\bOmega=\bOmega$ except that $\bar\bB_{k,s}=0$ and $\bar\bOmega_{k,i}=0$ whenever $s<k$. Then $\calG(\bar\bB, \bar\bOmega)$ contains all and only paths in $\calG(\bB, \bOmega)$ that satisfy $b$.
    \label{lem:comptuations-4}
\end{lemma}
\begin{proof}[\textbf{Proof of Lemma \ref{lem:comptuations-4}}]
    Suppose there was a path $p$ in $\calG(\bar\bB, \bar\bOmega)$ that did not satisfy $b$. Then $p$ violates $b'$ or does go through $x_k$. By Lemma \ref{lem:comptuations-3}, the latter cannot be the case. Thus, $p$ must violate $b'$. This is a contradiction, because $\calG(\bar\bB, \bar\bOmega)$ is obtained from $\calG(\bB, \bOmega)$ by removing paths, and all paths in $\calG(\bB, \bOmega)$ satisfy $b'$. Next, suppose there was a path $p$ in $\calG(\bB, \bOmega)$ not in $\calG(\bar\bB, \bar\bOmega)$ that satisfies $b$. Since all paths in $\calG(\bB, \bOmega)$ already satisfy $b'$, Lemma \ref{lem:comptuations-3} implies that this is not possible. 
\end{proof}

Since the first loop ends with a graph that contains all and only paths in $\calG(\bB, \bOmega)$ that satisfy $\land_{k\in N_1}x_k$, Lemma \ref{lem:comptuations-4} shows that the first iteration of the second loop ends with a graph that contains all and only paths in $\calG(\bB, \bOmega)$ that satisfy $\land_{k\in N_1}x_k \land \neg x_{l_1}$ where we write $N_2=\{l_1, \dots, l_{N_2}\}$. Repeated application of Lemma \ref{lem:comptuations-4} then shows that at the end of the second loop we are left with a graph that contains all and only paths in $\calG(\bB, \bOmega)$ that satisfy $\land_{k\in N_1}x_k \land_{k\in N_2}\neg x_k$. We are now ready to state the correctness of Algorithm \ref{alg:edge-deletion}. 

\begin{theorem}
    Algorithm \ref{alg:edge-deletion} correctly computes $\calQ_\xi(P_{\varepsilon_i, x_j})$ where $P_{\varepsilon_i, x_j}$ is implicitly defined by the Boolean formula $b = \land_{k\in N_1}x_k \land_{k\in N_2}\neg x_k$ with $k<j\forall k\in N_1\cup N_2$.
    \label{lem:comptuations-5}
\end{theorem}
\begin{proof}[\textbf{Proof of Theorem \ref{lem:comptuations-5}}]
    The result follows directly from Lemma \ref{lem:comptuations-1}, \ref{lem:comptuations-2}, \ref{lem:comptuations-3}, \ref{lem:comptuations-4} and Theorem \ref{thm:total-effect-and-decomposition}.
\end{proof}

\subsection{Uncertainty Quantification}
\label{appendix:computation-uncertainty}

Algorithm \ref{alg:edge-deletion} together with Lemma \ref{lem:properties} provides an efficient algorithm to compute the transmission effect of any transmission channel. In practice, the required quantities can only be estimated. This section therefore suggests two strategies to obtain point estimates and confidence (credible) intervals for transmission effects. 

Transmission effects can be estimated in a frequentist way by simply plugging in the estimates of the relevant impulse responses and (VAR) path effects into the computational procedures described above. Because transmission effects are smooth functions of impulse responses (Theorem \ref{thm:irf-sufficiency}) or of the SVAR coefficients themselves, consistency and asymptotic normality of the estimated transmission effects follow directly. Thus, point estimates of transmission effects can be obtained by using Algorithm \ref{alg:edge-deletion} and Lemma \ref{lem:properties} together with point estimates of the required quantities. Uncertainty can then either be quantified using the Delta method or a bootstrap approach. Whereas the former requires new cumbersome calculations for the asymptotic variance for every specific transmission effect, the latter can straightforwardly be adapted from existing bootstrap schemes for impulse responses by simply performing the point estimation of the transmission effect in every bootstrap iteration.

Alternatively to the frequentist approach, a Bayesian approach can be adopted. For each posterior draw of the parameters in equation \eqref{eq:general-model}, the transmission effect of a transmission channel can be computed using Algorithm \ref{alg:edge-deletion} and Lemma \ref{lem:properties}. This provides posterior draws of the transmission effect. A point estimate can then be obtained as the posterior mean or median, and credible intervals can be formed in the usual way. 

Both approaches can be used to quantify the uncertainty around point estimates of transmission effects for any number of separate transmission channels. However, when considering a full decomposition of total effects, appropriate measures of uncertainty are more complicated to construct. For example, two channels that perfectly decompose the total effect, are restricted by the fact that they have to add up to the total effect. This induces dependence between the estimates that should be accounted for when constructing uncertainty measures. By treating the two channels jointly one could construct elliptical confidence (credible) sets. How to optimally construct, or indeed communicate, them is an important direction for further research.

\section{Details for Empirical Applications}
\label{appendix:empirical-details}

In this section we present details regarding the empirical applications of Section~\ref{sec:sec5-applications} in the main paper. Section~\ref{appendix:empirical-details-mckay-wolf} presents the data sources used in Section~\ref{sec:sec5-mckay-wolf} and visualises the investigated transmission channels. Section~\ref{appendix:empirical-details-ramey} presents the data sources used in Section~\ref{sec:sec5-ramey}, the construction method for quarterly government defence spending, and visualises the investigated transmission channels. Finally, Section~\ref{appendix:empirical-details-dsge} provides details about the DSGE model used in Section~\ref{sec:sec5-dsge} and provides the formal definitions of the transmission channels.

\subsection{Forward Guidance of Monetary Policy
}
\label{appendix:empirical-details-mckay-wolf}

In Section \ref{sec:sec5-mckay-wolf} we apply TCA to a SVAR including the federal funds rate, the output gap, inflation, and commodity prices, and identify a monetary policy shock using either the \citet{romerNewMeasureMonetary2004} or \citet{gertlerMonetaryPolicySurprises2015} series as an internal instrument \citep{plagborg-mollerLocalProjectionsVARs2021}. Table \ref{tab:data-mckay-wolf} presents all variables used in the estimation, their description, and the source. To stay as close as possible to \citet{mckayWhatCanTime2023}, all data was taken from their replication material. Table \ref{tab:data-mckay-wolf} therefore is mostly a description of their replication material. We focus on the period 1969Q1 to 2007Q4.

\begin{table}
\centering
\caption{Data description for all variables used in Section \ref{sec:sec5-mckay-wolf}}
\begin{tabular}{L{0.2\textwidth} L{0.4\textwidth} L{0.4\textwidth}}
\toprule
\textbf{Variable} & \textbf{Description} & \textbf{Details/Source} \\
\midrule
Federal Funds Rate & Federal Funds Effective Rate & FRED series $\text{FEDFUNDS}$; taken from the replication material of \citet{mckayWhatCanTime2023}\\
Output Gap & Measure of the output gap & $\text{ygap\_hp}$ series of \citet{barnichonIdentifyingModernMacro2020}; taken from the replication material of \citet{mckayWhatCanTime2023}\\
Inflation & Annualised log-difference of the GDP deflator & $\text{pgdp}$ series of \citet{RameyMacroeconomicShocks2016}; taken from the replication material of \citet{mckayWhatCanTime2023}\\
Commodity Prices & Measure of overall commodity prices & $\text{lpcom}$ series of \citet{RameyMacroeconomicShocks2016}; taken from the replication material of \citet{mckayWhatCanTime2023}\\
RR Instrument & \citet{romerNewMeasureMonetary2004} shock series extended by \citet{wielandFinancialDampening2020} & $\text{rr\_3}$ series of \citet{wielandFinancialDampening2020}; taken from the replication material of \citet{mckayWhatCanTime2023}\\
GK Instrument & High-frequency monetary policy instrument of \citet{gertlerMonetaryPolicySurprises2015} & $\text{mp1\_tc}$ series of \citet{RameyMacroeconomicShocks2016}; taken from the replication material of \citet{mckayWhatCanTime2023} and corresponding to the current futures series of \citet{gertlerMonetaryPolicySurprises2015}\\
\bottomrule
\end{tabular}
\label{tab:data-mckay-wolf}
\end{table}

To estimate the SVAR and identify a monetary policy shock, we use standard OLS methods and include either the \citet{romerNewMeasureMonetary2004} or the \citet{gertlerMonetaryPolicySurprises2015} series as the first variable in the SVAR. \citet{plagborg-mollerLocalProjectionsVARs2021} then show that a monetary policy shock can be identified by taking the impulse response to the first Cholesky-orthogonalised shock and normalising it by the contemporaneous response of the federal funds rate to the the first Cholesky-orthogonalised shock. We follow this procedure but further normalise the contemporaneous reaction of the federal funds rate to a 25 basis point increase. Due to the normalisation procedure, direct comparison of the absolute effect sizes of the \citet{romerNewMeasureMonetary2004} and \citet{gertlerMonetaryPolicySurprises2015} shocks is difficult. However, the relative importance of each channel compared to the total effect can be compared. 

Figure \ref{fig:mckay-wolf-illustration} illustrates the two investigated transmission channels for horizon zero. The greyed-out edges and nodes are not involved in the transmission channels. Direct effects are defined as the effect of the monetary policy shock going through a contemporaneous adjustment of the federal funds rate ($\text{ffr})$. This corresponds to all paths going through $\text{ffr}$ contemporaneously. These edges are visualised in blue and the effect along all the paths made up of the blue edges corresponds to the blue bar in Figure \ref{fig:instrument-comparison}. The indirect effect, on the other hand, corresponds to all paths not going through $\text{ffr}$ contemporaneously. These edges are visualised in yellow and the effect through all paths using those edges corresponds to the yellow bar in Figure \ref{fig:instrument-comparison}. 

\begin{figure}
    \centering
    \begin{subfigure}{0.3\textwidth}
    \centering
    \begin{tikzpicture}[scale=0.6, transform shape]

        \node[circle, draw, line width=1pt, minimum size=1.22cm] (ffr) at (0, 6) {ffr};
        \node[circle, draw, line width=1pt, minimum size=1.22cm] (ygap) at (0, 4) {ygap};
        \node[circle, draw, line width=1pt, minimum size=1.22cm] (infl) at (0, 2) {infl};
        \node[circle, draw, line width=1pt, minimum size=1.22cm] (com) at (0, 0) {com};

        \draw[-latex, line width=1pt, mycolor1] (ffr) -- (ygap);
        \draw[-latex, line width=1pt, mycolor1] (ygap) -- (infl);
        \draw[-latex, line width=1pt, mycolor1] (infl) -- (com);
        \draw[-latex, line width=1pt, mycolor1] (ffr) edge [in=45, out=-45] (infl);
        \draw[-latex, line width=1pt, mycolor1] (ffr) edge [in=45, out=-40] (com);
        \draw[-latex, line width=1pt, mycolor1] (ygap) edge [in=60, out=-45] (com);

        \node[rectangle, draw=blue!50, fill=blue!20, line width=0.5pt, minimum size=0.3cm] (e) at (-2.0, 6) {$\varepsilon^i$};
        
        \draw[-latex, line width=1pt, mycolor1] (e) -- (ffr); 
        \draw[-latex, line width=1pt, black!10] (e) -- (ygap); 
        \draw[-latex, line width=1pt, black!10] (e) -- (infl); 
        \draw[-latex, line width=1pt, black!10] (e) -- (com); 
        
    \end{tikzpicture}
    \caption*{Contemporaneous}
    \end{subfigure}%
    \begin{subfigure}{0.3\textwidth}
    \centering
    \begin{tikzpicture}[scale=0.6, transform shape]

        \node[circle, draw, line width=1pt, minimum size=1.22cm, black!20] (ffr) at (0, 6) {\textcolor{black!20}{ffr}};
        \node[circle, draw, line width=1pt, minimum size=1.22cm] (ygap) at (0, 4) {ygap};
        \node[circle, draw, line width=1pt, minimum size=1.22cm] (infl) at (0, 2) {infl};
        \node[circle, draw, line width=1pt, minimum size=1.22cm] (com) at (0, 0) {com};

        \draw[-latex, line width=1pt, black!10] (ffr) -- (ygap);
        \draw[-latex, line width=1pt, mycolor2] (ygap) -- (infl);
        \draw[-latex, line width=1pt, mycolor2] (infl) -- (com);
        \draw[-latex, line width=1pt, black!10] (ffr) edge [in=45, out=-45] (infl);
        \draw[-latex, line width=1pt, black!10] (ffr) edge [in=45, out=-40] (com);
        \draw[-latex, line width=1pt, mycolor2] (ygap) edge [in=60, out=-45] (com);

        \node[rectangle, draw=blue!50, fill=blue!20, line width=0.5pt, minimum size=0.3cm] (e) at (-2.0, 6) {$\varepsilon^i$};
        
        \draw[-latex, line width=1pt, black!10] (e) -- (ffr); 
        \draw[-latex, line width=1pt, mycolor2] (e) -- (ygap); 
        \draw[-latex, line width=1pt, mycolor2] (e) -- (infl); 
        \draw[-latex, line width=1pt, mycolor2] (e) -- (com); 
        
    \end{tikzpicture}
    \caption*{Non-Contemporaneous}
    \end{subfigure}
    \caption{Illustration of the direct and indirect transmission channels at horizon zero for the application in Section \ref{sec:sec5-mckay-wolf}. $\text{ffr}$ is the federal funds rate, $\text{ygap}$ is the output gap, $\text{infl}$ is price inflation, and $\text{com}$ is a commodity price index.}
    \label{fig:mckay-wolf-illustration}
\end{figure}

Using the Boolean definitions of Appendix~\ref{appendix:Boolean}, the direct and indirect channels can be represented by 
\begin{equation*}
   \begin{split}
       \text{Direct} &:\mathrel{=} \text{ffr}_0 \\
       \text{Indirect} &:\mathrel{=} \neg \text{ffr}_0
   \end{split} 
\end{equation*}
Thus, the indirect channel is simply the negation of the direct channel. This provides an efficient way to compute the transmission effects. The algorithms in Appendix~\ref{appendix:computation} can be used to compute the direct (indirect) effect. The difference to the total effect then corresponds to the indirect (direct) effect.

\subsection{Anticipation Effects of Government Spending
}
\label{appendix:empirical-details-ramey}

Section \ref{sec:sec5-ramey}
investigates the role of implementation, and indirectly, anticipation in the transmission of government defence spending news shocks. Table \ref{tab:data-ramey} presents all variables used, their description, and their sources. To stay as close as possible to the original analysis of \citet{rameyGovernmentSpendingMultipliers2018}, all except government defence spending was directly taken from their replication material. A discussion about the construction of quarterly government defence spending is deferred to the next section. We focus throughout on the period 1890Q1 to 2015Q1.

\begin{table}
\centering
\caption{Data description for all variables used in Section \ref{sec:sec5-ramey}}
\begin{tabular}{L{0.2\textwidth} L{0.4\textwidth} L{0.4\textwidth}}
\toprule
\textbf{Variable} & \textbf{Description} & \textbf{Details/Source} \\
\midrule
Government defence News Shocks & Government defence spending news shocks as a percent of real potential GDP & $\text{news}$ series divided by $\text{pgdp}$ and $\text{rgdp\_pott6}$ series of \citet{rameyGovernmentSpendingMultipliers2018} replication material, where $\text{pdgp}$ is the GDP deflator and $\text{rgdp\_pott6}$ is a measure of real potential GDP based on a sixth-order polynomial\\
Government defence Spending & Real government defence spending as percent of real potential GDP & See Appendix \ref{appendix:construction-historical-defence-spending} \\
GDP & Real GDP as a percent of real potential GDP & $\text{rgdp}$ series divided by $\text{rgdp\_pott6}$ series of \citet{rameyGovernmentSpendingMultipliers2018} replication material, where $\text{rgdp\_pott6}$ is a measure of real potential GDP based on a sixth-order polynomial\\
\bottomrule
\end{tabular}
\label{tab:data-ramey}
\end{table}

We define an anticipation channel as the effect of a news shock not going through government defence spending ($\text{gdef}$). Contemporaneously, this is visualised in Figure \ref{fig:ramey-illustration} by the blue edges. All greyed-out edges and nodes are not involved in the transmission channels. The effect along all paths involving only the blue edges---the anticipation channel---corresponds to the blue bar in Figure~\ref{fig:sec5-ramey}. Similarly, we define the implementation channel as the effect going through a change in government defence spending, visualised by the yellow edges. The effect going through all paths using only the yellow edges---the implementation channel---corresponds to the yellow bar in Figure~\ref{fig:sec5-ramey}. As can be seen by superimposing one graph over the other, the sets of paths are non-overlapping and make up all paths in the graph. The two transmission channels therefore perfectly decompose the total effect. 

\begin{figure}[H]
    \centering
    \begin{subfigure}{0.3\textwidth}
    \centering
    \begin{tikzpicture}[scale=0.6, transform shape]

        \node[circle, draw, line width=1pt, minimum size=1.22cm] (news) at (0, 6) {news};
        \node[circle, draw, line width=1pt, minimum size=1.22cm, black!20] (gdef) at (0, 3) {\textcolor{black!20}{gdef}};
        \node[circle, draw, line width=1pt, minimum size=1.22cm] (gdp) at (0, 0) {gdp};

        \draw[-latex, line width=1pt, black!10] (news) -- (gdef);
        \draw[-latex, line width=1pt, mycolor1] (news) edge [in=45, out=-40] (gdp);
        \draw[-latex, line width=1pt, black!10] (gdef) -- (gdp);

        \node[rectangle, draw=blue!50, fill=blue!20, line width=0.5pt, minimum size=0.3cm] (e) at (-2.0, 6) {$\varepsilon^g$};
        
        \draw[-latex, line width=1pt, mycolor1] (e) -- (news); 
        \draw[-latex, line width=1pt, black!10] (e) -- (gdef); 
        \draw[-latex, line width=1pt, mycolor1] (e) -- (gdp); 
        
    \end{tikzpicture}
    \caption*{Anticipation}
    \end{subfigure}%
    \begin{subfigure}{0.3\textwidth}
    \centering
    \begin{tikzpicture}[scale=0.6, transform shape]

        \node[circle, draw, line width=1pt, minimum size=1.22cm] (news) at (0, 6) {news};
        \node[circle, draw, line width=1pt, minimum size=1.22cm] (gdef) at (0, 3) {gdef};
        \node[circle, draw, line width=1pt, minimum size=1.22cm] (gdp) at (0, 0) {gdp};

        \draw[-latex, line width=1pt, mycolor2] (news) -- (gdef);
        \draw[-latex, line width=1pt, black!10] (news) edge [in=45, out=-40] (gdp);
        \draw[-latex, line width=1pt, mycolor2] (gdef) -- (gdp);

        \node[rectangle, draw=blue!50, fill=blue!20, line width=0.5pt, minimum size=0.3cm] (e) at (-2.0, 6) {$\varepsilon^g$};
        
        \draw[-latex, line width=1pt, mycolor2] (e) -- (news); 
        \draw[-latex, line width=1pt, mycolor2] (e) -- (gdef); 
        \draw[-latex, line width=1pt, black!10] (e) -- (gdp); 
        
    \end{tikzpicture}
    \caption*{Implementation}
    \end{subfigure}%
    \caption{Illustration of the anticipation and implementation transmission channels at horizon zero for the application in Section \ref{sec:sec5-ramey}. $\text{news}$ represents the \citet{rameyDefenseNewsShocks2016} defence news series, $\text{gdef}$ is government defence spending, and $\text{gdp}$ is real GDP.}
    \label{fig:ramey-illustration}
\end{figure}

Using the Boolean definitions of Appendix~\ref{appendix:Boolean}, the anticipation and implementation channels for horizon $h$ are defined as 
\begin{equation*}
    \begin{split}
        \text{Anticipation Channel:}\;\;\; & \land_{i=0}^h \neg \text{gdef}_{t+i} \\ 
        \text{Implementation Channel:}\;\;\; & \lor_{i=0}^h \text{gdef}_{t+i}.
    \end{split}
\end{equation*}
Effects through these channels can then be computed using the IRFs obtained from the local projections and Lemma~\ref{lem:properties}.

\subsubsection{Construction of Historical Government defence Spending}
\label{appendix:construction-historical-defence-spending}

In Section \ref{sec:sec5-ramey} we define the anticipation channel as the effect of a news shock not going through government defence spending up to horizon $H$. Thus, to estimate such a channel, observed government defence spending is required. Official FRED data (series A997RC1A027NBEA) is only available from 1947Q1 onwards; however all remaining data is available from 1890Q1 onwards. Thus, to fully exploit the long sample, the FRED series needs to be extended back in time. In this section we explain our methodology. 

To extend the official FRED series back in time, we rely on OWID\footnote{Our World in Data (based on COW \& SIPRI 2018) – processed by Our World in Data, \url{https://ourworldindata.org/grapher/military-expenditure-as-a-share-of-gdp-long?time=earliest..2016} accessed on February 9, 2024 at 14:30}
data on ``Military expenditure as a share of GDP" going back to 1827. However, OWID data is at annual frequency, implying two problems. First, the series is as a share of GDP and thus must first be converted into levels. Second, all remaining data is at quarterly frequency, implying that the military spending series needs to be interpolated to quarterly frequency. 

To transform the OWID series into levels, we aggregate the $\text{ngdp}$ series of \citet{rameyGovernmentSpendingMultipliers2018} to yearly frequency and multiply this yearly nominal GDP measure with the OWID series, resulting in an approximate series of military spending in levels; however, still at annual frequency. The interpolation of this annual series to a quarterly series is then split into two parts. First, we calculate the share of annual total government spending falling into each quarter of a year (we rely on the series $\text{ngov}$ of \cite{rameyGovernmentSpendingMultipliers2018}). Second, we multiply the quarterly shares by the yearly military spending series to obtain an interpolated quarterly military spending series. Under the assumption that the total government spending shares are approximately equal to the true shares of military spending, this interpolation method should approximate quarterly military spending well. 

We check the quality of the interpolated series by comparing it to the official FRED data from 1947Q1 onwards. In levels, the two series have a correlation coefficient of 0.99 and the quarterly shares calculated using FRED data have a 0.82 correlation coefficient with the shares used to interpolate the military spending series. However, in differences, the FRED series and the interpolated series have a correlation coefficient of 0.4. Thus, some mismatch might exist. Since we only aim to demonstrate the proposed methodology, we leave robustness checks with respect to changes in the construction of the quarterly series for future research.

\subsection{Consumption and Investment Channels of Monetary Policy}
\label{appendix:empirical-details-dsge}

Section~\ref{sec:sec5-dsge} applies TCA to a DSGE with household-owned capital and investment adjustment costs to quantify the consumption and investment channels of monetary policy. In this appendix, we describe the exact DSGE model setup and define the consumption and investment channels used in Section~\ref{sec:sec5-dsge} for all horizons. 

\subsubsection{The Dynamic General Equilibrium Model}
\label{appendix:dsge-model}

This appendix derives the model used in Section~\ref{sec:sec5-dsge} and reports its log-linearised form and calibration. The economy is a New Keynesian model in the tradition of \citet{galiMonetaryPolicyInflation2015}, extended in two ways. Households own the capital stock and accumulate it subject to investment adjustment costs as in \citet{christianoNominalRigiditiesDynamic2005}, and the effective stock of installed capital is hit by a capital-quality shock as in \citet{gertlerModelUnconventionalMonetary2011}. Production is carried out by a competitive final-good firm and a unit continuum of monopolistically competitive intermediate firms that set prices subject to Calvo frictions, and a monetary authority sets the nominal rate by a Taylor rule. Throughout, $\hat{x}_t$ denotes the log-deviation of $x_t$ from its non-stochastic steady state, a superscript $n$ denotes the natural (flexible-price) counterpart, and $\hat{x}^g_t \equiv \hat{x}_t - \hat{x}^n_t$ the gap.

\paragraph{Households}
A representative household maximises
\begin{equation*}
E_0 \sum_{t=0}^{\infty} \beta^t\, Z_t \left[ \frac{C_t^{1-\sigma} - 1}{1-\sigma} - \frac{N_t^{1+\varphi}}{1+\varphi} \right],
\end{equation*}
where $C_t$ is consumption, $N_t$ is hours, $\beta$ is the discount factor, $\sigma$ is relative risk aversion, and $\varphi$ is the inverse Frisch elasticity. The term $Z_t$ is an intertemporal preference shock that acts as a demand disturbance. The household holds nominal bonds with gross return $R_t$ and owns capital $K_t$, rented to firms at the real rate $r^k_t$, where $K_t$ is the stock available at the start of period $t$. Capital accumulates as
\begin{equation}
\label{eq:capacc_nl}
K_{t+1} = (1-\delta) K_t + I_t \left[ 1 - S\!\left( \tfrac{I_t}{I_{t-1}} \right) \right],
\qquad
S(x) = \frac{\phi_I}{2}(x-1)^2,
\end{equation}
where $\delta$ is the depreciation rate. The adjustment cost $S(\cdot)$, with curvature $\phi_I$, penalises changes in investment and follows \citet{christianoNominalRigiditiesDynamic2005}; it is the source of the model's sluggish investment dynamics. With real variables $X^R_t = X_t/P_t$ and gross inflation $\Pi_t = P_t/P_{t-1}$, the real budget constraint is
\begin{equation}
\label{eq:budget}
C_t + I_t + Q_t B^R_t \leq W^R_t N_t + \frac{B^R_{t-1}}{\Pi_t} + r^k_t K_t + D^R_t,
\end{equation}
where $W^R_t$ is the real wage and $D^R_t$ real dividends.

Choosing $\{C_t, N_t, K_{t+1}, I_t, B^R_t\}$ and defining Tobin's $q_t$ as the ratio of the multiplier on \eqref{eq:capacc_nl} to that on \eqref{eq:budget}, the first-order conditions give the Euler equation, labour supply, and the capital-pricing condition
\begin{align}
1 &= E_t \left[ \beta \left( \tfrac{C_{t+1}}{C_t} \right)^{-\sigma} \tfrac{Z_{t+1}}{Z_t} \tfrac{R_t}{\Pi_{t+1}} \right], \label{eq:euler_nl}\\[2pt]
N_t^{\varphi} &= W^R_t C_t^{-\sigma}, \label{eq:labour_nl}\\[2pt]
q_t &= E_t \left[ \beta \left( \tfrac{C_{t+1}}{C_t} \right)^{-\sigma} \tfrac{Z_{t+1}}{Z_t} \left( r^k_{t+1} + (1-\delta) q_{t+1} \right) \right], \label{eq:qcapital_nl}
\end{align}
together with the investment condition that links $q_t$ to investment through the adjustment-cost margin,
\begin{equation*}
\begin{split}
1 ={}& q_t \left[ 1 - S\!\left( \tfrac{I_t}{I_{t-1}} \right) - S'\!\left( \tfrac{I_t}{I_{t-1}} \right) \tfrac{I_t}{I_{t-1}} \right] \\
&{}+ E_t \left[ \beta \left( \tfrac{C_{t+1}}{C_t} \right)^{-\sigma} \tfrac{Z_{t+1}}{Z_t}\, q_{t+1}\, S'\!\left( \tfrac{I_{t+1}}{I_t} \right) \left( \tfrac{I_{t+1}}{I_t} \right)^{2} \right].
\end{split}
\end{equation*}

\paragraph{Firms}
A competitive final-good firm aggregates intermediate varieties $i \in [0,1]$ with elasticity of substitution $\varepsilon$, yielding the demand schedule and price index
\begin{equation*}
Y_t(i) = \left[ \frac{P_t(i)}{P_t} \right]^{-\varepsilon} Y_t,
\end{equation*}
\begin{equation*}
P_t = \left( \int_0^1 P_t(i)^{1-\varepsilon}\, di \right)^{1/(1-\varepsilon)}.
\end{equation*}
Each intermediate firm produces with $Y_t(i) = [\Omega_t K_t(i)]^{\alpha} N_t(i)^{1-\alpha}$, where $\alpha$ is the capital share and $\Omega_t$ the capital-quality shock of \citet{gertlerModelUnconventionalMonetary2011}. The shock $\Omega_t$ scales effective capital $\Omega_t K_t$ and so moves the return to capital exogenously. Cost minimisation yields a real marginal cost common to all firms,
\begin{equation*}
\lambda^f_t = \left( \frac{r^k_t}{\alpha \Omega_t} \right)^{\!\alpha} \left( \frac{W^R_t}{1-\alpha} \right)^{\!1-\alpha},
\end{equation*}
along with the factor demands $r^k_t = \alpha \lambda^f_t Y_t(i)/K_t(i)$ and $W^R_t = (1-\alpha)\lambda^f_t Y_t(i)/N_t(i)$. Because marginal cost is common across firms and independent of the date of the last reset, aggregating the production functions over the continuum delivers the aggregate relation $Y_t P^D_t = [\Omega_t K_t]^{\alpha} N_t^{1-\alpha}$, where $P^D_t = \int_0^1 [P_t(i)/P_t]^{-\varepsilon}\,di$ is the price-dispersion index.

Intermediate firms reset prices each period with probability $1-\theta$. A firm resetting at $t$ maximises expected discounted real profits, discounting with the household factor $\mathcal{F}_{t,t+k} = \beta^k (C_{t+k}/C_t)^{-\sigma} Z_{t+k}/Z_t$, which gives a constant-markup reset price
\begin{equation*}
P^*_t = \frac{\varepsilon}{\varepsilon-1}\, \frac{S_t}{H_t},
\end{equation*}
where the auxiliary variables satisfy the recursions
\begin{align*}
S_t &= \lambda^f_t Y_t P_t^{\varepsilon} + \theta E_t\!\left[ \mathcal{F}_{t,t+1} S_{t+1} \right], \\ %
H_t &= Y_t P_t^{\varepsilon-1} + \theta E_t\!\left[ \mathcal{F}_{t,t+1} H_{t+1} \right]. %
\end{align*}
The price-index recursion is $\Pi_t^{1-\varepsilon} = \theta + (1-\theta)[P^*_t/P_{t-1}]^{1-\varepsilon}$. Combining the linearised reset rule with this recursion, and using that $P^D_t$ is unity to first order, yields the New Keynesian Phillips curve \eqref{eq:nkpc} reported with the full linearised system below.

\paragraph{Monetary Policy and Market Clearing}
The monetary authority sets the gross nominal rate by a Taylor rule with interest-rate smoothing,
\begin{equation}
\label{eq:taylor}
\hat{R}_t = \rho_R \hat{R}_{t-1} + (1-\rho_R)\left( \Phi_\pi \hat{\pi}_t + \Phi_y \hat{Y}^g_t \right) + \nu_t,
\qquad
\nu_t = \rho_\nu \nu_{t-1} + \varepsilon^\nu_t,
\end{equation}
where $\Phi_\pi$ and $\Phi_y$ are the responses to inflation and to the output gap $\hat{Y}^g_t$, and $\nu_t$ is the monetary policy shock. The rule responds to the output gap rather than the level of output, which requires solving the natural economy in parallel. The smoothing parameter $\rho_R$ allows for more plausible rate dynamics and is set to zero in the baseline. Bonds are in zero net supply and the goods market clears with output absorbed by consumption and investment,
\begin{equation}
\label{eq:clearing_nl}
Y_t = C_t + I_t.
\end{equation}

\paragraph{Log-linearised System}
We log-linearise around the non-stochastic, zero-inflation steady state. Denoting steady-state values with a bar, this state satisfies $\bar{\Pi}=1$, $\bar{R}=1/\beta$, and $\bar{q}=1$. The three exogenous disturbances follow first-order autoregressive processes,
\begin{align*}
\hat{Z}_t &= \rho_Z \hat{Z}_{t-1} + \varepsilon^Z_t, \\ %
\hat{\Omega}_t &= \rho_\Omega \hat{\Omega}_{t-1} + \varepsilon^\Omega_t, \\
\nu_t &= \rho_\nu \nu_{t-1} + \varepsilon^\nu_t, %
\end{align*}
for the preference, capital-quality, and monetary policy shocks. The sticky-price economy is then governed by
\begin{align}
\hat{C}_t &= E_t \hat{C}_{t+1} - \tfrac{1}{\sigma}\!\left( \hat{R}_t - E_t \hat{\pi}_{t+1} \right) - \tfrac{1}{\sigma}(\rho_Z - 1)\hat{Z}_t, \label{eq:euler}\\[1pt]
\varphi \hat{N}_t &= \hat{W}^R_t - \sigma \hat{C}_t, \label{eq:labour}\\[1pt]
\hat{q}_t &= \beta r^k\, E_t \hat{r}^k_{t+1} + \beta(1-\delta) E_t \hat{q}_{t+1} - \left( \hat{R}_t - E_t \hat{\pi}_{t+1} \right), \label{eq:qcapital}\\[1pt]
\hat{q}_t &= \phi_I \left[ \left( \hat{I}_t - \hat{I}_{t-1} \right) - \beta E_t\!\left( \hat{I}_{t+1} - \hat{I}_t \right) \right], \label{eq:invest}\\[1pt]
\hat{K}_{t+1} &= (1-\delta)\hat{K}_t + \delta \hat{I}_t, \label{eq:capacc}\\[1pt]
\hat{Y}_t &= \alpha \hat{\Omega}_t + \alpha \hat{K}_t + (1-\alpha)\hat{N}_t, \label{eq:production}\\[1pt]
\hat{r}^k_t &= \hat{\lambda}^f_t + \hat{Y}_t - \hat{K}_t, \label{eq:rk}\\[1pt]
\hat{W}^R_t &= \hat{\lambda}^f_t + \hat{Y}_t - \hat{N}_t, \label{eq:wage}\\[1pt]
\hat{\lambda}^f_t &= \alpha \hat{r}^k_t - \alpha \hat{\Omega}_t + (1-\alpha)\hat{W}^R_t, \label{eq:mc}\\[1pt]
\hat{\pi}_t &= \beta E_t \hat{\pi}_{t+1} + \tfrac{(1-\theta)(1-\theta\beta)}{\theta}\, \hat{\lambda}^f_t, \label{eq:nkpc}\\[1pt]
\hat{Y}_t &= s_C \hat{C}_t + s_I \hat{I}_t, \label{eq:clearing}
\end{align}
where $r^k$, $s_C$, and $s_I$ are the steady-state rental rate and the consumption and investment shares of output, $\hat{\lambda}^f_t$ is real marginal cost, and effective capital is $\hat{K}^{\mathrm{eff}}_t = \hat{\Omega}_t + \hat{K}_t$. Equations \eqref{eq:euler}--\eqref{eq:clearing} are the Euler equation, labour supply, the capital ($q$) equation, the investment equation, capital accumulation, aggregate production, the two factor demands, marginal cost, the Phillips curve, and goods-market clearing; the rule \eqref{eq:taylor} closes the block.

The natural economy is the flexible-price counterpart of the same real system, obtained by replacing the Phillips curve \eqref{eq:nkpc} with the constant-marginal-cost condition $\hat{\lambda}^{f,n}_t = 0$ and the ex-ante real rate $\hat{R}_t - E_t\hat{\pi}_{t+1}$ in \eqref{eq:euler} and \eqref{eq:qcapital} with the natural real rate $\hat{r}^n_t$. The gaps are then defined residually as $\hat{x}^g_t = \hat{x}_t - \hat{x}^n_t$, so that the output gap in the policy rule is $\hat{Y}^g_t = \hat{Y}_t - \hat{Y}^n_t$.

\paragraph{Calibration}
The model is calibrated at a quarterly frequency; Table~\ref{tab:calibration} reports the values and sources. The structural parameters follow \citet{christianoNominalRigiditiesDynamic2005} and \citet{galiMonetaryPolicyInflation2015}, implying an annual real rate near three percent, an annual depreciation rate near ten percent, an average price duration of four quarters, and a steady-state markup of $12.5\%$. The investment adjustment-cost curvature is from \citet{christianoNominalRigiditiesDynamic2005} and governs the strength of the hump-shaped investment response, and the persistence of the capital-quality shock is from \citet{gertlerModelUnconventionalMonetary2011}. The monetary policy parameters follow \citet{galiMonetaryPolicyInflation2015}, except that the persistence of the policy shock is set lower to keep the nominal rate's impact response positive and smoothing is switched off. The innovation standard deviations are reported in the second panel of the table. With the transmission-channel analysis of Section~\ref{sec:sec5-dsge} in mind, we set the standard deviations of the preference and capital-quality innovations to a very small value ($10^{-3}$) relative to the monetary policy innovation, whose standard deviation is normalised to $1$. As a result, the model's dynamics are almost entirely driven by the monetary policy shock. The steady-state rental rate and the expenditure shares are not set independently but are implied by the deep parameters through the steady-state restrictions in the bottom panel, keeping the shares internally consistent.

\begin{table}[ht]
\centering
\caption{Baseline calibration of the DSGE model.}
\label{tab:calibration}
\resizebox{\textwidth}{!}{%
\begin{tabular}{llcl}
\toprule
Parameter & Symbol & Value & Source / justification \\
\midrule
Discount factor & $\beta$ & $1.03^{-0.25}$ & \citet{christianoNominalRigiditiesDynamic2005}; 3\% annual real rate \\
Capital share & $\alpha$ & $0.36$ & \citet{christianoNominalRigiditiesDynamic2005} \\
Depreciation rate & $\delta$ & $0.025$ & \citet{christianoNominalRigiditiesDynamic2005}; $\approx 10\%$ annual \\
Relative risk aversion & $\sigma$ & $1$ & Log utility; \citet{galiMonetaryPolicyInflation2015} \\
Inverse Frisch elasticity & $\varphi$ & $5$ & \citet{galiMonetaryPolicyInflation2015} \\
Elasticity of substitution & $\varepsilon$ & $9$ & \citet{galiMonetaryPolicyInflation2015}; markup $12.5\%$ \\
Calvo non-adjustment prob. & $\theta$ & $0.75$ & \citet{galiMonetaryPolicyInflation2015}; duration 4 qrs. \\
Investment adj.-cost curvature & $\phi_I$ & $2.48$ & \citet{christianoNominalRigiditiesDynamic2005}, Tab.~2 \\
\midrule
Taylor rule, inflation & $\Phi_\pi$ & $1.5$ & \citet{galiMonetaryPolicyInflation2015}, p.~68 \\
Taylor rule, output gap & $\Phi_y$ & $0.125$ & \citet{galiMonetaryPolicyInflation2015}, p.~68 \\
Interest-rate smoothing & $\rho_R$ & $0$ & Baseline: no smoothing \\
Preference shock persistence & $\rho_Z$ & $0.5$ & \citet{galiMonetaryPolicyInflation2015}, Sec.~3.4.1.2 \\
Capital-quality persistence & $\rho_\Omega$ & $0.66$ & \citet{gertlerModelUnconventionalMonetary2011}, p.~27 \\
Monetary shock persistence & $\rho_\nu$ & $0.3$ & Lower than \citet{galiMonetaryPolicyInflation2015} \\
\midrule
Preference innovation s.d. & $\sigma_{\varepsilon^Z}$ & $10^{-3}$ & Small: dynamics driven by MP shock \\
Capital-quality innovation s.d. & $\sigma_{\varepsilon^\Omega}$ & $10^{-3}$ & Small: dynamics driven by MP shock \\
Monetary innovation s.d. & $\sigma_{\varepsilon^\nu}$ & $1$ & Normalisation \\
\midrule
Steady-state rental rate & $r^k$ & $\left[1-\beta(1-\delta)\right]/\beta$ & Derived \\
Investment share & $s_I$ & $\alpha\frac{\varepsilon-1}{\varepsilon}\frac{\delta}{r^k}$ & Derived \\
Consumption share & $s_C$ & $1 - s_I$ & Goods-market clearing \\
\bottomrule
\end{tabular}%
}
\end{table}

\subsubsection{Transmission Channels}

Sections~\ref{sec:sec2-illustrative-example} and \ref{sec:sec5-dsge} define a consumption and investment channel in two possible ways. At horizon $h=0$, the baseline definition of the consumption channel is given by $\{\varepsilon^M \to C \to \pi\}$; the alternative definition by $\{\varepsilon^M \to C \to \pi, \; \varepsilon^M \to C \to I \to \pi\}$; and the investment channel is in both cases defined by $\{\varepsilon^M \to C \to I \to \pi, \; \varepsilon^M \to I \to \pi\}$. We now formally extend these definitions to horizons $h\geq 0$ using the boolean definitions introduced in Appendix~\ref{appendix:Boolean}. We start with the alternative definitions, which are closer to the channels defined in Sections~\ref{sec:sec5-mckay-wolf} and \ref{sec:sec5-ramey}.

\paragraph{TCA Closest to Counterfactual}
The closest that TCA can come conceptually to the counterfactual experiment in Section~\ref{sec:sec2-illustrative-example} is by defining a channel that captures all the effects transmitted through consumption, respectively all the effects transmitted through investment. In Section~\ref{sec:sec2-illustrative-example}, we also refer to this as the \emph{alternative} definition of the consumption and investment channels.

To capture all the effects transmitted through consumption and investment, we define the channels at horizon $h$ as
\begin{equation*}
    \begin{split}
    \text{Consumption:} &\quad \lor_{i=0}^h C_i = \neg (\land_{i=0}^h \neg C_i) \\
    \text{Investment:} &\quad  \lor_{i=0}^h I_i = \neg (\land_{i=0}^h \neg I_i),
    \end{split}
\end{equation*}
where the negated definitions can be more efficiently computed using the algorithms in Appendix~\ref{appendix:computation}.

\paragraph{Baseline Definitions}
As discussed in Section~\ref{sec:sec5-dsge}, the baseline definition of the consumption channel, $\{\varepsilon^M \to C \to \pi\}$, can be extended in various ways to horizon $h=1$. We take a broad view and call the union of all those extensions the consumption channel. We follow the same logic for horizons $h>1$. Let $Q^C_{i,j}$ denote the channel that propagates within consumption and investment until period $i$ and moves directly from consumption in period $i$ to inflation in period $j$. The paths corresponding to this channel are then captured by the boolean statement
\begin{equation*}
\resizebox{0.8\textwidth}{!}{$\displaystyle
Q^C_{i,j}=
\begin{cases}
C_0 \land \neg I_0 \land \pi_0 & \text{if } i=j=0 \\
\bigwedge_{k=0}^{i-1}\left(\neg \pi_k\right) \land C_i \land \neg I_i \land \pi_j & \text{if } i=j, i\geq 0 \\
\bigwedge_{k=0}^{i-1}\left(\neg \pi_k\right) \land C_i \land \neg I_i \land \neg \pi_i \land \neg C_j \land \neg I_j \land \pi_j & \text{if } j=i+1, i\geq 0 \\ 
\bigwedge_{k=0}^{i-1}\left(\neg \pi_k\right) \land C_i \land \neg I_i \land \neg \pi_i \bigwedge_{k=i+1}^{j-1}\left(\neg C_k \land \neg I_k \land \neg \pi_k\right)  & \text{if } j > i+1, i \geq 0 \\
\quad\quad \land \neg C_j \land \neg I_j \land \pi_j.
\end{cases}
$}
\end{equation*}
The consumption channel at horizon $h\geq 0$ is then defined as
\begin{equation*}
    \text{Consumption:} \quad \bigvee_{i=0}^h\bigvee_{j=i}^h Q^C_{i,j}.
\end{equation*}

The investment channel can analogously be extended to all horizons. Let $Q^I_{i,j}$ denote the channel that propagates within consumption and investment until period $i$ and directly moves from investment in period $i$ to inflation in period $j$. The paths corresponding to this channel are then captured by the boolean statement
\begin{equation*}
\resizebox{0.8\textwidth}{!}{$\displaystyle
Q^I_{i,j} = 
\begin{cases}
    I_0 \land \pi_0 & \text{if } i = j = 0 \\ 
    \bigwedge_{k=0}^{i-1}\left(\neg \pi_k\right) \land I_i \land \pi_j & \text{if } i = j, i\geq 0 \\
    \bigwedge_{k=0}^{i-1}\left(\neg \pi_k\right) \land I_i \land \neg \pi_i \land \neg C_j \land \neg I_j \land \pi_j & \text{if } j = i + 1, i \geq 0 \\
    \bigwedge_{k=0}^{i-1}\left(\neg \pi_k\right) \land I_i \land \neg \pi_i \bigwedge_{k=i+1}^{j-1}\left(\neg C_k \land \neg I_k \land \neg \pi_k\right)  & \text{if } j > i+1, i \geq 0 \\
    \quad \quad \land \neg C_j \land \neg I_j \land \pi_j. &
\end{cases}
$}
\end{equation*}
The investment channel at horizon $h\geq 0$ is then defined as
\begin{equation*}
    \text{Investment:} \quad \bigvee_{i=0}^h\bigvee_{j=i}^h Q^I_{i,j}.
\end{equation*}

Note that the paths captured by $Q^C_{i,j}$ are disjoint from the paths captured by $Q^C_{k,l}$ as long as either $i\neq k$ or $j \neq l$; a similar result holds for $Q^I_{i,j}$. Thus, both the consumption and investment channel definitions above are disjunctions of disjoint channels. The most efficient way to compute the transmission effects is therefore to compute the transmission effects for each $Q^C_{i,j}$, respectively $Q^I_{i,j}$, using the algorithms in Appendix~\ref{appendix:computation}, and sum the effects over $0\leq i \leq h$, $i \leq j \leq h$.

\section{Counterexample: IRF Sufficiency for SVARMAs}
\label{appendix:counterexample-irf-sufficiency}

Consider the bivariate SVARMA given by
\begin{equation}
    \by_t = \bC\by_{t-1} + \bM\bu_{t-1} + \bu_t, \;\; \bu_t = \bPhi_0\bvarepsilon_t,
    \label{eq:appendix-counterexample-varma}
\end{equation}
where $\bPhi_0$ is non-singular and $\bvarepsilon_t$ is a vector of structural shocks. We consider the family of coefficient matrices
\begin{equation*}
    \bPhi_0 = \begin{bmatrix}1 & 0 \\ 1 & 1\end{bmatrix}, \;\;\;
    \bC = \begin{bmatrix}0 & C_{12} \\ -0.5 & 0.5\end{bmatrix}, \;\;\;
    \bM = \begin{bmatrix}-0.5 & M_{12} \\ 0 & 0\end{bmatrix},
\end{equation*}
indexed by $C_{12}$ and subject to the single restriction $C_{12} + M_{12} = 0.5$, so that $\bC + \bM$ is constant across the family. Every member of this family is a legitimate SVARMA in the sense of Section 3: the autoregressive polynomial $\bI - \bC z$ has all roots outside the unit circle for $C_{12} \in (0, 0.5]$; the moving average polynomial satisfies $\det(\bI + \bM z) = 1 - 0.5z$ irrespective of $M_{12}$, so that the process is invertible; and $\bI - \bC z$ and $\bI + \bM z$ are left-coprime, so that each member is statistically identified. Two convenient choices are $(C_{12}, M_{12}) = (0.5, 0)$ and $(C_{12}, M_{12}) = (0.3, 0.2)$.

Writing~\eqref{eq:appendix-counterexample-varma} in Wold form gives $\by_t = \bu_t + (\bC+\bM)\bu_{t-1} + \dots$, so that the structural impulse responses on impact and at horizon one are
\begin{equation*}
    \bPhi_0 = \begin{bmatrix}1 & 0 \\ 1 & 1\end{bmatrix}, \;\;\;
    \bPhi_1 = (\bC+\bM)\bPhi_0 = \begin{bmatrix}0 & 0.5 \\ 0 & 0.5\end{bmatrix}.
\end{equation*}
Both are constant across the family, since they depend on $C_{12}$ and $M_{12}$ only through their sum.

Premultiplying~\eqref{eq:appendix-counterexample-varma} by $\bA_0 = \bPhi_0^{-1}$ yields the structural representation $\bA_0\by_{t} = \bA_1\by_{t-1} + \bPsi_1\bvarepsilon_{t-1} + \bvarepsilon_t$ with $\bA_1 = \bA_0\bC$ and $\bPsi_1 = \bA_0\bM\bPhi_0$, that is,
\begin{equation*}
    \resizebox{0.9\textwidth}{!}{$
    \bA_0 = \begin{bmatrix}1 & 0 \\ -1 & 1\end{bmatrix}, \;\;\;
    \bA_1 = \begin{bmatrix}0 & C_{12} \\ -0.5 & 0.5-C_{12}\end{bmatrix}, \;\;\;
    \bPsi_1 = \begin{bmatrix}-0.5 + M_{12} & M_{12} \\ 0.5 - M_{12} & -M_{12}\end{bmatrix}.
    $}
\end{equation*}

Since $\bA_0$ is lower-triangular with a positive diagonal, the QL-decomposition under the transmission matrix $\bT = \bI$ gives $\bQ = \bI$ and $\bL = \bA_0$, and since the diagonal of $\bA_0$ consists of ones, $\bD = \bI$. Defining $\bx = (\by_t', \by_{t+1}')'$ and $\bvarepsilon = (\bvarepsilon_t', \bvarepsilon_{t+1}')'$, the systems form up to horizon $h=1$ is given by
\begin{equation*}
    \bx =
    \begin{bmatrix}
    0 & 0 & 0 & 0 \\
    1 & 0 & 0 & 0 \\
    0 & C_{12} & 0 & 0 \\
    -0.5 & 0.5-C_{12} & 1 & 0
    \end{bmatrix}
    \bx +
    \begin{bmatrix}
        1 & 0 & 0 & 0 \\
        0 & 1 & 0 & 0 \\
        -0.5 + M_{12} & M_{12} & 1 & 0 \\
        0.5 - M_{12} & -M_{12} & 0 & 1
    \end{bmatrix}
    \bvarepsilon.
\end{equation*}
The only path from $x_1$ to $x_3$ runs through $x_2$, since the direct edge has path coefficient $\bB_{3,1} = 0$. The total path-specific effect of all paths connecting $x_1$ to $x_3$ is therefore
\begin{equation*}
    Q_1[\calP_{x_1, x_3}] = C_{12},
\end{equation*}
which does vary across the family.

It remains to verify that the impulse responses do not. Solving the systems form yields the block-triangular impulse response matrix
\begin{equation*}
    \bPhi = (\bI - \bB)^{-1}\bOmega = \begin{bmatrix}\bPhi_0 & \bO \\ \bPhi_1 & \bPhi_0\end{bmatrix} =
    \begin{bmatrix}
    1 & 0 & 0 & 0 \\
    1 & 1 & 0 & 0 \\
    0 & 0.5 & 1 & 0 \\
    0 & 0.5 & 1 & 1
    \end{bmatrix},
\end{equation*}
so that $\bPhi$ is fully determined by $\bPhi_0$ and $\bPhi_1$ and is therefore constant across the family. Moreover, $\bPhi_0$ is lower-triangular with a positive diagonal and hence is the Cholesky factor of the reduced-form covariance $\bPhi_0\bPhi_0'$, so that the Cholesky-orthogonalised shocks under the ordering $\bT = \bI$ coincide with the structural shocks and $\widetilde{\bPhi} = \bPhi$.

Two distinct members of the family therefore share the same $\bPhi$ and the same $\widetilde{\bPhi}$ up to horizon one, while their path-specific effects $Q_1[\calP_{x_1,x_3}]$ differ. No function of $\bPhi$ and $\widetilde{\bPhi}$ can return both values, and impulse responses up to the horizon of interest are consequently not sufficient statistics for the computation of transmission effects in general SVARMA models. The same conclusion holds for effects originating from a structural shock, and not only for effects between endogenous variables. Consider the channel from the structural shock $\varepsilon_1$ to $x_3$ that runs through $x_1$. Its transmission effect is $Q_1[\calP_{\varepsilon_1,x_3}(x_1)] = C_{12}$, which again varies across the family while $\bPhi$ and $\widetilde{\bPhi}$ do not.

\end{appendices}
\end{document}